\documentclass[12pt,reqno]{amsart}
\usepackage[a4paper, total={17cm,25cm}]{geometry}

\usepackage{color, graphics}
 \usepackage{tikz}
 \usetikzlibrary{arrows.meta}
 \usetikzlibrary{patterns}

\allowdisplaybreaks
 \usepackage{hyperref}
 \usepackage{pdfsync}
 \usepackage{longtable}
 \usepackage{diagbox}
 \usepackage{colortbl}
 \usepackage{enumitem}
 \usepackage{float}
 \usepackage{mathrsfs}
\newtheorem{theorem}{Theorem}[section]
\newtheorem{definition}[theorem]{Definition}
\newtheorem{proposition}[theorem]{Proposition}
\newtheorem{corollary}[theorem]{Corollary}
\newtheorem{lemma}[theorem]{Lemma}
\newtheorem{remark}[theorem]{Remark}

\newcommand\C{\mathbb{C} }

\newcommand\Psix{\textrm{P}_{\textrm{VI}}}
\makeatletter\let\c@figure\c@table\makeatother
\numberwithin{equation}{section}
\numberwithin{figure}{section}
\numberwithin{table}{section}
\newcommand{\orcidauthorA}{0000-0001-7504-4444}

\definecolor{grass}{rgb}{0.14,0.72,0.2}

\begin{document}
%---------------------------------------------------------- 

\title[Global asymptotics of $P_{VI}$]{Global asymptotics of the sixth Painlev\'e equation in Okamoto's space}

\author[V.~Heu]{Viktoria Heu}

\address{IRMA, UMR 7501, 7 rue Ren\'e-Descartes, 67084 Strasbourg Cedex, France}

\email{heu@math.unistra.fr}
\author[N.~Joshi]{Nalini Joshi}
\address{School of Mathematics and Statistics F07, The University of Sydney, Sydney, NSW 2006, Australia}
\thanks{N.J.'s ORCID ID is \orcidauthorA.}
\email{nalini.joshi@sydney.edu.au}

\author[M.~Radnovi\'c]{Milena Radnovi\'c}
\address{School of Mathematics and Statistics F07, The University of Sydney, Sydney, NSW 2006, Australia
\newline
Mathematical Institute SANU, Belgrade, Serbia (on leave)
}
\email{milena.radnovic@sydney.edu.au}


\thanks{The research of V.H.~is supported by the ANR grant ANR-16-CE40-0008, that of N.J.~ and M.R.~ is supported by the Discovery Grant \#DP200100210 from the Australian Research Council.
M.R.'s research was also partially supported by the
Mathematical Institute of the Serbian Academy of Sciences and Arts, the Science Fund of Serbia, grant Integrability and Extremal Problems in Mechanics, Geometry and
Combinatorics, MEGIC, Grant No.~7744592.
}

%\date{\colorbox{yellow}{\textcolor{red}{\today}}}

\begin{abstract}
We study dynamics of solutions in the initial value space of the sixth Painlev\'e equation as the independent variable approaches zero. 
Our main results describe the repeller set, show that the number of poles and zeroes of general solutions is unbounded, and that the complex limit set of each solution exists and is compact and connected. 
\end{abstract}

\maketitle
\tableofcontents
%\listoffigures

\section{Introduction}
In this paper, we consider the celebrated equation
\begin{equation}\label{PVI}
\begin{split}
 & y''=\frac{1}{2}\left(\dfrac{1}{y}+\dfrac{1}{y-1}+\dfrac{1}{y-x}\right)(y')^2 -\, \left(\dfrac{1}{x}+\dfrac{1}{x-1}+\dfrac{1}{y-x}\right)y'\\
&\qquad +\dfrac{y (y-1) (y-x)}{2 x^2 (x-1)^2 }\Bigl(\theta_\infty^2 -\,\dfrac{\theta_0^2 x}{y^2} + \dfrac{\theta_1^2(x-1)}{(y-1)^2} + \dfrac{(1-\theta_x^2)x (x-1)}{(y-x)^2}\Bigr),\\
\end{split}
\end{equation} 
for $x\in\mathbb C$, $(\theta_0, \theta_1, \theta_x, \theta_\infty)\in\mathbb C^4$, 
in its initial value space, where initial values are given at a point $x_0\in\mathbb C$, for small $|x_0|$.

The equation is the sixth Painlev\'e equation, first derived in \cite{Fuchs05} from deformations of a linear system with four regular singular points, a generalization of Gauss' hypergeometric equation; we will refer to it as $\Psix$. Subsequently, it was recognized as the most general equation in the study of second-order ODEs whose movable singularities are poles \cite{Gambier06,Painleve1897}. It has been studied widely because of its relation to mathematical physics and algebraic geometry; see \cite{hitchin95, manin96}. For special values of the parameters $(\theta_0, \theta_1, \theta_x, \theta_\infty)$, $\Psix$ has algebraic and elliptic solutions that turn out to be related to a broad range of mathematical structures; see \cite{DubrovinMazzocco2000, LisovyyTykhyy2014} and references therein. For generic parameters, the solutions are higher transcendental functions that cannot be expressed in terms of algebraic or classical functions \cite{Watanabe99}.

A large amount of work has been devoted to the description of these higher transcendental solutions. In this paper, we study global properties of such solutions of $\Psix$ in the limit as $x\to0$ in its initial value space (see Okamoto \cite{Okamoto}). Under appropriate M\"obius transformations of the variables \cite{OkamotoStudies}, our results also apply in the limit as $x$ approaches $1$ or $\infty$. Further information and properties of $\Psix$ are given in \S\ref{s:background} below.

Our starting point is the equivalent non-autonomous Hamiltonian system
\begin{subequations} \label{eq:yzsys}
\begin{align}
&y'=\phantom{-}\frac{\partial H}{\partial z},\label{eq:ypr}\\
& z'=-\frac{\partial H}{\partial y},\label{eq:zpr}
\end{align}
\end{subequations}
with Hamiltonian
\[
\begin{split} 
H=&\dfrac{y(y-1)(y-x)}{x(x-1)}\biggl( z^2-z\,\left(\frac{\theta_0}{y} +\frac{\theta_1}{y-1}+\frac{\theta_x-1}{y-x}\right) +\frac{\theta \overline{\theta}}{y(y-1)} \biggr).
\end{split}
\]
We will refer to the right side of Equations \eqref{eq:yzsys} as the Painlev\'e vector field and use the terminology 
$$\theta:=\frac{\theta_0+\theta_x+\theta_1+\theta_\infty-1}{2}, \quad \quad \overline{\theta}:=\theta-\theta_\infty .$$
To see  that the system \eqref{eq:yzsys} is equivalent to $\Psix$ (as shown by \cite{OkamotoHamil}), note that $z$ is given by Equation \eqref{eq:ypr} as
\[
2z:= \left( \frac{x-1}{y}-\frac{x}{y-1}+\frac{1}{y-x}\right) y' +\frac{\theta_0}{y}+\frac{\theta_1}{y-1}+\frac{\theta_x-1}{y-x}.
\]
Substituting this into Equation \eqref{eq:zpr} gives $\Psix$.

The Painlev\'e vector field becomes undefined at certain points in $\mathbb C^2$. 
{Those points correspond to the following initial values of the system \eqref{eq:yzsys}: $y=0$ or $y=1$ or $y=x$.}
Okamoto \cite{Okamoto} showed how to regularize the system at such points. For each $x_0\in \mathbb{C}\setminus\{0,1\}$, he compactified the  space of initial values $(y,z)\in (\mathbb{C}\setminus \{0,1,x_0\})\times \mathbb{C}$ to a smooth complex surface $S(x_0)$. The flow of the Painlev\'e vector field is well-defined in $\mathcal{S}:=\bigcup_{x_0\in \mathbb{C}\setminus\{0,1\}} \mathcal{S}(x_0)$, which we refer to as Okamoto's space of initial values.

Our main purpose is to describe the significant features of the flow in the singular limit $x\to0$. In similar studies of  the first, second, and fourth Painlev\'e equations \cite{Nalini1,Nalini2,NaliniMilena1} in singular limits, we showed that successive resolutions of the Painlev\'e vector field at base points terminates after nine blow-ups of $\mathbb C\mathbb P^2$, while for the fifth and third Painlev\'e equations we showed that the construction consists of eleven blow-ups and two blow-downs \cite{NaliniMilena2,NaliniMilena3}.
The initial value space in each case is then obtained by removing the infinity set, denoted $\mathcal{I}$, which are blow-ups of points not reached by any solution. 

Our main results fall into three parts:
\begin{enumerate}[leftmargin=0.8cm,label={\rm(\alph*)}]
\item {\em Existence of a repeller set:} Corollary \ref{cor:repeller} in Section \ref{s:estimates} shows that $\mathcal I$ is a repeller for the flow. Theorem \ref{th:estimates} provides the range of the independent variable for which a solution may remain in the vicinity of $\mathcal{I}$.
\item {\em Numbers of poles and zeroes:} In Corollary \ref{cor:repeller}, we prove that each solution that is sufficiently close to $\mathcal{I}$ has a pole in a neighbourhood of the corresponding value of the independent variable. Moreover, Theorem \ref{th:poles-zeroes-ones} shows that each solution with essential singularity at $x=0$ has infinitely many poles and infinitely many zeroes in each neighbourhood of that point.
\item {\em The complex limit set:} We prove in Theorem \ref{th:limit-set} that the limit set for each solution is non-empty, compact, connected, and invariant under the flow of the autonomous equation obtained as $x\to0$.
\end{enumerate}

\subsection{Background}\label{s:background}

$\Psix $ is the top equation in the well-known list of six Painlev\'e equations. Each of the remaining Painlev\'e equations can be obtained as a limiting form of $\Psix$.

To describe the complex analytic properties of their solutions, we recall that a normalized differential equation of the form $y''=\mathcal R(y',y,x)$ gives rise to two types of singularities, i.e., where the solution is not holomorphic.  A solution may have a \textit{fixed} singularity where $\mathcal R(\cdot, \cdot, x)$ fails to be holomorphic; in the case of $\Psix$, these lie at $x=0, 1, \infty$. 
The solutions may also have \textit{movable} singularities. 
{A movable singularity is a singularity whose location changes in a continuous fashion when going from one solution to a neighbouring solution under small changes in the initial conditions.   We note that this informal definition, which is somewhat difficult to make more precise, dates back to Fuchs \cite[p.699]{fuchs1884}.}

$\Psix$ was discovered by R. Fuchs in 1905 \cite{Fuchs05} in his study of deformations of a linear system of differential equations with four regular singularities, generalizing Gauss' hypergeometric equation. The latter has three regular singularities, placed at $0$, $1$, and $\infty$ by convention, and Fuchs took the fourth one to be at a location, which is deformable. The compatibility of the linear system with the deformation equation gives rise to $\Psix$. 

It is well known that $\Psix$ also has an elliptic form, which arises when we introduce an incomplete elliptic integral on a curve parametrized by $y(x)$.  $\Psix$ then becomes expressible in terms of the Picard-Fuchs equation for the corresponding elliptic curve. This form has been used for the investigation of its special solutions, which exist for special parameter values. This fact was rediscovered by Manin \cite{manin96} in his study of the mirror symmetries of the projective plane.

Given a Painlev\'e equation and $x$ not equal to a fixed singularity of the equation, Okamoto showed \cite{Okamoto} that the space of initial values forms a connected, compactified  and regularised space corresponding to a nine-point blow-up of the two-complex-dimensional projective space $\mathbb C\mathbb P^2$. For each given $x$, this is recognizable as an elliptic surface. These elliptic surfaces form fibres of a vector bundle as $x$ varies, with $\mathbb C$ as the base space. Starting with a point (initial value) on such a fibre, a solution of the Painlev\'e equation follows a trajectory that pierces each successive fibre, forming leaves of a foliated vector bundle \cite{milnor}. 

\subsection{Outline of the paper} The plan of the paper is as follows. In \S \ref{s:okamoto}, we construct the surface $\mathcal{S}(x_0)$. We define the notation and describe the results, with detailed calculations being provided in Appendix \ref{app:resolution}. In \S \ref{s:S0}, we describe the corresponding vector field for the limit $x\to0$. The movable singularities of $\Psix$ correspond to points $x_0$ where  the Painlev\'e vector field becomes unbounded. In  \S \ref{s:movable}, we consider neighbourhoods of exceptional lines where this occurs. Estimates of the Painlev\'e vector field as $x$ approaches $0$ are deduced in \S \ref{s:estimates}. In \S \ref{s:limit}, we consider the limit set. Finally, we give concluding remarks in \S \ref{s:con}.

\section{Resolution of singularities}\label{sec:okamoto}

In this section, we explain how to construct the space of initial values for the system \eqref{eq:yzsys}.
The notion of initial value spaces described in Definition \ref{def:initial-values-space} is based on foliation theory, and we start by first motivating the reason for this construction.
We then explain how to construct such a space by carrying out resolutions or blow-ups, based on the process described in Definition \ref{def:blow-up}.

The system \eqref{eq:yzsys} is a system of two first-order ordinary differential equations for $(y(x), z(x))$.
Given initial values $(y_0, z_0)$ at $x_0$, local existence and uniqueness theorems provide a solution that is defined on a local polydisk $U\times V$ in $\mathbb C\times \mathbb C^2$, where $x_0\in U\subset \mathbb C\setminus\{0,1\}$ and
$(y_0, z_0)\in V\subset (\mathbb C\setminus\{0\})\times\mathbb{C}$.
Our interest lies in global extensions of these local solutions.

However, the occurence of movable poles in the Painlev\'e transcendents acts as a barrier to the extension of $U\times V$ to the whole domain of \eqref{eq:yzsys}.
The first step to overcome this obstruction is to compactify the space $\mathbb C^2$, in order to include the poles.
We carry this out by embedding $\mathbb C^2$ into the {first} Hirzebruch surface $\mathbb{F}_1$ {\cite{hirzebruch, beauville}. $\mathbb{F}_1$ is a projective space covered by four affine coordinate charts (given in \S \ref{s:okamoto}).}  

The {next} step in this process results from the occurence of singularities in the Painlev\'e vector field \eqref{eq:yzsys} in $V$. By the term {\em singularity} we mean points where $(y', z')$ becomes either unbounded or undefined because at least one component approaches the undefined limit $0/0$. We are led therefore to construct a space in which the points where the singularities appear are regularised. The process of regularisation is called ``blowing up" or \emph{resolving a singularity}.

{The appearance of these singularities is related to the irreducibility of the solutions of Painlevé equations, originally due to Painlevé \cite{Painleve1897}, which we have restated below in modern terminology. A function is said to be reducible to another function if it is related to it through a series of allowable operations (described by Painlev\'e and itemized as (O), (P1)--(P5) in \cite[p.33]{Umemura90}).}
\begin{theorem}
	If the space of initial values for a differential equation is a compact rational variety, then the equation can be reduced either to a linear differential equation of higher order or to an equation {governing} elliptic functions.
\end{theorem}
{Modern proofs of the irreducibility of the Painlev\'e equations have been developed by many authors, including Malgrange \cite{Malgrange01}, Umemura \cite{Umemura96,Umemura90} and Watanabe \cite{Watanabe99}.} Since the Hirzebruch surface is a compact rational variety, the {above} theorem implies that it cannot be the space of initial values for \eqref{eq:yzsys}.

We are now in a position to define the notion of initial value space.

\begin{definition}[\cite{Gerard1975}, \cite{GerardSec1972,Gerard1983,Okamoto}]\label{def:initial-values-space} 
	Let $(\mathcal{E},\pi,\mathcal{B})$ be a complex analytic fibration, $\Phi$ a foliation of $\mathcal{E}$, and $\Delta$ a holomorphic differential system on $\mathcal{E}$, such that:
	\begin{itemize}
		\item the leaves of $\Phi$ correspond to the solutions of $\Delta$;
		\item the leaves of $\Phi$ are transversal to the fibres of $(\mathcal{E},\pi,\mathcal{B})$;
		\item for each path $p$ in the base $\mathcal{B}$ and each point $X\in \mathcal{E}$, such that $\pi(X)\in p$, the path $p$ can be lifted into the leaf of $\Phi$ containing point $X$.
	\end{itemize}
	Then each fibre of the fibration is called \emph{a space of initial values} for the system $\Delta$.
\end{definition}

The properties listed in Definition \ref{def:initial-values-space} imply that each leaf of the foliation is isomorphic to the base $\mathcal{B}$.
Since the transcendental solutions of the sixth Painlev\'e equation can be globally extended as meromorphic functions of $x\in\mathbb{C}\setminus\{0,1\}$, we search for the fibration with the base equal to $\mathbb{C}\setminus\{0,1\}$. 

In order to construct the fibration, we apply the blow-up procedure, defined below \cite{HartshorneAG,GrifHarPRINC,DuistermaatBOOK} to the singularities of the system  \eqref{eq:yzsys} that occur where at least one component becomes undefined of the form $0/0$.
Okamoto \cite{Okamoto} showed that such singular points are contained in the closure of infinitely many leaves.
Moreover, these leaves are holomorphically extended at such a point.

\begin{definition}\label{def:blow-up}
	\emph{The blow-up} of the plane $\mathbb{C}^2$ at point $(0,0)$ is the closed subset $X$ of $\mathbb{C}^2\times\mathbb{CP}^1$ defined by the equation $u_1t_2=u_2t_1$, where $(u_1,u_2)\in\mathbb{C}^2$ and $[t_1:t_2]\in\mathbb{CP}^1$, see Figure \ref{fig:blow-up}.
	There is a natural morphism $\varphi: X\to\mathbb{C}^2$, which is the restriction of the projection from $\mathbb{C}^2\times\mathbb{CP}^1$ to the first factor.
	$\varphi^{-1}(0,0)$ is the projective line $\{(0,0)\}\times\mathbb{CP}^1$, called \emph{the exceptional line}.
\end{definition}
\begin{figure}[h]
	\includegraphics[width=11.2cm, height=6.07cm]{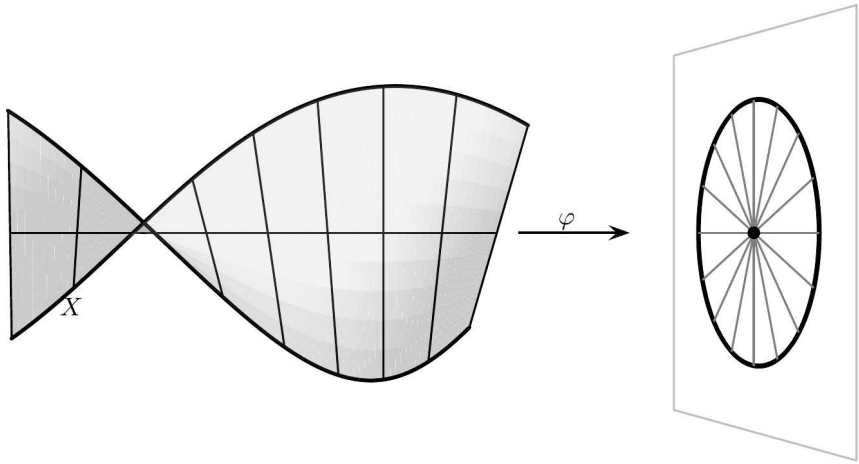}
	\caption{The blow-up of the plane at a point.}\label{fig:blow-up}
\end{figure}

\begin{remark}
	Notice that the points of the exceptional line $\varphi^{-1}(0,0)$ are in bijective correspondence with the lines containing $(0,0)$.
	On the other hand,
	$\varphi$ is an isomorphism between $X\setminus\varphi^{-1}(0,0)$ and $\mathbb{C}^2\setminus\{(0,0)\}$.
	More generally, any complex two-dimensional surface can be blown up at a point \cite{HartshorneAG,GrifHarPRINC,DuistermaatBOOK}.
	In a local chart around that point, the construction will look the same as described for the case of the plane.
\end{remark}

Notice that the blow-up construction separates the lines containing the point $(0,0)$ in Definition \ref{def:blow-up}, as shown in Figure \ref{fig:blow-up}.
In this way, the solutions of \eqref{eq:yzsys} containing the same point can be separated.
Additional blow-ups may be required if the solutions have a common tangent line or a tangency of higher order at such a point.
The explicit resolution of the vector field \eqref{eq:yzsys} is carried out in Appendix \ref{app:resolution}.

Okamoto described so called \emph{singular points of the first class} that are not contained in the closure of any leaf of the foliation given by the system of differential equations.
At such points, the corresponding vector field is infinite.

\section{The construction of Okamoto's space}\label{s:okamoto}
In this section, we construct Okamoto's space of initial values, in such a way as to ensure that the process yields a well-defined compact surface if we set $x_0=0$. We start by defining a new time coordinate $t=\ln x$, or  $x=\exp (t)$, suitable for taking the limit $x\to0$, and rewrite the dependent variables as 
\[
u(t)=y(x),\ v(t)=z(x).
\]
For conciseness, we continue to use the notation $x=e^t$ where needed.

Denoting $t$-derivatives by dots, we get  
$\dot{u}=x\frac{\partial H}{\partial v},\ \dot{v}=-x\frac{\partial H}{\partial u},$
or, equivalently
\begin{equation}\label{eq:vectorfield}
\dot{u}=\dfrac{\partial E}{\partial v},\ \dot{v}=-\dfrac{\partial E}{\partial u},
\end{equation}
where
\[
\begin{split} 
E=&\frac{u(u-1)(u-e^t)}{e^t-1}\Bigl\{ v^2-v\left(\frac{\theta_0}{u} +\frac{\theta_1}{u-1}+\frac{\theta_x-1}{u-e^t}\right)  +{\frac{\theta \overline{\theta}}{u(u-1)}} \Bigr\}.
\end{split}
\]

Suppose we are given $x_0=e^{t_0}\in \C\setminus \{0,1\}$. We compactify the space of initial values $ (u(t_0),v(t_0))\in \C^2$ to the first Hirzebruch surface {$\mathbb{F}_1$ \cite{hirzebruch}}, which is covered by four affine charts in $\mathbb C^2$ \cite{beauville}
\begin{align*}
&(u_0,v_0)=(u,v),&(u_1,v_1)&=\left(u,\frac{1}{v}\right),\\
&(u_2,v_2)=\left(\frac{1}{u},\frac{1}{uv}\right), &(u_3,v_3)&=\left(\frac{1}{u},uv\right).
\end{align*}

{ Let $L$ be the unique section of the natural projection $\mathbb{F}_1\to \mathbb{P}^1$ defined by $(u,v)\mapsto u$. Then, $L$ is given by $\{v_0=0\}\cup \{v_3=0\}$ and the self-intersection of its divisor class is $-1$.}
We identify four particular fibers of this projection: $$\mathcal{V}_j:=\{u_0=j\}\cup\{u_1=j\}\quad \forall j\in \{0,x,1\}\,, \quad \mathcal{D}_\infty :=\{u_2=0\}\cup\{u_3=0\}.$$
 Note that as fibers of the projection, these lines all have self-intersection $0$.
Then $\mathbb{F}_1\setminus\mathbb{C}^2$ is given by $\mathcal{D}_\infty \cup \mathcal{H}$, where 
$$\mathcal H := \{v_1=0\}\cup\{v_2=0\}.$$
This section $\mathcal H$, called a ``horizontal line'' in the following, by a small abuse of common terminology, is topologically equivalent to the formal sum $L+\mathcal{D}_\infty$ in $\mathrm{H}_2(\mathbb{F}_1,\mathbb{Z})$. In particular, its self-intersection number is given by 
$ \mathcal H\cdot \mathcal H=L\cdot L+\mathcal{D}_\infty\cdot \mathcal{D}_\infty+2L\cdot \mathcal{D}_\infty=-1+0+2=+1$,
{where the dot $\cdot$ denots the intersection form of divisor classes in the Picard group of the surface.}

In each chart, the vector field respectively becomes
\begin{equation*}
  \begin{cases}
    \dot{u} &= \dfrac{u (u -1)(u -x )}{x-1}\left(2v -\dfrac{\theta_0}{u }-\dfrac{\theta_1}{u -1}-\dfrac{\theta_x-1}{u -x }\right),\\
    \dot{v} &= -\dfrac{3u^2-2(x+1)u +x}{x-1}v^2+2\dfrac{\theta+\overline{\theta} }{x-1}u v -\left( \dfrac{x\theta_0}{ x-1}+  \theta_1 +\dfrac{\theta+\overline{\theta}}{x-1} \right)v \\
    &\qquad\quad -\dfrac{\theta \overline{\theta} }{x-1}\, , 
  \end{cases}
  \end{equation*}
 \begin{align*}
  &\begin{cases}
  \dot{u}_1 &= \dfrac{u_1 (u_1 -1)(u_1 -x)}{x-1}\left(2\dfrac{1}{v_1} -\dfrac{\theta_0}{u_1 }-\dfrac{\theta_1}{u_1 -1}-\dfrac{\theta_x -1}{u_1 -x}\right),\\
\dot{v}_1 &= \dfrac{3u_1^2-2(x+1)u_1 +x}{x-1}-2\dfrac{\theta+\overline{\theta} }{x-1}u_1v_1  +\left( \dfrac{x\theta_0}{ x-1}+  \theta_1+\dfrac{\theta+\overline{\theta}}{x-1} \right)v_1 \\
    &\qquad\quad+\dfrac{  \theta\overline{\theta}}{x-1}v_1^2 , 
\end{cases}\\
 & \begin{cases}
 \dot u_2 &= 2\dfrac{(u_2-1)(xu_2-1) }{(1-x)v_2}+\dfrac{(u_2-1)(\theta+\overline{\theta}-x\theta_0u_2) }{1-x}- \theta_1 u_2 , \\
\dot v_2 &= -\dfrac{(\theta v_2-1)(\overline{\theta}v_2-1)}{(1-x)u_2}-\dfrac{x}{1-x}(\theta_0v_2-1)u_2 ,
 \end{cases}
  \end{align*}
 \begin{equation*}
  \begin{cases}
 \dot u_3 &= -\dfrac{(u_3-1)(2v_3-(\theta+\overline{\theta}))}{1-x}- \theta_1 u_3+\dfrac{x}{1-x}u_3(u_3-1)(2v_3-\theta_0),\\
\dot v_3 &=\dfrac{(v_3-\theta)(v_3-\overline{\theta})}{(1-x)u_3}-x\dfrac{(v_3-\theta_0)u_3v_3}{1-x}.
 \end{cases}
  \end{equation*}

One realizes that the vector field is infinite on $\mathcal{H}:\{v_1=0\}\cup   \{ v_2=0\}$. More precisely, it is infinite or undetermined %(of the form ``$\frac{0}{0}$'') 
precisely there. We use the term {\em base point} for points where the vector field becomes undetermined. For example, the point $(u_1,v_1)=(0,0)$ in the coordinate chart $\C^2_{u_1,v_1}$ is a base point because the equation for $\dot{u}_1$ approaches $0/0$ as $(u_1, v_1)\to(0,0)$. 
In total, we find the following five base points in $\mathbb{F}_1$, possibly visible in several charts.
This initial situation is summarized in Table \ref{table:base-points} and Figure \ref{fig:F2Init}.
{Where needed in figures, we indicate the self-intersection number $n$ of an exceptional divisor by annotating it by $(n)$.  
}

\begin{longtable}{|c|c|c|c|}
 \hline
  \diagbox{Points}{Charts}&$( u_1, v_1)$ & $( u_2, v_2)$ & $( u_3, v_3)$ \\
  \hline
$ \beta_0$ & $\left(0, 0\right)$ & &   \\
 \hline 
$\beta_x$    & $\left(x,0\right)$ &    $\left( {1}/{x} , 0\right)$ &\\
\hline 
$ \beta_1$  &  $\left(1,0\right)$ &    $ \left(1,0\right)$ & \\
\hline
$\beta_\infty$  & &  $\left(0,  {1}/{\theta} \right)$ &  $\left(0, \theta \right)$\\
 \hline 
$\beta_\infty^-$  & &  $\left(0,  {1}/{\overline{\theta}} \right)$   &  $\left(0, \overline{\theta} \right)$
\\ \hline
  \caption{Five base points and the charts in which they are visible. The chart $( u_0, v_0)$ is omitted because no base points are visible in this chart.}\label{table:base-points}
\end{longtable}

 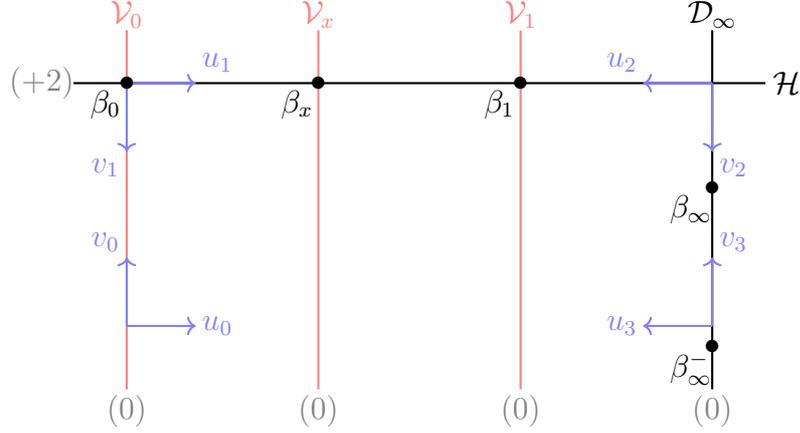
\begin{figure}[h]
\begin{center}
\begin{tikzpicture}[scale=0.7]
\draw[thick] (-6.5,4.3)--(6.5,4.3);
  \draw (6.9,4.3) node {$ {\mathcal H}$} ; \draw[gray] (-7.1,4.3) node {$(+2)$} ;
\draw[ thick] (5.5,-1.5)--(5.5,5.3);
  \draw (5.5,5.6) node {$\mathcal{D}_\infty$} ;\draw[gray] (5.5,-1.9) node {$(0)$} ;

  \draw[red!50,thick] (1.9,-1.5)--(1.9,5.3); 
  \draw[red!50] (1.9,5.6) node {$\mathcal{V}_1$} ;; \draw[gray] (1.9,-1.9) node {$(0)$} ;
    \draw[red!50,thick] (-1.9,-1.5)--(-1.9,5.3);
  \draw[red!50] (-1.9,5.6) node {$\mathcal{V}_x$} ;\draw[gray] (-1.9,-1.9) node {$(0)$} ;
      \draw[red!50,thick] (-5.5,-1.5)--(-5.5,5.3);
  \draw[red!50] (-5.5,5.6) node {$\mathcal{V}_0$} ; \draw[gray] (-5.5,-1.9) node {$(0)$} ;
  
  \draw[blue!50, thick, ->] (-5.5,-0.3)--(-5.5,1);%   (-5.5,-3.3)--(-5.5,-2)
            \draw[blue!50, thick, ->] (-5.5,-.3)--(-4.2,-.3);
               \draw[blue!50] (-3.8 , -.3)  node  {$u_0$} ;
         \draw[blue!50] (-5.9 , 1.3)  node  {$v_0$} ;%(-5.9 , -1.7)  
         
          \draw[blue!50, thick, ->] (-5.5,4.3)--(-5.5,3);
            \draw[blue!50, thick, ->] (-5.5,4.3)--(-4.2,4.3);
               \draw[blue!50] (-3.8 , 4.7)  node  {$u_1$} ;
         \draw[blue!50] (-5.9 , 2.7)  node  {$v_1$} ;
          
          \draw[blue!50, thick, ->] (5.5,4.3)--(5.5,3);
            \draw[blue!50, thick, ->] (5.5,4.3)--(4.2,4.3);
               \draw[blue!50] (3.8 , 4.7)  node  {$u_2$} ;
         \draw[blue!50] (5.9 , 2.7)  node  {$v_2$} ;
          
          \draw[blue!50, thick, ->] (5.5,-.3)--(5.5,1);
            \draw[blue!50, thick, ->] (5.5,-.3)--(4.2,-.3);
               \draw[blue!50] (3.8 , -.3)  node  {$u_3$} ;
         \draw[blue!50] (5.9 , 1.3)  node  {$v_3$} ;
         
           \draw (-5.5,4.3) node {$\bullet$} ;   \draw (-5.9,3.9) node {$\beta_0$} ;
     \draw (-1.9,4.3) node {$\bullet$} ;    \draw (-2.3,3.9) node {$\beta_x$} ;
    \draw (1.9,4.3) node {$\bullet$} ; \draw (1.5,3.9) node {$\beta_1$} ;
     \draw (5.5,2.3) node {$\bullet$} ;    \draw (5.1,1.9) node {$\beta_\infty$} ;   
       \draw (5.5,-.7) node {$\bullet$} ;    \draw (5.1,-1.1) node {$\beta_\infty^-$} ;   

         \end{tikzpicture}
         \caption[The first Hirzebruch surface]{The surface $\mathbb{F}_1$ with its coordinates and the base point configuration. The numbers in parentheses indicate self-intersection numbers.}\label{fig:F2Init}
\end{center}
  \end{figure}
 Okamoto's procedure consists in resolving the vector field by successively blowing up the base points until the vector field becomes determined. Since later on we need a well-defined compact surface if we set $x=0$, we may not blow up $\beta_0$ and $\beta_x$ simultaneously. As detailed in Appendix \ref{s:DetailsA},
 the blow-up of $\beta_0,\beta_x,\beta_1$ with  $\beta_x$ after $\beta_0$ consists of
replacing the charts $\C^2_{u_1,v_1}$ and $\C^2_{u_2,v_2}$ by the following five $\C^2$-charts, endowed with the obvious rational transition maps,
$$\begin{array}{rclcrcl}(\tilde u_{1},\tilde v_{1})&:=&\left(u_1, \frac{v_1}{u_1(u_1-1)(u_1-x)} \right) &\quad 
&(\tilde u_{2},\tilde v_{2})&:=&\left(u_2, \frac{v_2}{(1-u_2)(1-xu_2)} \right) 
\vspace{.2cm}\\
(u_{02},v_{02})&:=&\left(\frac{u_1}{v_1}, {v_1}\right)&\quad &(u_{12},v_{12})&:=&\left(\frac{u_1-1}{v_1}, {v_1}\right)\vspace{.2cm}\\
 (u_{x2},v_{x2})&:=&\left(\frac{u_1(u_1-x)}{v_1}, \frac{v_1}{u_1}\right) 
\end{array}$$
 For each   $i\in \{0,1,x\}$, what formerly was the point $\beta_i$ is now replaced by an exceptional line 
$$\mathcal{D}_i:\{\tilde{u}_{1}=i\} \cup \{v_{i2}=0\}$$ of self-intersection $-1$.  The strict transform of  $\mathcal{H}$, i.e. the closure of $\mathcal{H}\setminus \{\beta_0,\beta_x,\beta_1\}$ after blow-up is given by 
$$\mathcal{H}^*:=  \{\tilde{v}_{1}=0\}\cup  \{\tilde{v}_{2}=0\} \, .$$ As a general fact, each time we blow up a point on a curve, the self-intersection number of the strict transformation of the curve is the former self-intersection number decreased by unity. Since here we have blown up three points, $\mathcal{H}^*$ has self-intersection number $(-2)$. 
The blow-up of $\beta_\infty$ consists of removing the point $(0,1/\theta)$  (corresponding to $\beta_\infty$) from the chart $\C^2_{\tilde{u}_2,\tilde{v}_2}$
and replacing the chart $\C^2_{{u}_3, {v}_3}$ by the following pair of $\C^2$-charts.
$$\begin{array}{rclcrcl}
(\tilde u_3,  \tilde v_3)&:=&\left(u_3,- \frac{v_3-\theta}{u_3}\right)  &\quad 
&(  u_{\infty 2},  v_{\infty 2})&:=&\left(\frac{u_3}{v_3-\theta}, v_3-\theta\right)  
 \end{array}$$
 Again we obtain an exceptional line $\mathcal{E}_\infty$ and a strict transform $ \mathcal{D}_\infty^*$ such that $\mathcal{D}_\infty = \mathcal{E}_\infty\cup \mathcal{D}_\infty^*$, where  $$\mathcal{E}_\infty:=\{ \tilde u_3 =0\}\cup \{ v_{\infty 2}=0\}\, , \quad \mathcal{D}_\infty^*=\{\tilde{u}_2=0\}\cup \{ u_{\infty 2}=0\}\, .$$
 In each of the seven new charts that we have to add to $\C^2_{u,v}$ in order to fully describe the surface resulting of $\mathbb{F}_1$ after this first sequence of blow-ups, we again look at the resulting vector field (see section \ref{s:DetailsB}) and find the following base points, including the still unresolved $\beta_\infty^-$.
 The situation is summarized in Table \ref{t:add-dp}.
 
 \begin{table}
 \begin{tabular}{|c|c|c|c|c|}
   \hline
\diagbox{Charts}{Points}&$\gamma_0$&$\gamma_x$&$\gamma_1$&$\beta_\infty^-$\\
\hline
$( \tilde{u}_1, \tilde{v}_1)$
&$\left(0 , \frac{1}{x\theta_0} \right)$
&$\left(x , \frac{1}{x(x-1)\theta_x} \right)$
&$\left(1 , \frac{1}{(1-x)\theta_1} \right)$
&\\
\hline
$( \tilde{u}_2, \tilde{v}_2)$
&&$\left(\frac{1}{x} , \frac{x}{(x-1)\theta_x} \right)$
&$\left(1 , \frac{1}{(1-x)\theta_1} \right)$
&$\left(0 , \frac{1}{\overline{\theta}} \right)$\\
\hline
$( u_{02}, v_{02})$
&$\left(\theta_0 ,0 \right)$&&&\\
\hline
$( u_{x2}, v_{x2})$
&$\left(-x\theta_0 , \frac{1}{\theta_0} \right)$
&$\left(x\theta_x,0\right)$
&&\\
\hline
$( u_{12}, v_{12})$
&&&$ \left(\theta_1 ,0 \right)$&\\
\hline
$( u_{\infty 2}, v_{\infty 2}) $
&&&&$\left(0 , -\theta_\infty \right)$\\
\hline
\end{tabular}
   \caption{Base points remaining after blowing up $\beta_0,\beta_x,\beta_1$ and $\beta_\infty$. The chart $( \tilde{u}_3, \tilde{v}_3)$ is ommitted as there is no base point remaining in this chart.}\label{t:add-dp}
\end{table}

In Figure \ref{fig:F2BlowUp}, the notation ``$(n)$'' again indicates ``self-intersection number equal to $n$''. Moreover, as a visual guideline, we again included the 
strict transforms $\mathcal{V}_i^*:= \{u_{i2}=0\} \cup \{u_0=i\} $ of the former vertical lines $\mathcal{V}_i.$ Those  have self-intersection $(-1)$.
\begin{figure}[H] \begin{center}
\begin{tikzpicture}[scale=0.7]
\draw[thick] (-6.5,4.3)--(6.5,4.3);
  \draw (6.9,4.3) node {$ {\mathcal H}^*$} ;    \draw[gray] (-7.1,4.3) node {$(-2)$} ;

            \draw[red!50, thick] (-8,-3.3)--(-3.2,-3.3);            \draw[red!50] (-8.4,-3.3)  node  {$\mathcal{V}_0^*$} ; 
            
                 \draw[red!50, thick] (-4,-1)--(0.8,-1);            \draw[red!50] (-4.4,-1)  node  {$\mathcal{V}_x^*$} ;
                 
                        \draw[red!50, thick] (-0,1.5)--(4.8,1.5);            \draw[red!50] (-0.4,1.5)  node  {$\mathcal{V}_1^*$} ;

 \draw[red!50, thick] (3.5,-1.5)--(8.3,-1.5);            \draw[red!50] (8.9,-1.5)  node  {$\mathcal{E}_{\infty}$} ;

\draw[thick] (5.5,-4.5)--(5.5,5.3);
  \draw (5.5,5.6) node {$\mathcal{D}_\infty^*$} ;\draw[gray] (5.5,-4.9) node {$(-1)$} ;
  \draw[thick] (1.9,-4.5)--(1.9,5.3);
  \draw  (1.9,5.6) node {${\mathcal{D}}_1$} ;\draw[gray] (1.9,-4.9) node {$(-1)$} ;
    \draw[thick] (-1.9,-4.5)--(-1.9,5.3);
  \draw  (-1.9,5.6) node {${\mathcal{D}}_x$} ;\draw[gray] (-1.9,-4.9) node {$(-1)$} ;
      \draw[thick] (-5.5,-4.5)--(-5.5,5.3);
  \draw  (-5.5,5.6) node {${\mathcal{D}}_0$} ;\draw[gray] (-5.5,-4.9) node {$(-1)$} ;

       \draw[blue!50, thick, ->] (-5.5,-3.3)--(-5.5,-2);
            \draw[blue!50, thick, ->] (-5.5,-3.3)--(-4.2,-3.3);
               \draw[blue!50] (-3.8 , -3.6)  node  {$v_{02}$} ;
         \draw[blue!50] (-5.9 , -1.7)  node  {$u_{02}$} ;
         
          \draw[blue!50, thick, ->] (-5.5,4.3)--(-5.5,3);
            \draw[blue!50, thick, ->] (-5.5,4.3)--(-4.2,4.3);
               \draw[blue!50] (-3.8 , 4.7)  node  {$\tilde u_1$} ;
         \draw[blue!50] (-5.9 , 2.7)  node  {$\tilde v_1$} ;

          \draw[blue!50, thick, ->] (5.5,4.3)--(5.5,3);
            \draw[blue!50, thick, ->] (5.5,4.3)--(4.2,4.3);
               \draw[blue!50] (3.8 , 4.7)  node  {$\tilde u_2$} ;
         \draw[blue!50] (5.9 , 2.7)  node  {$\tilde v_2$} ;

          \draw[blue!50, thick, ->] (5.5,-1.5)--(5.5,-0.2);
            \draw[blue!50, thick, ->] (5.5,-1.5)--(6.8,-1.5);
               \draw[blue!50] (7.4 , -1.8)  node  {${u}_{\infty 2} $} ;
         \draw[blue!50] (6.1 , -0.1)  node  {${v}_{\infty 2}$} ;
         
            \draw[blue!50, thick, ->] (-1.9,-1 )--(-1.9,0.3);
            \draw[blue!50, thick, ->] (-1.9,-1 )--(-0.6,-1);
               \draw[blue!50] (-0.1 , -1.2)  node  {$v_{x 2}$} ;
         \draw[blue!50] (-1.5 , 0.6)  node  {$u_{x 2}$} ;
         
               \draw[blue!50, thick, ->] (1.9,1.5 )--(1.9,2.8);
            \draw[blue!50, thick, ->] (1.9,1.5 )--(3.2,1.5);
               \draw[blue!50] (3.5 , 1.8)  node  {$v_{1 2}$} ;
         \draw[blue!50] (2.4 , 3.1)  node  {$u_{1 2}$} ;
         
      \draw (-5.5,1.3) node {$\bullet$} ;   \draw (-5.9,0.9) node {$\gamma_0$} ;
          \draw (-1.9,3.3) node {$\bullet$} ;    \draw (-2.3,2.9) node {$\gamma_x$} ;
           \draw (1.9,-2) node {$\bullet$} ; \draw (1.5,-2.4) node {$\gamma_1$} ;
      \draw (5.5,-3.5) node {$\bullet$} ;  \draw (5.0,-3.5) node {$\beta_\infty^-$} ;

         \end{tikzpicture}
   \caption[The surface after a first sequence of blow-ups]{The surface $\mathbb{F}_1$ after the first sequence of blow-ups and the new base point configuration. }\label{fig:F2BlowUp}
\end{center}
\end{figure}
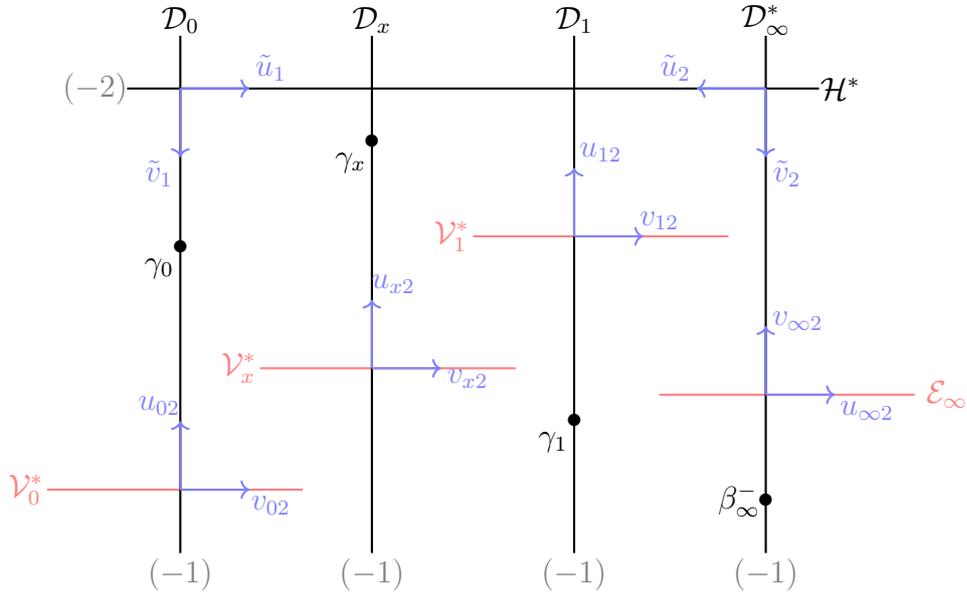

We blow-up the remaining base points by replacing each chart $\C^2_{u_{i2},v_{i2}}$ by a pair of $\C^2$-charts with corresponding index as follows, and then removing the already blown up base points that are still visible from other charts.
$$\begin{array}{rclcrcl}(u_{03},v_{03})&:=& \left(u_{02}-\theta_0,\frac{v_{02}}{u_{02}-\theta_0} \right) &\quad 
&(u_{04},v_{04})&:=& \left(\frac{u_{02}-\theta_0}{v_{02}}, v_{02}\right)
\vspace{.2cm}\\
(u_{x3},v_{x3})&:=& \left(u_{x2}-x(\theta_x-1),\frac{v_{x2}}{u_{x2}-x(\theta_x-1)} \right) &\quad 
&(u_{x4},v_{x4})&:=& \left(\frac{u_{x2}-x(\theta_x-1)}{v_{x2}}, v_{x2}\right)
\vspace{.2cm}\\
(u_{13},v_{13})&:=& \left(u_{12}-\theta_1,\frac{v_{12}}{u_{12}-\theta_1} \right) &\quad 
&(u_{14},v_{14})&:=& \left(\frac{u_{12}-\theta_1}{v_{12}}, v_{12}\right)
\vspace{.2cm}\\
(u_{\infty 3},v_{\infty 3})&:=& \left(u_{ \infty 2} ,\frac{v_{\infty 2}+\theta_\infty}{u_{\infty 2}} \right) &\quad 
&(u_{\infty 4},v_{\infty 4})&:=& \left(\frac{u_{\infty 2} }{v_{\infty 2}+\theta_\infty}, v_{\infty 2}\right)
\end{array}$$
We obtain the following new exceptional lines, for $i\in \{0,1,x\}$.
$$\mathcal{E}_i:=\{u_{i3}=0\}\cup \{v_{i4}=0\} \, , \quad \mathcal{E}_\infty^-:=\{u_{\infty 3}=0\}\cup \{v_{\infty 4}=0\}.$$
Moreover, we have the following new strict transforms, for $i\in \{0,1,x\}$.
$$\mathcal{D}_i^*:=\{v_{i3}=0\}\cup \{\tilde{u}_1=i\} \, , \quad \mathcal{D}_\infty^{**}:=\{\tilde u_2=0 \}\cup \{u_{\infty 4}=0\}.$$
The above charts of the Hirzebruch surface blown up in our eight base points are detailed in appendix section \ref{s:DetailsC}. As we can see from the equations there, the vector field  is now free of base points.  We say that the initial value space is resolved or regularized. Moreover, the function $E$ is well-defined there, i.e. when resolving the base points of the vector field, we also resolved the indeterminacy points of $E$. For each of the new coordinate charts $(u_{mn}, v_{mn})$, we also define the Jacobian 
\[ \omega_{mn} =\frac{\partial u_{mn}}{\partial u}\, \frac{\partial v_{mn}}{\partial v}\,-\,\frac{\partial u_{mn}}{\partial v}\,\frac{\partial v_{mn}}{\partial u}.\]

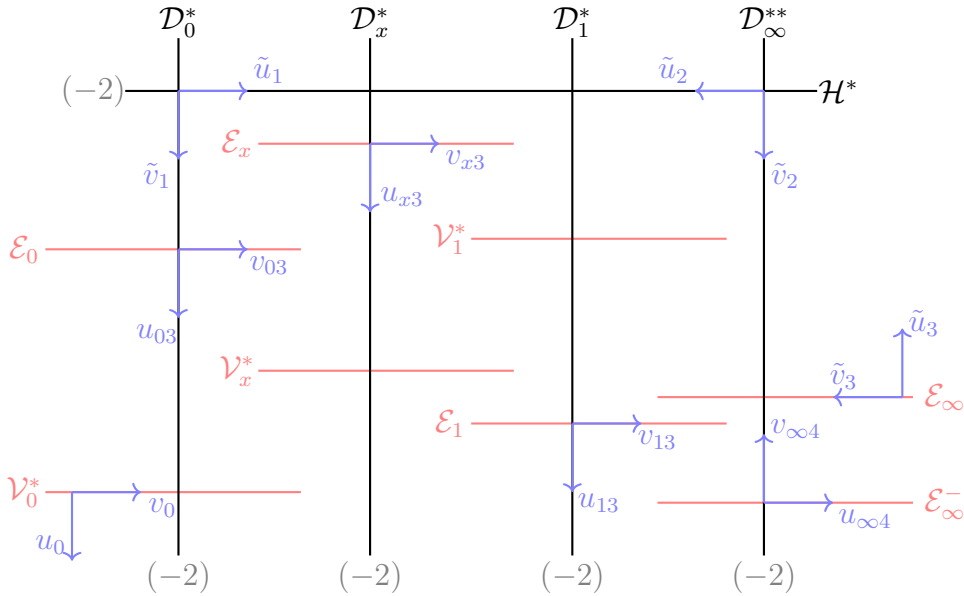
\begin{figure}[H]
\begin{center}
\begin{tikzpicture}[scale=0.7]
\draw[thick] (-6.5,4.3)--(6.5,4.3);
  \draw (6.9,4.3) node {$ {\mathcal H}^*$} ;    \draw[gray] (-7.1,4.3) node {$(-2)$} ;

            \draw[red!50, thick] (-8,-3.3)--(-3.2,-3.3);            \draw[red!50] (-8.4,-3.3)  node  {$\mathcal{V}_0^*$} ; 
            \draw[red!50, thick] (-8,1.3)--(-3.2,1.3);            \draw[red!50] (-8.4,1.3)  node  {$\mathcal{E}_0$} ;

  \draw[blue!50, thick, ->] (-7.5,-3.3)--(-7.5,-4.6);
            \draw[blue!50, thick, ->] (-7.5,-3.3)--(-6.2,-3.3);
               \draw[blue!50] (-5.8 , -3.6)  node  {$v_{0}$} ;
         \draw[blue!50] (-7.9 , -4.3)  node  {$  u_{0}$} ;

                 \draw[red!50, thick] (-4,-1)--(0.8,-1);            \draw[red!50] (-4.4,-1)  node  {$\mathcal{V}_x^*$} ;
                  \draw[red!50, thick] (-4,3.3)--(0.8,3.3);            \draw[red!50] (-4.4,3.3)  node  {$\mathcal{E}_x$} ;

                        \draw[red!50, thick] (-0,1.5)--(4.8,1.5);            \draw[red!50] (-0.4,1.5)  node  {$\mathcal{V}_1^*$} ;
           \draw[red!50, thick] (-0,-2)--(4.8,-2);            \draw[red!50] (-0.4,-2)  node  {$\mathcal{E}_1$} ;

 \draw[red!50, thick] (3.5,-1.5)--(8.3,-1.5);            \draw[red!50] (8.9,-1.5)  node  {$\mathcal{E}_{\infty}$} ;
   \draw[red!50, thick] (3.5,-3.5)--(8.3,-3.5);            \draw[red!50] (8.9,-3.5)  node  {$\mathcal{E}_{\infty}^-$} ;

\draw[thick] (5.5,-4.5)--(5.5,5.3);
  \draw (5.5,5.6) node {$\mathcal{D}_\infty^{**}$} ;\draw[gray] (5.5,-4.9) node {$(-2)$} ;
  \draw[thick] (1.9,-4.5)--(1.9,5.3);
  \draw  (1.9,5.6) node {${\mathcal{D}}_1^*$} ;\draw[gray] (1.9,-4.9) node {$(-2)$} ;
    \draw[thick] (-1.9,-4.5)--(-1.9,5.3);
  \draw  (-1.9,5.6) node {${\mathcal{D}}_x^*$} ;\draw[gray] (-1.9,-4.9) node {$(-2)$} ;
      \draw[thick] (-5.5,-4.5)--(-5.5,5.3);
  \draw  (-5.5,5.6) node {${\mathcal{D}}_0^*$} ;\draw[gray] (-5.5,-4.9) node {$(-2)$} ;

                    \draw[blue!50, thick, ->] (-1.9,3.3 )--(-1.9,2);
            \draw[blue!50, thick, ->] (-1.9,3.3 )--(-0.6,3.3);
               \draw[blue!50] (-0.1 , 3)  node  {$v_{x 3}$} ;
         \draw[blue!50] (-1.3 , 2.3)  node  {$u_{x 3}$} ;
              
          \draw[blue!50, thick, ->] (-5.5,4.3)--(-5.5,3);
            \draw[blue!50, thick, ->] (-5.5,4.3)--(-4.2,4.3);
               \draw[blue!50] (-3.8 , 4.7)  node  {$\tilde u_1$} ;
         \draw[blue!50] (-5.9 , 2.7)  node  {$\tilde v_1$} ;

          \draw[blue!50, thick, ->] (5.5,4.3)--(5.5,3);
            \draw[blue!50, thick, ->] (5.5,4.3)--(4.2,4.3);
               \draw[blue!50] (3.8 , 4.7)  node  {$\tilde u_2$} ;
         \draw[blue!50] (5.9 , 2.7)  node  {$\tilde v_2$} ;

          \draw[blue!50, thick, ->] (8.1,-1.5)--(8.1,-0.2);
            \draw[blue!50, thick, ->] (8.1,-1.5)--(6.8,-1.5);
               \draw[blue!50] (7 , -1.1)  node  {$\tilde{v}_{3} $} ;
         \draw[blue!50] (8.5 , -0.1)  node  {$\tilde{u}_{3}$} ;
         
    \draw[blue!50, thick, ->] (5.5,-3.5)--(5.5,-2.2);
            \draw[blue!50, thick, ->] (5.5,-3.5)--(6.8,-3.5);
               \draw[blue!50] (7.4 , -3.8)  node  {${u}_{\infty 4} $} ;
         \draw[blue!50] (6.1 , -2.1)  node  {${v}_{\infty 4}$} ;
         
               \draw[blue!50, thick, ->] (1.9,-2 )--(1.9,-3.3);
            \draw[blue!50, thick, ->] (1.9,-2 )--(3.2,-2);
               \draw[blue!50] (3.5 , -2.3)  node  {$v_{1 3}$} ;
         \draw[blue!50] (2.4 , -3.5)  node  {$u_{1 3}$} ;

            \draw[blue!50, thick, ->] (-5.5,1.3)--(-5.5,0);
            \draw[blue!50, thick, ->] (-5.5,1.3)--(-4.2,1.3);
               \draw[blue!50] (-3.8 , 1)  node  {$v_{03}$} ;
         \draw[blue!50] (-5.9 , -.3)  node  {$u_{03}$} ;

  \end{tikzpicture}

\caption[The space of initial values of $\Psix$]{The space of initial values of the resolved Painlev\'e VI vector field for $x\neq 0$.}\label{fig:inter}
\end{center}
\end{figure}

Figure \ref{fig:inter} illustrates a schematic drawing of the resultant collection of exceptional lines, $\mathcal H^*$ and $D_\infty^{**}$ and their intersections in the resolved space, as well as the coordinates that will be most important in the following. For each $x=x_0\not= 0, 1$, this regularized space will be denoted as $\mathcal{S}(x_0)$. Moreover, we define  $\mathcal{S}(0)$ to be the result of the blow-up procedure for $x=0$. Its relation to the vector field is studied in the next section.
The union of $S(x_0)$ forms a fibre bundle
\[
\mathcal{S}:=\bigcup_{x_0\in \mathbb{C}\setminus\{1\}} \mathcal{S}(x_0).
\]
From the detailed charts in appendix section \ref{s:DetailsC}, one sees that for $x_0\neq 0$, the Painlev\'e vector field is ``vertical'' or tangent to the lines $\mathcal H^*$, $\mathcal D_\infty^{**}$, $\mathcal D_0^*$, $\mathcal D_x^*$, $\mathcal D_1^*$, which each have self-intersection $-2$. For this reason, such curves are often referred to as \lq\lq vertical leaves\rq\rq\ in Okamoto's construction.  For each $x=x_0\not=0, 1$, we define 
$$\mathcal{I}(x_0):=\mathcal H^*\cup \mathcal D_\infty^{**}\cup \mathcal D_0^*\cup \mathcal D_x^*\cup \mathcal D_1^*$$
the \emph{infinity set}, corresponding to the black part of the diagram shown in Figure \ref{fig:inter}. Okamoto's space of initial values for $x_0\not=0, 1$ is $\mathrm{Oka}(x_0):=\mathcal{S}(x_0)\setminus \mathcal{I}(x_0)$.

Note that the strict transforms $\mathcal H^*$, $\mathcal D_\infty^{**}$, $\mathcal D_0^*$, $\mathcal D_x^*$, $\mathcal D_1^*$ each have  self-intersection $-2$.
The corresponding Dynkin diagram reflecting their intersections, given in Figure \ref{fig:dynkin}, is equivalent to that for $D_4^{(1)}$.
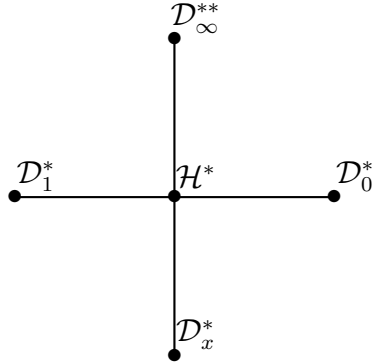
\begin{figure}[H]
\begin{center}
\begin{tikzpicture}[scale=0.7]
 \draw (0,0) node{$\bullet$}; \draw (0.4,0.4) node{$\mathcal{H}^*$};
\draw[thick] (0,0)--(3,0);
 \draw (3,0) node{$\bullet$}; \draw (3.4,0.4) node{$\mathcal D_0^*$};
 \draw[thick] (-3,0)--(0,0);
 \draw (-3,0) node{$\bullet$}; \draw (-2.6,0.4) node{$\mathcal D_1^*$};
 \draw[thick] (0,0)--(0,3);
 \draw (0,3) node{$\bullet$}; \draw (0.4,3.4) node{$\mathcal D_\infty^{**}$};
 \draw[thick] (0,0)--(0,-3);
 \draw (0,-3) node{$\bullet$}; \draw (0.4,-2.6) node{$\mathcal D_x^*$};
     
         \end{tikzpicture}\caption[The Dynkin diagram]{The Dynkin diagram with nodes representing $(-2)$-lines in Okamoto's space, for $x\neq 0,1$, is equivalent to that for $D_4^{(1)}$.}\label{fig:dynkin}
\end{center}
\end{figure}

\section{The vector field in the limit space}\label{s:S0}

 When $x\to0$ (or more precisely $\Re (t)\to-\infty$), we get the autonomous limiting system
 \begin{align}
\label{eq:autouv}&\begin{cases}
   \dot u &= -u \{ (u-1)\left(2uv-2\theta+\theta_\infty \right)-\theta_1\},\\
   \dot v &=  uv\left((3u-2)v-4\theta+2\theta_\infty\right)+(2\theta-\theta_\infty-\theta_1)v   + \theta (\theta-\theta_\infty),
   \end{cases}\\
           \nonumber \textrm{where}\ & \dot u=\partial E_0/\partial v, \dot v=-\,\partial E_0/\partial u, \textrm{with}\\
  \nonumber &E_0:=-u \{ (u-1)v\left(uv-2\theta+\theta_\infty\right)-\theta_1v  + \theta (\theta-\theta_\infty)  \}.
\end{align}
We can solve this Hamiltonian system completely: if the values of the $\theta_i$'s are generic, {i.e. if they belong to an open dense subset of the set of all possible values of those parameters}, we obtain a one-parameter family of solutions that lies on the line $\{u=0\}$. Again for generic $\theta_i$ values, no solutions lie on the line $\{u=1\}$.  
Let us assume $u\not \equiv 0,1$. Then the Hamiltonian system \eqref{eq:autouv} yields
\begin{align*}
  v =&- \frac{\dot u}{2u^2(u-1)}+\frac{\theta_1}{2u(u-1)}+\frac{\theta+\overline{\theta }}{2u},
\end{align*}
leading to
       \begin{align*}
  \ddot u =& \frac{3u-2}{2u(u-1)}\dot u^2+\frac{\theta_\infty^2}{2} (u-1)u^2-\frac{\theta_1^2u^2}{2(u-1)} .
\end{align*}
Note that if $(u(t),v(t))$ is a solution of the autonomous Hamiltonian system, then $\eta_0:=E_0(u(t),v(t))$ is constant. Setting $u_3:=1/u$, the autonomous differential equation for $u$ yields
 $(\dot u_3)^2=\alpha u_3^2+\beta u_3+\gamma$, where 
 $\alpha =  4\eta_0+( \theta+\overline{\theta}-\theta_1)^2\, , \beta=   \theta_1^2-\theta_\infty^2 -\alpha $ and $\gamma=\theta_\infty^2$.
 This integrates as $$u_3(t)=\left\{\begin{array}{lll}
 \frac{\sqrt{4\gamma \alpha-\beta^2}\mathrm{sinh}(\sqrt{\alpha}\, t+\eta_1)-\beta}{2\alpha} &\mathrm{if} & \alpha \neq 0\vspace{.3cm}\\
\left(\frac{\sqrt{ \beta}}{2}\, t+\eta_1\right)^2-\frac{\gamma }{ \beta}  &\mathrm{if} & \alpha = 0\, , \end{array} \right.$$ where $\eta_1$ is an arbitrary integration constant. 
In particular, we find the following list of equilibrium points (trajectories reduced to one point) of the autonomous Hamiltonian system for generic $\theta_i$'s:
 $$\left\{\begin{array}{rcllll}
(u,v)&=& \left(\frac{\theta_\infty-\theta_1}{ \theta_\infty}\, , \frac{\theta_\infty\overline{\theta}}{\theta_\infty-\theta_1}\right) &\mathrm{for} & \eta_0=\frac{(\theta_1-\theta_\infty)^2-(\theta_0+\theta_x-1)}{4}\vspace{.2cm}\\
(u,v)&=& \left(\frac{\theta_\infty+\theta_1}{ \theta_\infty}\, , \frac{\theta_\infty \theta}{\theta_\infty+\theta_1}\right) &\mathrm{for} & \eta_0=\frac{(\theta_1+\theta_\infty)^2-(\theta_0+\theta_x-1)}{4} \, . \end{array} \right.$$

We may now compactify the space of initial values $\C^2_{u,v}$ to $\mathcal{S}(0)$. Figure \ref{fig:inter0} contains a schematic drawing of how the limits of the components of the infinity set and the exceptional lines arrange in this space. Here as usual, red lines have self-intersection $(-1)$. The notable differences with the configuration in $\mathcal{S}(x)$ with $x=0$ are the following, where we use the superscript ``$0$'' when convenient to indicate particularities for the $x=0$ case:
\begin{itemize}
\item
 After blow-up of $\beta_0:(u_1,v_1)=(0,0)$, the point $\beta_x^0:(u_{01},v_{01})=(0,0)$ which has to be blown up lies on the intersection of (the strict transform) of $\mathcal{H}$ and the exceptional line $\mathcal{D}_0=\mathcal{D}_0^0$.
\item As a result, we still have $\mathcal{H}^*=\{\tilde v_1=0\}\cup \{\tilde v_2=0\}$, but $\mathcal{D}_0^0=\mathcal{D}_0^{0*}\cup \mathcal{D}_x^0$. 
\item Moreover, the point $\gamma_x^0:(u_{x2},v_{x2})=(0,0)$ now corresponds to the intersection $\mathcal{D}_0^{0*}\cap \mathcal{D}_x^0$. 
\item As a result, we still have $\mathcal{D}_x^{0*}=\{v_{x3}=0\}\cup \{\tilde u_1=0\}$, but $\mathcal{D}_0^{0*}=\mathcal{D}_0^{0**}\cup \mathcal{E}_x^0$, where $\mathcal{E}_x^0:\{u_{x3}=0\}\cup \{v_{x4}=0\} . $
\item Finally, the blow-up of $\gamma_0^0:(u_{02},v_{02})=(\theta_0,0)$ yields the strict transform $\mathcal{D}_0^{0***}:\{v_{03}=0\}\cup \{u_{x4}=0\}$ of self-intersection $(-4)$.
\end{itemize}

\begin{figure}[H]
\begin{center}
\begin{tikzpicture}[scale=0.5]
\draw[thick] (-6.5,4.3)--(6.5,4.3);
  \draw (6.9,4.3) node {$ {\mathcal H}^*$} ;    \draw[gray] (-7.1,4.3) node {$(-2)$} ;

              \draw[red!50, thick] (-14,2.3)--(-7.1,2.3);              \draw[red!50] (-14.4,2.3)  node  {$\mathcal{V}_0^*$} ; 
            \draw[red!50, thick] (-12,1.3)--(-7.1,1.3);            \draw[red!50] (-12.4,1.3)  node  {$\mathcal{E}_0$} ; 
               \draw[red!50,thick]  (-12,-4)--(-2.2,-4);          \draw[red!50] (-12.4,-4)  node  {$\mathcal{E}_x^0$} ; 
                 
                         \draw[thick] (-9,-4.5)--(-9,3.3);
     \draw  (-9,3.6) node {${\mathcal{D}}_0^{0***}$} ;\draw[gray] (-9.2,-4.9) node {$(-4)$} ;

  \draw[blue!50, thick, ->] (-13.5,2.3)--(-13.5,1);
            \draw[blue!50, thick, ->] (-13.5,2.3)--(-12.2,2.3);
               \draw[blue!50] (-12.2 , 2.6)  node  {$v_{0}$} ;
         \draw[blue!50] (-13.5 , 0.7)  node  {$  u_{0}$} ;

                        \draw[red!50, thick] (-0,1.5)--(4.8,1.5);            \draw[red!50] (-0.4,1.5)  node  {$\mathcal{V}_1^*$} ;
           \draw[red!50, thick] (-0,-2)--(4.8,-2);            \draw[red!50] (-0.4,-2)  node  {$\mathcal{E}_1$} ;

 \draw[red!50, thick] (3.5,-1.5)--(8.3,-1.5);            \draw[red!50] (8.9,-1.5)  node  {$\mathcal{E}_{\infty}$} ;
   \draw[red!50, thick] (3.5,-3.5)--(8.3,-3.5);            \draw[red!50] (8.9,-3.5)  node  {$\mathcal{E}_{\infty}^-$} ;

\draw[thick] (5.5,-4.5)--(5.5,5.3);
  \draw (5.5,5.6) node {$\mathcal{D}_\infty^{**}$} ;\draw[gray] (5.5,-4.9) node {$(-2)$} ;
  \draw[thick] (1.9,-4.5)--(1.9,5.3);
  \draw  (1.9,5.6) node {${\mathcal{D}}_1^*$} ;\draw[gray] (1.9,-4.9) node {$(-2)$} ;
         \draw[thick] (-5.5,-4.5)--(-5.5,5.3);
  \draw  (-5.5,5.6) node {${\mathcal{D}}_x^{0*}$} ;\draw[gray] (-5.5,-4.9) node {$(-2)$} ;
      
   \draw[blue!50, thick, ->] (-5.5,-4)--(-5.5,-2.7);
            \draw[blue!50, thick, ->] (-5.5,-4)--(-4.2,-4);
               \draw[blue!50] (-3.7 , -3.7)  node  {${v}_{x3} $} ;
         \draw[blue!50] (-5 , -2.4)  node  {${u}_{x3}$} ;

\draw[blue!50, thick, ->] (-9,-4)--(-9,-2.7);
            \draw[blue!50, thick, ->] (-9,-4)--(-7.7,-4);
               \draw[blue!50] (-7.3 , -3.7)  node  {${u}_{x4} $} ;
         \draw[blue!50] (-8.5 , -2.4)  node  {${v}_{x4}$} ;

          \draw[blue!50, thick, ->] (-5.5,4.3)--(-5.5,3);
            \draw[blue!50, thick, ->] (-5.5,4.3)--(-4.2,4.3);
               \draw[blue!50] (-3.8 , 4.7)  node  {$\tilde u_1$} ;
         \draw[blue!50] (-5.9 , 2.7)  node  {$\tilde v_1$} ;

          \draw[blue!50, thick, ->] (5.5,4.3)--(5.5,3);
            \draw[blue!50, thick, ->] (5.5,4.3)--(4.2,4.3);
               \draw[blue!50] (3.8 , 4.7)  node  {$\tilde u_2$} ;
         \draw[blue!50] (5.9 , 2.7)  node  {$\tilde v_2$} ;

          \draw[blue!50, thick, ->] (8.1,-1.5)--(8.1,-0.2);
            \draw[blue!50, thick, ->] (8.1,-1.5)--(6.8,-1.5);
               \draw[blue!50] (7 , -1.1)  node  {$\tilde{v}_{3} $} ;
         \draw[blue!50] (8.5 , -0.1)  node  {$\tilde{u}_{3}$} ;
         
    \draw[blue!50, thick, ->] (5.5,-3.5)--(5.5,-2.2);
            \draw[blue!50, thick, ->] (5.5,-3.5)--(6.8,-3.5);
               \draw[blue!50] (7.4 , -3.8)  node  {${u}_{\infty 4} $} ;
         \draw[blue!50] (6.1 , -2.1)  node  {${v}_{\infty 4}$} ;
         
               \draw[blue!50, thick, ->] (1.9,-2 )--(1.9,-3.3);
            \draw[blue!50, thick, ->] (1.9,-2 )--(3.2,-2);
               \draw[blue!50] (3.5 , -2.3)  node  {$v_{1 3}$} ;
         \draw[blue!50] (2.4 , -3.5)  node  {$u_{1 3}$} ;

            \draw[blue!50, thick, ->] (-9,1.3)--(-9,0);
            \draw[blue!50, thick, ->] (-9,1.3)--(-7.7,1.3);
               \draw[blue!50] (-7.3 , 0.9)  node  {$v_{03}$} ;
         \draw[blue!50] (-8.5 , -.3)  node  {$u_{03}$} ;

  \end{tikzpicture}
\caption[The limit space]{The limit space for $x=0$ of the space of initial values for $x\neq 0$.}\label{fig:inter0}
\end{center}
\end{figure}
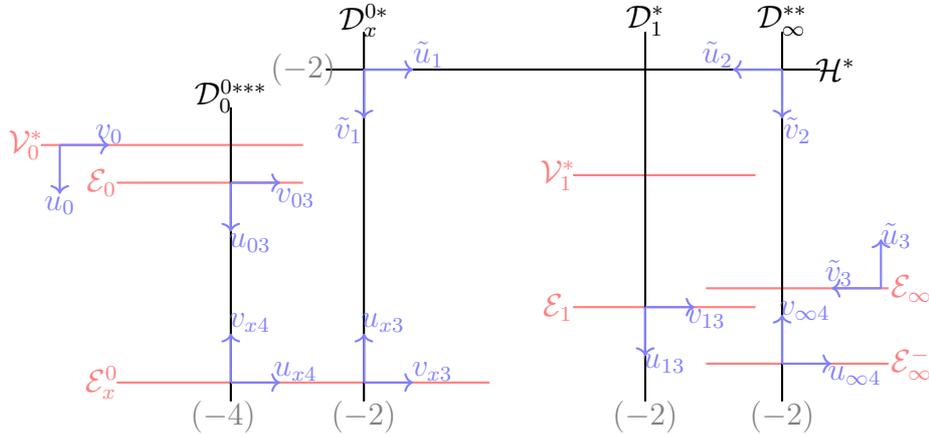

The resulting autonomous vector field in $\mathcal{S}(0)$ is obtained from the one in appendix section \ref{s:DetailsC} by systematically setting $x=0$. 
For convenience of the reader, the formulae are given in appendix section \ref{s:Details0B}.

We use the term \emph{elliptic base points} for a point where the induced autonomous vector field in $\mathcal{S}(0)$ is undetermined. There is one such elliptic base point, given by 
$$\mathfrak{u}: (u_{x4},v_{x4})=(0,0) \in \mathcal{D}_0^{***}\cap \mathcal{E}_x^0\, .$$
This elliptic base point cannot be resolved by blow-ups!\footnote{Moreover, when following through the process of Okamoto desingularization, one realizes that $\gamma_0^0:(u_{02},v_{02})=(\theta_0,0)$ was in fact not an elliptic base point. } Note however that the autonomous  energy function $E_0$ is well-defined (and infinite) at $\mathfrak{u}$.

Let us denote $\mathcal{I}^0$ the subset of $\mathcal{S}(0)$ where the autonomous vector field is infinite or undefined. We find
$$\mathcal{I}^0=\mathcal{E}_x^{0}\cup \mathcal{D}_x^{0*}\cup\mathcal{H}^*\cup \mathcal{D}_1^*\cup \mathcal{D}_\infty^{**}.$$
This set corresponds precisely to the points where the autonomous  energy function $E_0$ is infinite.  As explained above, we have $$\lim_{x\to 0} \mathcal{I}(x)=\mathcal{D}_0^{0***}\cup \mathcal{I}^0\, .$$

In order to complete the description of the autonomous vector field in $\mathcal{S}(0)\setminus \{\mathfrak{u} \}$, it remains to investigate  trajectories that might be contained in $\mathcal{S}(0)\setminus \left(\mathcal{I}^0\cup \C^2_{u_0,v_0}  \right)$. We find the following, where as usual we assume the values of the $\theta_i$'s to be generic:
\begin{itemize}
\item There is no trajectory contained in any of the following: 
\begin{itemize}
\item $\mathcal{E}_\infty\setminus \mathcal{D}_\infty^{**}:\{\tilde u_3=0\}$, \item $\mathcal{E}_\infty^-\setminus \mathcal{D}_\infty^{**}:\{u_{\infty 3}=0\}$, \item $\mathcal{E}_1\setminus \mathcal{D}_1^*:\{v_{14}=0\}$.
\end{itemize}
\item The line  $ \mathcal{E}_0^{0}\setminus \mathcal{D}_0^{0***}:\{v_{04}=0\}$ is the union of one trajectory and one equilibrium point, given by 
$$u_{04}=-\frac{ ( \theta_0-\theta)( \theta_0-\overline{\theta}) \theta_0}{2\theta_0+\theta_1-(\theta+\overline{\theta})}$$
with energy $\eta_0=- (\theta_x-1)\theta_0$.
\item Every point of $\mathcal{D}_0^{0***}\setminus \mathcal{E}_x^{0}:\{v_{03}=0\}$ is an equilibrium point of the autonomous vector field, with energy
$\eta_0 =(u_{03}-(\theta_x-1))(u_{03}+\theta_0)$.
\end{itemize}

\section{Movable singularities in the Okamoto's space}\label{s:movable}
In this section, we will consider neighbourhoods of exceptional lines where the Painlev\'e vector field becomes unbounded.
The construction given in Appendix \ref{app:resolution} shows that these are given by the lines $\mathcal{E}_0$, $\mathcal{E}_x$, $\mathcal{E}_1$, $\mathcal{E}_{\infty}$, $\mathcal{E}_{\infty}^-$.

\subsection{Points where $u$ has a zero and $v$ a pole.}
The set $\mathcal{E}_0\setminus\mathcal{I}$ is given by $\{v_{04}=0\}$, in the $(u_{04},v_{04})$ chart, see Section \ref{app:0404}.
Suppose $u_{04}(\tau)=B$, $v_{04}(\tau)=0$, for some complex numbers $\tau$, $B$.
From the system of differential equations in Section \ref{app:0404}, we get:
$$
\begin{aligned}
v_{04}(t)
=\ &
\frac{e^\tau}{e^\tau-1}(t-\tau)-e^{\tau}\frac{\theta_0(1+e^{\tau}) -\theta_x -\theta_1e^{\tau}+2}
{2(e^{\tau}-1)^2}
(t-\tau)^2
\\
&
+\left(
\frac{2B e^{2\tau}(1+e^{\tau})}{3(1-e^{\tau})^3}
+
F_1(\tau)
\right)
(t-\tau)^3
+
O((t-\tau)^4),
\end{aligned}
$$
with
$$
\begin{aligned}
F_1(\tau)=\ &
-
\frac{e^{\tau}}{6(1-e^{\tau})^3}
\Bigl[
3( \theta_0-\theta_x+1) + ( \theta_0-\theta_x)^2   + 3 e^{\tau} 
- 8 \theta_0 (\theta_x+\theta_1) e^{\tau} \\
&\qquad\qquad  \qquad +13\theta_0e^{\tau} - 2 \theta_x e^{\tau} 
+\theta_1 ( 2 \theta_x-5  ) e^{\tau} \\ &\qquad\qquad  \qquad +2(\theta+\theta_0)(\bar \theta+\theta_0) e^{\tau} + 
( \theta_0-\theta_1)^2 e^{2\tau}  
\Bigr]\, .
\end{aligned}
$$
Since (see Section \ref{app:0404})
$$
u=u_{04}v_{04}^2+\theta_0v_{04},
\quad
v=\frac{1}{v_{04}},
$$
we obtain the series expansions for $(u,v)$:
$$
\left\{\begin{array}{rcl}
u(t)&=&\ 
\frac{\theta_0 e^\tau}{e^\tau-1}(t-\tau)
+
\left(
B\frac{e^{2\tau}}{(e^{\tau}-1)^2}-
\theta_0e^{\tau}\frac{\theta_0(1+e^{\tau})-\theta_x -\theta_1e^{\tau}+2}{2(e^{\tau}-1)^2}
\right)
(t-\tau)^2
+
O((t-\tau)^3)\, ,
\vspace{.2cm}\\
v(t)&=&\  \frac{e^{\tau}-1}{e^{\tau}}\frac{1}{t-\tau} 
+
\frac{ \theta_0(1+e^{\tau})-\theta_x -\theta_1e^{\tau}+2} {2e^{\tau}}
+\left(\frac{2B(e^{\tau}+1)}{3(e^{\tau}-1)}+F_2(t)\right)(t-\tau)
+O((t-\tau)^2)\, ,
\end{array}\right.
$$
with
$$
\begin{aligned}
F_2(\tau)=\ &\frac{1}{12e^{\tau}(e^{\tau}-1)}
\Bigl[
(\theta_0-\theta_x +3)^2 -3         -4( \theta-\theta_0)(\bar  \theta-\theta_0) e^{\tau}  +4\theta_0^2 e^{\tau} 
\\
&  \qquad \qquad\qquad
- 2(2\theta_0+1)^2 e^{\tau}   + 2(\theta_0+\theta_x-1 )(\theta_0+ \theta_1+2) e^{\tau} \\
&  \qquad \qquad\qquad +  (\theta_0-\theta_1)^2 e^{2\tau} 
\Bigr].
\end{aligned}
$$
Note that $u$ has a simple zero at $t=\tau$ and $v$ a simple pole with residue $1-e^{-\tau}$.

\subsection{Points where $u=1$ and $v$ has a pole.}
The set $\mathcal{E}_1\setminus\mathcal{I}$ is given by $\{v_{14}=0\}$, in the $(u_{14},v_{14})$ chart, see Section \ref{app:1414}.
Suppose $u_{14}(\tau)=B$, $v_{14}(\tau)=0$, for some complex numbers $\tau$, $B$.
From the system of differential equations in Section \ref{app:1414}, we get:
$$
\begin{aligned}
v_{14}(t)
=\ &
-(t-\tau)+\frac{\theta+\bar\theta+\theta_1(e^{\tau}-3)
	-
	\theta_0e^{\tau}}{2(e^{\tau}-1)}(t-\tau)^2
\\&
+
\left(
\frac{2B(e^{\tau}-2)}{3(1-e^{\tau})}+F_3(\tau)
\right)
(t-\tau)^3
+O((t-\tau)^4),
\end{aligned}
$$
with
$$
\begin{aligned}
F_3(\tau)=-\frac{1}{6(e^{\tau}-1)^2}&\Bigl[
( \theta+\bar\theta-5\theta_1)^2 + 2( \theta \bar\theta -5 \theta_1^2)+ ( \theta_0 -\theta_1)^2e^{2\tau}
\\&\  
-2( \theta_0( \theta+ \bar\theta+2) + \theta \bar\theta) e^{\tau} \\
& \ + 2(3\theta_1+1)( \theta  + \bar\theta +   \theta_0-2\theta_1  ) e^{\tau}   
\ \Bigr] .
\end{aligned}
$$

Since (see Section \ref{app:1414})
$$
u=u_{14}v_{14}^2+\theta_1v_{14}+1,
\quad
v=\frac{1}{v_{14}},
$$
we obtain the series expansions for $(u,v)$:
$$
\left\{\begin{array}{rcl}
u(t)&=&\ 1- \theta_1(t-\tau) +
\left(B+
\theta_1
\frac{\theta+\bar\theta+\theta_1(e^{\tau}-3)
	-
	\theta_0e^{\tau}}{2(e^{\tau}-1)}
\right)(t-\tau)^2
+O((t-\tau)^3),
\vspace{.2cm}\\
v(t)&=&\ -\frac{1}{t-\tau}
-
\frac{\theta+\bar\theta+\theta_1(e^{\tau}-3)-\theta_0e^{\tau}}{2(e^{\tau}-1)}
-
\left(
\frac{2B(e^{\tau}-2)}{3(1-e^{\tau})}+F_4(\tau)
\right)(t-\tau)
+O((t-\tau)^2),
\end{array}\right.
$$
with
$$
F_4(\tau)=F_3(\tau)+\left(\frac{\theta+\bar\theta+\theta_1(e^{\tau}-3)
	-
	\theta_0e^{\tau}}{2(e^{\tau}-1)}\right)^2.
$$
At $t=\tau$, $u-1$ has a simple zero, while $v$ has a simple pole with residue $-1$.

\subsection{Points where $u(\tau)=e^{\tau}$ and $v$ has a pole.}

The set $\mathcal{E}_x\setminus\mathcal{I}$ is given by $\{v_{x4}=0\}$, in the $(u_{x4},v_{x4})$ chart, see Section \ref{app:x4x4}.
Suppose $u_{x4}(\tau)=B$, $v_{x4}(\tau)=0$, for some complex number $B$.
From the system of differential equations in Section \ref{app:x4x4}, we get:
$$
\begin{aligned}
v_{x4}=&\ (t-\tau)
+
\frac{\theta_0 +\theta_x   - (\theta +\bar \theta -\theta_x) e^{\tau} }{2(e^{\tau}-1)}(t-\tau)^2
\\
&+\left(\frac{ B (e^{2\tau}-1 )}{3e^{\tau}(e^{\tau}-1)^2} + F_5(\tau) \right)(t-\tau)^3
+O((t-\tau)^4),
\end{aligned}
$$
with
$$
F_5(\tau)=\frac{
\Bigl[( \theta+\bar \theta-2\theta_x)^2 
	-3  \theta_{x}^2
	+2\theta\bar \theta\Bigr] e^{2\tau}    +( \theta_{0}+2\theta_x)^2- \theta_{x}^2 
	-2 \Bigl[1+\theta_{0}(\theta+\bar \theta)+\theta_{x}  -\theta_1+\theta \bar \theta\Bigr]e^{\tau}
}
{6 (e^{\tau}-1)^2}.
$$
Since, as calculated in Section \ref{app:x4x4}:
$$
u=(u_{x4}v_{x4}+e^{t}\theta_x)v_{x4}+e^{t},
\quad
v=\frac{1}{((u_{x4}v_{x4}+e^{t}\theta_x)v_{x4}+e^{t})v_{x4}},
$$
we obtain:
$$\left\{
\begin{array}{rcl}
u(t)&=&\ e^{\tau}+e^{\tau}(\theta_x+1)(t-\tau)
+\left(B+\frac{e^{\tau}}{2}+ e^{\tau}\theta_x\frac{\theta_0 +\theta_x-2  - (\theta_0+\theta_1-3) e^{\tau}}{2(e^{\tau}-1)}\right)(t-\tau)^2+O((t-\tau)^3)\, ,
\vspace{.2cm}\\
v(t)&=&\ \frac{e^{-\tau}}{t-\tau}
-
e^{-\tau} \frac{\theta_0 -\theta_x-2  - (\theta +\bar \theta -3\theta_x-2) e^{\tau}}{2(e^{\tau}-1)} 
+e^{-\tau}\left(B\cdot\frac{2-4e^{\tau}}{3e^{\tau}(e^{\tau}-1)}+F_6(\tau)\right)(t-\tau) +O((t-\tau)^2)\, ,
\end{array}\right.
$$
with
$$\begin{aligned}
F_6(\tau)=-\frac{1}{12(e^\tau -1)^2}\Bigl\{&\left[\theta_\infty^2+2(\theta+\bar \theta+2\theta_x)(\theta_x-3)+(3\theta_x+5)^2-19 \right]e^{2\tau}\\
& -\left[\theta_\infty^2-(\theta_1-2\theta_x)^2+2(\theta_x-12)+(5\theta_x+1)^2 +(\theta_0-6)^2\right]e^\tau
\\&+\left[8\theta_x^2+(\theta_0-\theta_x-3)^2-2\right]\Bigr\}
\end{aligned}
$$
At $t=\tau$, obviously $v$ has a simple pole with residue $e^{-\tau}$, while $u(t)-e^{\tau}$ has a simple zero.

\subsection{Points where $u$ has a pole and $v$ a zero.}

Such points belong to $\mathcal{E}_{\infty}$ and $\mathcal{E}_{\infty}^-$, which are obtained by blowing up the points $\beta_{\infty}$ and $\beta_{\infty}^-$ on $\mathcal{D}_{\infty}$. 
We notice that the initial vector field (see Section \ref{app:00}) does not depend on the sign of $\theta_{\infty}$. Moreover, if we replace $\theta_\infty$ by $-\theta_\infty$, the roles of  $\beta_{\infty}$ and $\beta_{\infty}^-$  are interchanged. Because of that symmetry, we may consider only the case when the solution intersects $\mathcal{E}_{\infty}$.

The set $\mathcal{E}_{\infty}\setminus\mathcal{I}$ is given by $\{\tilde u_3=0\}$ in the $(\tilde u_3,\tilde v_3)$ chart, see Section \ref{app:3tilde}.
Suppose $\tilde u_3(\tau)=0$, $\tilde v_3(\tau)=B$.
From the differential equations in Section \ref{app:3tilde}, we get:
$$
\begin{aligned}
\tilde u_3(t)=&\ \frac{\theta_{\infty}}{1-e^{\tau}}(t-\tau)
-\theta_{\infty}
\frac{ (\theta_{\infty}+\theta_{x}-2 )e^{\tau}+\theta_{\infty}+\theta_1 +2 B }{2(e^{\tau}-1)^2}(t-\tau)^2
+
O((t-\tau)^3).
\end{aligned}
$$
Then, using the relations:
$$
u=\frac{1}{\tilde u_3},
\quad
v=(\theta -\tilde u_3\tilde v_3)\tilde u_3,
$$
we get:
$$\left\{
\begin{array}{rcl}
u(t)&=&\frac{1-e^{\tau}}{\theta_{\infty}(t-\tau)} 
+
\frac{ (\theta_{\infty}+\theta_{x}-2 )e^{\tau}+\theta_{\infty}+\theta_1 +2 B}{2\theta_\infty}
+
O(t-\tau),
\vspace{.2cm}\\
v(t)&= &
\frac{\theta\theta_{\infty}}{(1-e^{\tau})}(t-\tau)
-
\theta_{\infty}\frac{  \theta(\theta_{\infty}+\theta_{x}-2)e^{\tau}+\theta (\theta_{\infty}  +\theta_1 )+2(\theta+\theta_{\infty})B}{ 2(e^{\tau}-1)^2}
(t-\tau)^2
+O((t-\tau)^3).
\end{array}\right.
$$
Note that $u$ has a simple pole with residue $(1-e^{\tau})/\theta_{\infty}$, while $v$ has a simple zero.
In the intersection points with $\mathcal{E}_{\infty}^-$, $u$ has a simple pole with residue $-(1-e^{\tau})/\theta_{\infty}$ and $v$ a simple zero.

\section{Estimates and the main result}\label{s:estimates}
In this section, we estimate the distance of the vector field from each vertical leaf, for sufficiently small $x$. These estimates allow us to describe the domain of each solution in $\mathcal S\backslash \mathcal I$, which is Okamoto's space of initial values. The results will be used to prove properties of the limit set of each solution. 

 Given  $0<\epsilon<1$, $\epsilon\in\mathbb R$, define a disk $R=R_{\epsilon}=\{x\in\mathbb C \mid |x|<\epsilon\}$. Letting $\xi\in R$, $r<|\xi|<\epsilon$, define a disk $D=D_r(\xi)=\{x\in R \bigm| |x-\xi|<r\}$ that lies in the interior of $R$. Defining a new time coordinate $t=\ln x$, we have corresponding domains $R_t$ and $D_t$ in the $t$-plane.
 Note that $D_t$ is no longer a circular disk, but lies inside a rectangular region in the left half of the $t$-plane, see Figure \ref{fig:Dt}.
 \begin{figure}[h]
\begin{tikzpicture}
 \draw[fill=gray!30](0,0) circle (2);
 \draw[fill=black](0,0) circle (0.05);
 \draw[red,thick](0,0.618)--(0,1.618);
 \draw[red,thick](0.5,0.36)--(1.31,0.95);
 \draw [red,thick,domain=36:90] plot ({1.618*cos(\x)}, {1.618*sin(\x)});
 \draw [red,thick,domain=36:90] plot ({0.618*cos(\x)}, {0.618*sin(\x)});
 \node[left] at (0,0){$0$};
 \node at (-1.6,-1.6) {$R_{\epsilon}$};
 \draw[fill=white](0.5,1) circle (0.5) node {$D_r$};
 \draw[-{Latex[length=3mm, width=2mm]}] (2.5,0) .. controls (3,0.3) .. (4.5, 0) node[pos=0.5,above]{$t=\ln x$};
 \filldraw[gray!30](5,-2) rectangle (9,2);
 \draw(9,-2)--(9,2) node[pos=0.5,right]{$\Re t=\ln\epsilon<0$};
 \filldraw[fill=white,draw=red,thick](8.5,0.5) rectangle (7,1.5); 
 \draw[pattern=north east lines, pattern color=gray!30] (8.5,0.5) rectangle (7,1.5);
 \draw[draw=red,thick](8.5,0.5) rectangle (7,1.5); 
 \end{tikzpicture}
 	 \caption{Domains $R_{\epsilon}$ and $D_r$. $R_{\epsilon}$ is the disk centred at the origin with radius {$0<\epsilon<1$}. $D_r$ is a disk within $R_{\epsilon}$ and does not contain the origin.  The image of $R_{\epsilon}$ by the logarithmic function is the half-plane placed on the left to the boundary $\Re t=\ln\epsilon$.
In the left side of the figure, notice a curvilinear ``quadrangle'' consisting of two circular arcs centred at the origin and two segments placed on radii of $R_{\epsilon}$, such that it is circumscribed about $D_r$.
That ``quadrangle'' is mapped to the red rectangle in the right side, which thus will contain the image of $D_r$.}\label{fig:Dt}
\end{figure}
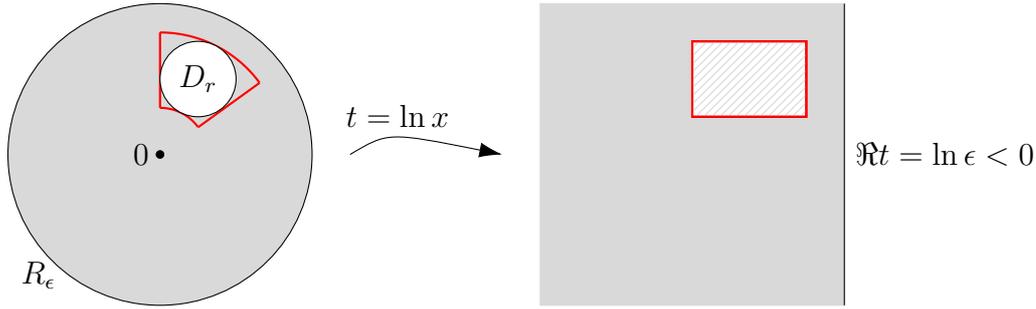

 The reader may find it useful to consult Figure \ref{fig:inter} in the proofs of the following results.
 
\begin{lemma}\label{lem:distancefn}
Given $x\in\mathbb C\backslash \{0\}$, there exists a continuous complex valued function $d$ in a neighbourhood of the infinity set $\mathcal I$ in Okamoto's space, such that 
\[
d=\begin{cases} 
	\frac{1}{E}& \ \textrm{in a neighbourhood of}\  \mathcal{H}^*\cup \mathcal{D}_1^*\cup \mathcal{D}_\infty^{**}  \setminus ( \mathcal{D}_0^*\cup\mathcal{D}_x^*) ,\\	
	-\frac{x-1}{x}\omega_{03}& \ \textrm{in a neighbourhood of}\ \mathcal{D}_0^*\backslash \mathcal{H}^*,\\
 	-\omega_{x3}& \ \textrm{in a neighbourhood of}\ \mathcal{D}_x^*\backslash \mathcal{H}^*.\\
	\end{cases}	
\]
Note that $d$ vanishes on $\mathcal I$ and that $d$ is not defined at $x=0$.
\end{lemma}
\begin{proof}
From Section \ref{app:0303}, the set $\mathcal{D}_0^*\backslash \mathcal{H}^*$ is given by $v_{03}=0$ in the $(u_{03},v_{03})$ chart.
As we approach $\mathcal{D}_0^*$, we have:
$$
E\omega_{03}\sim-\frac{x}{x-1},\quad \text{as}\quad v_{03}\to0.
$$
 
From Section \ref{app:x3x3}, the set $\mathcal{D}_x^*\backslash \mathcal{H}^*$ is given by $v_{x3}=0$ in the $(u_{x3},v_{x3})$ chart.
As we approach $\mathcal{D}_x^*$, we have:
$$
E\omega_{x3}\sim-1-\frac{x}{u_{x3}}\quad\text{as}\quad v_{x3}\to0.
$$ 
%Note that

Thus, as we approach $\mathcal{H}^*$: $u_{x3}\to\infty$, we have that $-\omega_{x3}\sim1/E$.

 \end{proof}

\begin{lemma}\label{lemmaDinfty}
For every $\epsilon>0$, there exists a neighbourhood $U$ of $\mathcal{D}_{\infty}^{**}$ such that
$$
\left|
\frac{\dot{E}}{E}+\frac{e^{t}}{e^{t}-1}
\right|
<\epsilon.
$$
\end{lemma}
\begin{proof}
The proof follows from the expressions for $\dot{E}/E$ in $(u_2,v_2)$ and $(u_3,v_3)$ charts (see Sections \ref{app:22} and \ref{app:33}), where $\mathcal{D}_{\infty}$ is given by $u_2=0$ and $u_3=0$ respectively.
\end{proof}

%%%
%%% Lemma 3.2
%%%
\begin{lemma}\label{lemmaH} 
For every compact subset $K$ of $\mathcal D_\infty^{**}\cup \mathcal D_1^{*}\cup \mathcal H^*\backslash \mathcal{D}_x^*$, there exists a neighbourhood $V$ of $K$ and a constant $C>0$, such that
\[
\left|e^{-t}\,\frac{\dot E}{E}\right|<C 
\]
in $V$ for all $t$ such that $e^t$ is bounded away from $1$.
\end{lemma}

\begin{proof}
Note that $\mathcal H^*=\{\tilde v_1=0\}\cup \{\tilde v_2=0\}$ is parametrized by $\tilde u_1$ and $\tilde u_2=1/\tilde u_1$ respectively.  Moreover, $ \mathcal{D}_x^*$ is given in these charts by $\{\tilde u_1=x\}$ and $ \{x\tilde u_2=1\}$. In the respective coordinate charts $(\tilde u_1,\tilde  v_1)$, $(\tilde u_2, \tilde v_2)$ (see Sections \ref{app:1tilde} and \ref{app:2tilde}), we have
\[
e^{-t}\,\frac{\dot E}{E}=\begin{cases} 
	\frac{\tilde u_1-1}{\left(e^t-1\right) \left(e^t-\tilde u_1\right)}+\frac{ \left(\theta _{x}-1\right)\tilde u_1(\tilde u_1-1)}{e^t-\tilde u_1}\tilde v_1+{O}\left(\tilde v_1^2\right),\vspace{.2cm}\\
	-\frac{ \tilde u_2-1}{\left(e^t-1\right) \left(e^t \tilde u_2-1\right)}-\frac{ \left(\theta _{x}-1\right)\tilde u_2(\tilde u_2-1)}{ e^t \tilde u_2-1 }\tilde v_2 +{O}\left(\tilde v_2^2\right).\vspace{.2cm}\\
	\end{cases}
\]

So as long as we consider compact subsets of  $ \mathcal H^*\backslash \mathcal{D}_x^*$, the values of $\frac{1}{\tilde u_1-x}$ and $\frac{1}{x\tilde u_2-1}$ are bounded. We have now proven the desired result in a neighbourhood of any compact subset of $\mathcal{H}^*\setminus  \mathcal{D}_x^*$. Since $\mathcal D_\infty^{**}$ intersects with $\mathcal{H}^*\setminus  \mathcal{D}_x^*$, the result holds in a neighbourhood of $\mathcal D_\infty^{**}\cap \mathcal H^*$.

On the other hand, near $\mathcal D_\infty^{**}\setminus \mathcal H^*$, given by  $\{u_{\infty 4}=0\}$, we may consider only bounded values of $v_{\infty 4}$, and so we have  (see Section \ref{app:inf3inf3})
\[
e^{-t}\,\frac{\dot E}{E}=-\frac{1}{e^t-1}+\left[\theta v_{\infty 4}-\left(v_{\infty 4}+\overline{\theta}\right)\left(v_{\infty 4}+\theta_x-1\right) \right]u_{\infty 4}+O(u_{\infty 4}^2) \, .\]
Hence the result holds in a neighbourhood of the compact set $\mathcal D_\infty^{**}=\{u_{\infty 4}=0\}\cup \{\tilde u_2=0\}$.  Similarly, near $\mathcal D_1^{**}\setminus \mathcal H^*$, given by $\{v_{13}=0\}$, where we may consider only bounded values of $u_{13}$, we have (see Section \ref{app:1313}) \[
e^{-t}\,\frac{\dot E}{E}=\frac{(u_{13}+\theta_1)(u_{13}-\theta_x+1)}{(e^t-1)^2}v_{13}+O(v_{13}^2)  \, .\]
Hence the result holds for any compact subset $K$ of $\mathcal D_\infty^{**}\cup \mathcal D_1^{*}\cup \mathcal H^*\backslash \mathcal{D}_x^*$ and any $t$ as long as $\frac{1}{e^t-1}$ is bounded.
\end{proof}

\begin{remark} The estimate in the above Lemma \ref{lemmaH} applies to all compact subsets of $\mathcal D_\infty^{**}\cup \mathcal D_1^{*}\cup \mathcal H^*\backslash \mathcal{D}_x^*$ and, therefore, in particular to $\mathcal D_\infty^{**}\cup \mathcal D_1^{*}$. 
\end{remark}
\begin{lemma}[Behaviour near $\mathcal{D}_x^*\setminus\mathcal{H}^*$]\label{lemma:nearDx}
If a solution at a complex time $t$ is sufficiently close to $\mathcal{D}_x^*\setminus\mathcal{H}^*$, then there exists a unique $\tau\in\mathbf{C}$ such that $(u(\tau),v(\tau))$ belongs to the exceptional line $\mathcal{E}_x$.
In other words,  $u(\tau)=e^\tau$ and $v(t)$ has a pole at $t=\tau$.
Moreover, for sufficiently small $d(t)$ and bounded $u_{x3}$, we have $|t-\tau|=O(|e^{-t}d(t)||u_{x3}(t)|)$.

For large $R_x>0$, consider the set $\{t\in\mathbb{C}\mid |u_{x3}(t)|\le R_x\}$.
Its connected component containing $\tau$ is an approximate disk $\Delta_x$ with centre $\tau$ and radius $|d(\tau)e^{-\tau}|R_x$, and
$t\mapsto u_{x3}(t)$ is a complex analytic diffeomorphism from $\Delta_x$ onto $\{u\in\mathbb{C}\mid|u|\le R_x\}$.
\end{lemma}

{
\begin{remark}\label{rem:approximate-disk}
An \emph{approximate disk} with centre $\tau$ and radius $R$ is an open simply connected set which, for some $\varepsilon>0$, contains the disk centred at $\tau$ with radius $R-\varepsilon$ and is contained in the disk centred at $\tau$ with radius $R+\varepsilon$.
\end{remark}
}

\begin{proof}
For the study of the solutions near $\mathcal{D}_x^*\setminus\mathcal{H}^*$, we use coordinates $(u_{x3},v_{x3})$; see Section \ref{app:x3x3}.
In this chart, the set $\mathcal{D}_x^*\setminus\mathcal{H}^*$  is given by $v_{x3}=0$ and parametrized by $u_{x3}\in\mathbb{C}$.
Moreover, $\mathcal{E}_x$ is given by $\{u_{x3}=0\}$ and parametrized by the variable $v_{x3}$. From Lemma \ref{lem:distancefn}, we recall that $d=-\omega_{x3}$ in this chart.

Asymptotically, for $v_{x3}\to0$, bounded $u_{x3}$, and $x=e^t$ bounded away from $0$ and $1$, we have:
\begin{subequations}
\begin{align}
\label{ux3dot}&\dot u_{x3}\sim \frac{1}{v_{x3}},\\
\label{vx3dot}&\dot v_{x3}\sim \left(\frac{2u_{x3}}{x(x-1)}-(\theta_0+\theta_x)-\frac{x\theta_1-\theta_x-1}{x-1}\right)v_{x3},\\
\label{omegax3}&\omega_{x3}\sim -x v_{x3},\\
\label{omegax3dot}&\frac{\dot \omega_{x3}}{\omega_{x3}}\sim(1-\theta_0-\theta_1)\frac{x}{x-1}+\frac{2u_{x3}+\theta_0+\theta_x }{x-1}, \\
\label{ex3omegax3}&E\omega_{x3}\sim -1-\frac{x}{u_{x3}}.%\\
%&\frac{\dot E}{E}\sim \frac{ (u_{x3}+x)\dot{}}{u_{x3}+x}-\frac{\dot u_{x3}}{u_{x3}}-x\frac{u_{x3}}{u_{x3}+x}
\end{align}
\end{subequations}
{Note that integrating Equation \eqref{omegax3dot} from $t_0$ to $t_1$, where $t_0, t_1\in D_t$ (see Figure \ref{fig:Dt}) leads to
\[
\log\left(\frac{\omega_{x3}(t_1)}{\omega_{x3}(t_0)}\right)\sim  (1-\theta_0-\theta_1)\log\left(\frac{1-e^{t_1}}{1-e^{t_0}}\right)+\int_{t_0}^{t_1}\frac{2u_{x3}+\theta_0+\theta_x }{e^t-1}\,dt.
\]
Therefore, if for all $t$ on the line segment from $t_0$ to $t_1$, we have $|e^{t}-e^{t_0}|\ll |e^{t_0}|$ and $|u_{x3}(t)|$ is bounded, then $\omega_{x3}(t)/\omega_{x3}(t_0)\sim\left((1-e^{t})/(1-e^{t_0})\right)^{1-\theta_0-\theta_1}$, where the right side is upper-bounded by $e^{t_0}$. In view of this situation, Equation \eqref{omegax3} shows that $v_{x3}$ is approximately given by a small constant. We take $t_0=\tau$ in the following analysis.}
From \eqref{ux3dot}, it follows that: 
$$
u_{x3}\sim u_{x3}(\tau)+\frac{t-\tau}{v_{x3}(\tau)}.
$$
Thus, if $t$ runs over an approximate disk $\Delta$ centred at $\tau$ with radius $|v_{x3}|R$, then $u_{x3}$ fills an approximate disk centred at $u_{x3}(\tau)$ with radius $R$.
Therefore, if $|v_{x3}|\ll1/|\tau|$, the solution has the following properties for $t\in\Delta$:
$$
\frac{v_{x3}(t)}{v_{x3}(\tau)}\sim1,
$$
and $u_{x3}$ is a complex analytic diffeomorphism from $\Delta$ onto an approximate disk with centre $u_{x3}(\tau)$ and radius $R$.
If $R$ is sufficiently large, we will have $0\in u_{x3}(\Delta)$, i.e.~the solution of the Painlev\'e equation will have a pole at a unique point in $\Delta$.

Now, it is possible to take $\tau$ to be the pole point.
We have:
$$
u_{x3}(t)\sim  \frac{t-\tau}{v_{x3}(\tau)}\sim -\frac{(t-\tau)e^{\tau}}{d(\tau)}.
$$
Let $R_x$ be a large positive real number.
Then the equation $|u_{x3}(t)|=R_x$ corresponds to $|t-\tau|\sim|e^{-\tau}d(\tau)|R_x$, which is still small compared to $|\tau|$ if $|d(\tau)|$ is sufficiently small.
It follows that the connected component $\Delta_x$  of the set of all $t\in\mathbb{C}$ such that $\{t\mid |u_{x3}(t)|\le R_x\}$ is an approximate disk with centre $\tau$ and radius $|d(\tau)e^{-\tau}|R_x$.
More precisely, $u_{x3}$ is a complex analytic diffeomorphism from $\Delta_x$ onto $\{u\in\mathbb{C}\mid |u|\le R_x\}$, and
$$
\frac{d(t)}{d(\tau)}\sim1
\quad\text{for all}\quad
t\in\Delta_x.
$$
 \end{proof}

\begin{remark} Similar arguments show that if a solution comes sufficiently close to $\mathcal{D}_1^*$ or $\mathcal{D}_\infty^{**}$, then it will cross the corresponding exceptional lines $\mathcal{E}_1$, respectively $\mathcal{E}_\infty$ and $\mathcal{E}_\infty^-$ transversally at a unique nearby value of time. We prove this in Appendix  \ref{app:EstD1Dinfty}. This is, however, not needed for our main  result.
\end{remark}

\begin{lemma}[Behaviour near $\mathcal{D}_0^*\setminus\mathcal{H}^*$]\label{lemma:nearD0}
If a solution at a complex time $t$ is sufficiently close to $\mathcal{D}_0^*\setminus\mathcal{H}^*$, then there exists unique $\tau\in\mathbf{C}$ such that $(u(\tau),v(\tau))$ belongs to the line $\mathcal{E}_0$.
In other words, $u$ vanishes and $v$ has a pole at $t=\tau$.
Moreover, $|t-\tau|=O(|d(t)||u_{03}(t)|)$ for sufficiently small $d(t)$ and bounded $u_{03}$.

For large $R_0>0$, consider the set $\{t\in\mathbb{C}\mid |u_{03}(t)|\le R_0\}$.
Its connected component containing $\tau$ is an approximate disk $\Delta_0$ with centre $\tau$ and radius $|d(\tau)(e^{\tau}+e^{-\tau})|R_0$, and
$t\mapsto u_{03}(t)$ is a complex analytic diffeomorphism from that $\Delta_0$ onto $\{u\in\mathbb{C}\mid|u|\le R_0\}$.
\end{lemma}

\begin{proof}
For the study of the solutions near $\mathcal{D}_0^*\setminus\mathcal{H}^*$, we use coordinates $(u_{03},v_{03})$; see Section \ref{app:0303}.
In this chart, the set $\mathcal{D}_0^*\setminus\mathcal{H}^*$  is given by $v_{03}=0$ and parametrized by $u_{03}\in\mathbb{C}$.
Moreover, $\mathcal{E}_0$ is given by $u_{03}=0$ and parametrized by $v_{03}$.

Asymptotically, for $v_{03}\to0$, bounded $u_{03}$, and $x=e^t$ bounded away from $0$ and $1$, we have:
\begin{subequations}
\begin{align}
\label{u03dot}&\dot u_{03}\sim -\frac{x}{(1-x)v_{03}},\\
\label{v03dot}&\dot v_{03}\sim -\frac{(x+1)(2u_{03}+\theta_0)-\theta_x+1-x\theta_1}{x-1}v_{03},\\
\label{omega03}&\omega_{03}=-v_{03},\\
\label{omega03dot}&\frac{\dot\omega_{03}}{\omega_{03}}
\sim
%\frac{(x+1)(2u_{03}+\theta_0)-\theta_x+1-x\theta_1}{x-1}
%=
2u_{03}+\theta_0-\theta_1+\frac{4u_{03}+\theta_0-\theta_1+1}{x-1},
\\
\label{e03omega03}&E\omega_{03}\sim -\frac{x}{x-1}.
\end{align}
\end{subequations}

Arguments similar to those in the proof of Lemma \ref{lemma:nearDx} show that $v_{03}$ is approximately equal to a small constant, and from \eqref{u03dot} it follows that: 
$$
u_{03}\sim u_{03}(\tau)-\frac{e^t-e^{\tau}}{v_{03}(\tau)}.
$$
Thus, if $t$ runs over an approximate disk $\Delta$ centred at $\tau$ with radius $|v_{03}|\log R$, then $u_{03}$ fills an approximate disk centred at $u_{03}(\tau)$ with radius $R$.
Therefore, if $|v_{03}|\ll e^{-|\tau|}$, the solution has the following properties for $t\in\Delta$:
$$
\frac{v_{03}(t)}{v_{03}(\tau)}\sim1,
$$
and $u_{03}$ is a complex analytic diffeomorphism from $\Delta$ onto an approximate disk with centre $u_{03}(\tau)$ and radius $R$.
If $R$ is sufficiently large, we will have $0\in u_{03}(\Delta)$, i.e.~the solution of the Painlev\'e equation will vanish at a unique point in $\Delta$.

Now, it is possible to take $\tau$ to be that point.
We have:
$$
u_{03}(t)\sim - \frac{e^t-e^{\tau}}{v_{03}(\tau)}\sim -\frac{(e^t-e^\tau)e^{\tau}}{(e^{\tau}-1)d(\tau)}.
$$
Let $R_0$ be a large positive real number.
Then the equation $|u_{03}(t)|=R_0$ corresponds to $|1-e^{t-\tau}|\sim|e^{-2\tau}(e^{\tau}-1)d(\tau)|R_0$, which is still small compared to $|e^{\tau}|$ if $|d(\tau)|$ is sufficiently small.
It follows that the connected component $\Delta_0$  of the set of all $t\in\mathbb{C}$ such that $\{t\mid |u_{03}(t)|\le R_0\}$ is an approximate disk with centre $\tau$ and radius $|d(\tau)(e^{-\tau}+e^{\tau})|R_0$.
More precisely, $u_{03}$ is a complex analytic diffeomorphism from $\Delta_0$ onto $\{u\in\mathbb{C}\mid |u|\le R_0\}$, and
$$
\frac{d(t)}{d(\tau)}\sim1
\quad\text{for all}\quad
t\in\Delta_0.
$$

\end{proof}

\begin{theorem}\label{th:estimates}
Let $\epsilon_1$, $\epsilon_2$, $\epsilon_3$ be given such that $0<\epsilon_1<1$, $0<\epsilon_2<1$, $0<\epsilon_3<1$.
Then there exists $\delta>0$ such that if $|e^{t_0}|<\epsilon_1$ and $|d(t_0)|<\delta$, it follows that
$$
\rho=\inf\{r<|e^{t_0}|\ \text{such that}\ |d(t)|<\delta\ \text{whenever}\ |e^{t_0}|\ge|e^t|\ge r\}
$$
satisfies:
\begin{itemize}
 \item[(i)] $\rho>0$ and is bounded  below by the relation:
  $$\delta\ge|d(t_0)|\left((1-\rho)/|1-e^{t_0}|\right)^{1-\epsilon_2}(1-\epsilon_3);$$
 \item[(ii)] if $|e^{t_0}|\ge|e^t|\ge \rho$ then 
 $$
 d(t)=d(t_0)\left(\frac{1-e^{t}}{1-e^{t_0}}\right)^{1+\varepsilon_2(t)}(1+\varepsilon_3(t)),
 $$
 where $|\varepsilon_2(t)|\le\epsilon_2$ and $|\varepsilon_3(t)|\le\epsilon_3$; and,
 \item[(iii)] if $|e^t|$ is less than $\rho$, but still sufficiently close to $\rho$, then $|d(t)|\ge\delta(1-\epsilon_3)$.
\end{itemize}
\end{theorem}

\begin{proof}
Suppose a solution of the system \eqref{eq:vectorfield} is close to the infinity set at times $t_0$ and $t_1$.
If follows from Lemmas \ref{lemma:nearDx} and \ref{lemma:nearD0} that for every solution close to $\mathcal{I}$, the set of complex times $t$ such that the solution is not close to $\mathcal{I}\setminus (\mathcal{D}_0^*\cup \mathcal{D}_x^*)$ is the union of approximate disks of radius $\sim|d|$.
Hence if the solution is near $\mathcal{I}$ for all complex times $t$ such that
$|e^{t_0}|\ge|e^t|\ge|e^{t_1}|$, 
then there exists a path $\mathcal{P}$ from $t_0$ to $t_1$, such that the solution is close to $\mathcal{I}\setminus (\mathcal{D}_0^*\cup \mathcal{D}_x^*)$ for all $t\in\mathcal{P}$ and $\mathcal{P}$ is $C^1$-close to the path: $s\mapsto t_1^st_0^{1-s}$, $s\in[0,1]$.

Then Lemma \ref{lemmaDinfty} implies that near $\mathcal{D}_{\infty}^{**}$:
$$
\frac{E(t)}{E(t_0)}\sim \frac{1-e^{t_0}}{1-e^t},
$$
and by using Lemma \ref{lem:distancefn} we find:
\begin{equation}\label{eq:d-estimate}
d(t)\sim d(t_0)\frac{1-e^{t}}{1-e^{t_0}}.
\end{equation}

For the first statement of the theorem, we have:
$$
\delta>|d(t)|
\ge
|d(t_0)|\left(\frac{1-|e^{t}|}{|1-e^{t_0}|}\right)^{1-\epsilon_2}(1-\epsilon_3),
$$
and the desired result follows from $\rho\le e^t$.
For $|e^t|\le |e^{t_0}|$, the second statement follows from \eqref{eq:d-estimate} and the third one from the  definition of $\rho$.

The symmetries of the sixth Painlev\'e equation show that the same statments follow near other lines of the infinity set $\mathcal I$.
\end{proof}

As a consequence of Theorem \ref{th:estimates}, we can prove the repelling property of the set $\mathcal{I}$.

\begin{corollary}\label{cor:repeller}
No solution with the initial conditions in the space of the initial values intersects $\mathcal{I}$.
A solution that is close to $\mathcal{I}$ for a certain value of the independent variable $t$ will stay in the vicinity of $\mathcal{I}$ only for a limited range of $t$.
Moreover, if a solution is sufficiently close to $\mathcal{I}$ at a point $t$, then it will have a pole in a neighbourhood of $t$.
\end{corollary}
\begin{proof}
The statement follows from Theorem \ref{th:estimates} and  Lemmas \ref{lemma:nearDx}, \ref{lemma:nearD0}.
\end{proof}
\begin{remark}
Parts (i) and (ii) of Theorem \ref{th:estimates} give estimates on the behaviour of the solutions near the infinity set.
Part (iii) implies that a solution does not stay indefinitely near the infinity set as $e^t\to0$.
\end{remark}

\section{The limit set}\label{s:limit}

Our definition of the limit set is the extension of the standard concept of limit sets in dynamical systems to complex-valued solutions.
\begin{definition}
	Let $(u(t),v(t))$ be a solution of (\ref{eq:vectorfield}). \emph{The limit set} $\Omega_{u,v}$ of $(u(t),v(t))$ is the set of all $s\in\mathcal{S}(0)\setminus\mathcal{I}(0)$ such that there exists a sequence $t_n\in\mathbf{C}$ satisfying:
	$$
	\lim_{n\to\infty}\Re(t_n)=-\infty
	\quad\text{and}\quad
	\lim_{n\to\infty}(u(t_n),z(v_n))=s.
	$$
\end{definition}

\begin{theorem}\label{th:limit-set}
There exists a compact set $K\subset\mathcal{S}(0)\setminus\mathcal{I}(0)$ such that the limit set $\Omega_{u,v}$ of any solution $(u,v)$ of (\ref{eq:vectorfield}) is contained in $K$.
Moreover, $\Omega_{u,v}$ is a nonempty, compact and connected set, which is invariant under the flow of the autonomous system given in Section \ref{s:S0}.
\end{theorem}

\begin{proof}
	For any positive numbers $\eta$, $r$ let $K_{\eta,r}$ denote the set of all $s\in\mathcal{S}(x)$ such that $|x|\le r$ and $|d(s)|\ge\eta$.
	Since $\mathcal{S}(x)$ is a complex analytic family over $\mathbb{P}^1\setminus\{0,1\}$ of compact surfaces $\mathcal{S}(x)$, $K_{\eta,r}$ is also compact.
	Furthermore, $K_{\eta,r}$ is a compact subset of the Okamoto's space $\mathcal{S}\setminus\mathcal{I}(0)$.
	When $r$ approaches $0$, the sets $K_{\eta,r}$ shrink to the compact set:
	$$
	K_{\eta,0}
	=
	\{s\in\mathcal{S}(0)\mid |d(s)\ge\eta \}
	\subset\mathcal{S}(0)\setminus\mathcal{I}(0).
	$$
If follows from Theorem \ref{th:estimates} that there is $\eta>0$ such that for every solution $(u,v)$ there exists $r_0>0$ with the following property:
$$
(u(t),v(t))\in K_{\eta,r_0}
\text{ for every }
t
\text{ such that }
|e^t|\le r_0.
$$
Hereafter, we take $r\le r_0$, when it follows that $(u(t),v(t))\in K_{\eta,r}$ whenever $|e^t|\le r$.

	Let $T_r=\{t\in\mathbb{C}\mid |e^t|\le r\}$, and let $\Omega_{(u,v),r}$ denote the closure of the image set $\bigl(u(T_r), v(T_r)\bigr)$ in $\mathcal{S}$.
Since $T_r$ is connected and $(u,v)$ is continuous, $\Omega_{(u,v),r}$ is also connected.
Since $\bigl(u(T_r), v(T_r)\bigr)$ is contained in the compact set $K_{\eta, r}$, its closure $\Omega_{(u,v),r}$ is also contained in $K_{\eta, r}$, 
and therefore $\Omega_{(u,v),r}$ is a nonempty compact and connected subset of $\mathcal{S}\setminus\mathcal{S}(0)$.

The intersection of a decreasing sequence of nonempty, compact, and connected sets is a nonempty, compact, and connected. Therefore, as $\Omega_{(u,v),r}$ decreases to $\Omega_{(u,v)}$ as $r$ approaches zero, it follows that $\Omega_{(u,v)}$ is a nonempty, compact and connected subset of $\mathcal{S}$.
Since $\Omega_{(u,v),r}\subset K_{\eta,r}$, for all $r\le r_0$, and the sets $K_{\eta,r}$ shrink to the compact subset $K_{\eta,0}$ of $\mathcal{S}(0)\setminus\mathcal{I}(0)$ as $r$ decreases to zero, it follows that $\Omega_{(u,v)}\subset K_{\eta,0}$.
This proves the first statement of the theorem with $K=K_{\eta,0}$.

Since $\Omega_{(u,v)}$ is the intersection of the decreasing family of compact sets $\Omega_{(u,v),r}$, there exists for every neighbourhood $A$ of $\Omega_{(u,v)}$ in $\mathcal{S}$, an $r>0$ such that $\Omega_{(u,v),r}\subset A$. Hence $(u(t),v(t))\in A$ for every $t\in\mathbb{C}$ such that $|e^t|\le r$.
If $\{t_j\}$ is any sequence in $\mathbb{C}\setminus\{0\}$ such that $|t_j|\to0$, then the compactness of $K_{\eta,r}$, in combination with 
$\bigl(u(T_r), v(T_r)\bigr)\subset K_{\eta,r}$, implies that there is a subsequence $j=j(k)\to\infty$ as $k\to\infty$ and an $s\in K_{\eta,r}$, such that:
$$
(u(t_{j(k)}),v(t_{j(k)}))\to s\ \text{as}\ k\to\infty.
$$
It follows, therefore, that $s\in\Omega_{(u,v)}$.

Next, we prove that $\Omega_{(u,v)}$ is invariant under the flow $\Phi^{\tau}$ of the autonomous Hamiltonian system.
Let $s\in\Omega_{(u,v)}$ and $t_j$ be a sequence in $\mathbb{C}\setminus\{0\}$ such that $e^{t_j}\to0$ and $(u(t_j),v(t_j))\to s$.
Since the $t$-dependent vector field of the Painlev\'e system converges in $C^1$ to the vector field of the autonomous Hamiltonian system as $e^t\to0$, it follows from the continuous dependence on initial data and parameters, that the distance between $(u(t_j+\tau),v(t_j+\tau))$ and $\Phi^{\tau}(u(t_j),v(t_j))$ converges to zero as $j\to\infty$.
Since $\Phi^{\tau}(u(t_j),v(t_j))\to\Phi^{\tau}(s)$ and $|e^t_j|\to0$ as $j\to\infty$, it follows that $(u(t_j+\tau),v(t_j+\tau))\to\Phi^{\tau}(s)$ and $e^{t_j+\tau}\to0$ as $j\to\infty$, hence $\Phi^{\tau}(s)\in\Omega_{(u,v)}$.
\end{proof}

\begin{proposition}\label{prop:pole-lines}
Every solution $(u(t),v(t))$ with the essential singularity at $x=0$ intersects each of the exceptional lines $\mathcal{E}_0$, $\mathcal{E}_x$, $\mathcal{E}_1$, $\mathcal{E}_{\infty}$, $\mathcal{E}_{\infty}^-$ infinitely many times in any neighbourhood of that singular point.
\end{proposition}
\begin{proof}
  For conciseness, we refer to the solution $(u(t),v(t))$ of the system as the Painlev\'e vector field and denote the vector field near each of the five exceptional lines $\mathcal{E}_0$, $\mathcal{E}_x$, $\mathcal{E}_1$, $\mathcal{E}_{\infty}$, $\mathcal{E}_{\infty}^-$ by $(U(t), V(t))$. Furthermore, let
  \[
\mathscr E=\mathcal{E}_0\cup\mathcal{E}_x\cup\mathcal{E}_1\cup\mathcal{E}_{\infty}\cup\mathcal{E}_{\infty}^-.
    \]

    Now suppose that $(U(t),V(t))$ intersects $\mathscr E$ only finitely many times.
    According to Theorem \ref{th:limit-set}, the limit set $\Omega_{(u,v)}$ is a compact set in $\mathcal{S}(0)\setminus\mathcal{I}(0)$. 
    
If $\Omega_{(u,v)}$ intersects one the five exceptional lines $\mathcal{E}_0$, $\mathcal{E}_x$, $\mathcal{E}_1$, $\mathcal{E}_{\infty}$, $\mathcal{E}_{\infty}^-$ at a point $p$, then there exists a $t$ such that $e^t$ is arbitrarily close to zero and the Painlev\'e vector field is arbitrarily close to $p$,
when the transversality of the vector field to the exceptional line implies that $(U(\tau),V(\tau))\in \mathscr E$
for a unique $\tau$ ($\not= t$) near $t$. This is a contradiction to our assumption, as it follows that $(U(t),V(t))$ intersects $\mathscr E$ infinitely many times. Therefore,  we must have that $\Omega_{(u,v)}$ is a compact subset of 
$\mathcal{S}(0)
\setminus
(\mathcal{I}(0)\cup\mathscr E)
$.

However, $\mathscr E$ is equal to the set of all points in 
$\mathcal{S}(0)\setminus\mathcal{I}(0)$, which project (blow-down) to the line $\mathcal D_\infty\cup\mathcal{H}^*$, and therefore 
$\mathcal{S}(0)
\setminus
(\mathcal{I}(0)\cup\mathscr E)
$
is the affine $(u,v)$-coordinate chart, of which $\Omega_{(u,v)}$ is a compact subset, which implies that $u(t)$ and $v(t)$ remain bounded for small $|e^t|$.
%It follows from the boundedness of $u$ and $v$ that $u(t)$ and $v(t)$ are holomorphic functions of $x=e^t$ in a neighbourhood of $x=0$, which implies that there are complex numbers $u_0$, $v_0$ which are the limit points of $u(t)$ and $v(t)$ as 
$|e^t|\to0$. 
From there, $x=0$ is not an essential singularity.
\end{proof}
	 
\begin{theorem}\label{th:poles-zeroes-ones}
Every solution of the sixth Painlev\'e equation has infinitely many poles, infinitely many zeroes, and infinitely many times takes value $1$ in any neighbourhood of its essential singularity.
\end{theorem}
\begin{proof}
At the intersection points with $\mathcal{E}_0$, $\mathcal{E}_1$, $\mathcal{E}_{\infty}$, $\mathcal{E}_{\infty}^-$, the solution will have zeroes, 1s, and poles, as explained in detail in Section \ref{s:movable}.
Thus, the statement for an essential singularity at $x=0$ follows from Proposition \ref{prop:pole-lines}.

{
If $y(x)$ is a solution of \eqref{PVI},
observe that the following B\"acklund transformations:
\begin{gather*}
\mathcal{S}_1: y_1(x_1)=\frac{y}x,
\quad
x_1=\frac1x,
\quad
(\theta_{\infty,1},\theta_{0,1},\theta_{1,1},\theta_{x,1})
=
\left(\theta_{\infty},\theta_0,\sqrt{\theta_x^2+\frac12},\sqrt{\theta_1^2+\frac12}\right)
\\
\mathcal{S}_2: y_2(x_2)=1-y,
\quad
x_2=1-x,
\quad
(\theta_{\infty,2},\theta_{0,2},\theta_{1,2},\theta_{x,2})
=
\left(\theta_{\infty},i\theta_1,i\theta_0,\theta_x\right)
\end{gather*}
give the soltions $y_1(x_1)$ and $y_2(x_2)$ of the sixth Painlev\'e equation with respective parameters 
$(\theta_{\infty,1},\theta_{0,1},\theta_{1,1},\theta_{x,1})$ and $(\theta_{\infty,2},\theta_{0,2},\theta_{1,2},\theta_{x,2})$, \cite[\S32.7(vii)]{DLMF}.
Transformation $\mathcal{S}_1$ maps point $x=\infty$ to $x_1=0$, while $\mathcal{S}_2$ maps point $x=1$ to $x_2=0$, thus the statement will also hold for essential singularities at $x=1$ and $x=\infty$.
}
\end{proof}

\section{Conclusion}\label{s:con} 

The Painlev\'e equations have been playing an increasingly important role in mathematical physics, especially in the applications to classical and quantum integrable systems and random matrix theory.
The sixth Painlev\'e equation, which is the focus of this work, is very prominent in these areas, in particular, in conformal field theory in recent times \cite{GIL2012}.
For further relations with conformal block expansions and supersymmetric gauge theories, see the references in \cite{GIL2012}.

Although the initial values space for the Painlev\'e equations was described by Okamoto \cite{Okamoto}, our aim in this work was to describe the dynamics of the solutions by analysing that construction.

Many questions beyond the limit behaviour remain open about particular families of transcendental solutions, from the dynamical systems point of view.
For example, the existence of limit cycles of transcendental solutions with particular symmetry properties and whether there are periodic cycles in the combined space of parameters and initial values remain open.

\subsection*{Acknowledgments}
The authors are grateful to the referee for their careful reading of our paper and useful suggestions for its improvement.
\appendix
 \section{Charts of the initial surface $\mathbb{F}_1$}\label{app:initial}
\subsubsection{Initial chart $(u_0,v_0)=(u,v)$}\label{app:00}

\[\begin{array}{l}
\left\{\begin{array}{rcl}
(u,v)&=&(u_0,v_0)\vspace{.2cm}\\
\omega_0 &=& 1\vspace{.2cm}\\
E&=&\left.\frac{1}{x-1}\right[u_0(u_0-x)(u_0-1)v_0^2-(\theta+\overline{\theta})u_0^2v_0+\theta\overline{\theta}(u_0-x)-x\theta_0v_0+\\&&\quad \quad\quad\quad \quad \quad\quad \quad \Bigl. +\Bigl((x+1)\theta_0+x\theta_1 +(\theta_x-1)\Bigr)u_0v_0\Bigr]\vspace{.2cm}\\
\dot E&=&-\frac{x(u_0-1)}{(x-1)^2}\Bigl[(u_0v_0-\theta)(u_0v_0-\overline{\theta})-(u_0v_0-\theta_0)v_0\Bigr]\\
% E|_{x=0}&=&-u_0\left[u_0(u_0-1)v_0^2-\left( \theta_1+(\theta+\overline{\theta})(u_0-1)\right)v_0+\theta\overline{\theta}\right]\\
% E&=&\frac{1}{1-x}E|_{x=0}+\frac{x}{1-x}\left[(u_0^2(u_0-1)v_0^2 -\left(\theta_0(u_0-1)+\theta_1u_0\right)v_0+\theta\overline{\theta})\right]\vspace{.2cm}\\

\end{array}\right.
\vspace{.5cm}\\
\left\{\begin{array}{rcl}
\dot u_0 &=&-2\frac{u_0(u_0-x)(u_0-1)v_0}{1-x}+\frac{\theta+\overline{\theta}}{1-x}u_0(u_0-1)+ \theta_1  u_0-\frac{x}{1-x}  \theta_0(u_0-1) 
\vspace{.2cm}\\
\dot v_0 &=&\frac{(3u_0 -2)u_0}{1-x}v_0^2-\left[\frac{\theta+\overline{\theta}}{1-x}(2u_0-1)+\theta_1 \right]v_0+\frac{\theta \overline{\theta}}{1-x}-\frac{x}{1-x}\left[(2u_0-1)v_0-\theta_0\right]v_0 \vspace{.2cm}\\
\frac{\dot \omega_0}{\omega_0}&=&0
\end{array}\right.\end{array}
\]
No base points. \\
No elliptic base points. \\
No visible components of the infinity set.
%Equilibrium points for $x=0$ : $$\begin{array}{lcl}(u_0,v_0)=\left(0\, ,\, -1-\frac{\theta_1}{\theta+\overline{\theta}-\theta_1}\right)&\leadsto &E|_{x=0}=0 \vspace{.2cm}\\  (u_0,v_0)=\left(1-\frac{\theta_1}{\theta+\overline{\theta}}\, ,\, \overline{\theta}+\frac{\theta_1\overline{\theta}}{\theta-\overline{\theta}-\theta_1}\right)&\leadsto &E|_{x=0}=\overline{\theta}(\theta_1-\theta)\vspace{.2cm}\\ (u_0,v_0)=\left(1+\frac{\theta_1}{\theta+\overline{\theta}}\, ,\, {\theta}-\frac{\theta_1 {\theta}}{\theta-\overline{\theta}+\theta_1}\right) &\leadsto &E|_{x=0}=\theta(\theta_1-\overline{\theta}) \end{array} $$

\subsubsection{First chart $(u_1,v_1)=\left(u,\frac{1}{v}\right)$}\label{app:11}

$$\begin{array}{l}\left\{\begin{array}{rcl}
(u,v)&=&\left(u_1,\frac{1}{v_1}\right)\vspace{.2cm}\\
\omega_1&=&-v_1^2\vspace{.2cm}\\
E&=&\left.\frac{1}{(x-1)v_1^2}\right[u_1(u_1-x)(u_1-1) -(\theta+\overline{\theta})u_1^2v_1+\theta\overline{\theta}(u_1-x)v_1^2-x\theta_0v_1+\\&&\quad \quad\quad\quad \quad \quad\quad \quad \Bigl. +\Bigl((x+1)\theta_0+x\theta_1 +(\theta_x-1)\Bigr)u_1v_1\Bigr]\vspace{.2cm}\\
\dot E&=&-\frac{x(u_1-1)}{(x-1)^2v_1^2}\Bigl[(u_1 -\theta v_1)(u_1 -\overline{\theta}v_1)-(u_1-\theta_0v_1) \Bigr]\\
% E|_{x=0}&=&-\frac{1}{v_1^2}\left[\theta\overline{\theta}v_1^2-\theta_1u_1v_1+(u_1-1)(\theta v_1-u_1)(\overline{\theta}v_1-u_1)  \right]\vspace{.2cm}\\&=&-\frac{u_1}{v_1^2}\left[u_1(u_1-1)-\left( \theta_1+(\theta+\overline{\theta})(u_1-1)\right)v_1+\theta\overline{\theta}v_1^2\right]\vspace{.2cm}\
\end{array}\right.
\vspace{.5cm}\\
\left\{\begin{array}{rcl}
\dot u_1 &=& -2\frac{u_1(u_1-x)(u_1-1)}{(1-x)v_1}+\frac{\theta+\overline{\theta}}{1-x}u_1(u_1-1)+\frac{\theta_1}{1-x} u_1-\frac{x}{1-x}\left[ \theta_0(u_1-1)+\theta_1u_1 \right]   
\vspace{.2cm}\\
\dot v_1 &=& -\frac{\theta \overline{\theta}}{1-x}v_1^2+\frac{1}{1-x}\left[ (\theta+\overline{\theta})(2u_1-1)+\theta_1\right]v_1-\frac{(3u_1-2)u_1}{1-x}-\frac{x}{1-x}\left[ (\theta_0+\theta_1)v_1-(2u_1-1)\right]
\vspace{.2cm}\\
\frac{\dot \omega_1}{\omega_1}&=& -2\frac{\theta \overline{\theta}}{1-x}v_1+2\frac{1}{1-x}\left[ (\theta+\overline{\theta})(2u_1-1)+\theta_1\right]-2\frac{(3u_1-2)u_1}{(1-x)v_1}-2\frac{x}{1-x}\left[ (\theta_0+\theta_1)-\frac{2u_1-1}{v_1})\right] .
\end{array}\right.
\end{array}$$
Base points of the vector field:
$$b_0 : (u_1,v_1)=(0,0)$$
$$b_x : (u_1,v_1)=(x,0)$$
$$b_1 : (u_1,v_1)=(1,0)$$
Elliptic base points are $b_0$ and $b_1$.\\
Visible components of the infinity set:  
{ $ H: \{v_{1}=0\} $}
\\
\textbf{Estimates near $H$, i.e. $ v_{1}\longrightarrow 0$}:\\
\begin{eqnarray}
\omega &= &  -v_{1}^2
\\
\omega E  & \sim &\frac{u_1(u_1-1)(u_1-x)}{1-x}\\
\frac{\dot E}{E} & \sim & \frac{x}{1-x}-\frac{x}{u_{1}-x}\\
\frac{\dot{\omega}}{\omega} & \sim &\left(\frac{1}{u_1}+\frac{1}{u_1-x}+\frac{1}{u_1-1}\right)\dot{u}_1\\
\dot{u}_{1}& \sim & 2\frac{u_1(u_1-1)(u_1-x)}{(x-1)v_1}\\
\dot{v}_{1}& \sim & \frac{3u_1^2-2(1+x)u_1+x}{x-1}\\
\omega E_0  & \sim & u_1^2(u_1-1)\\
\frac{\dot E_0}{E_0} & \sim & \left(\frac{1}{u}-\frac{1}{u-x} \right)\dot{u}_1
\end{eqnarray}
%\clearpage 

\subsubsection{Second chart $(u_2,v_2)=\left(\frac{1}{u},\frac{1}{uv}\right)$}\label{app:22}

$$
\begin{array}{l}\left\{\begin{array}{rcl}
		(u,v)&=&\left(\frac{1}{u_2},\frac{u_2}{v_2}\right)\vspace{.2cm}\\
		\omega_2&=&u_2v_2^2\vspace{.2cm}\\
		E&=&\left.\frac{1}{(x-1)v_2^2}\right[\frac{(xu_2-1)(u_2-1)-(\theta+\overline{\theta})v_2}{u_2}  -\theta\overline{\theta}\frac{xu_2-1}{u_2}v_2^2 -x\theta_0u_2v_2+\\&&\quad \quad\quad\quad \quad \quad\quad \quad \Bigl. +\Bigl((x+1)\theta_0+x\theta_1 +(\theta_x-1)\Bigr)v_2\Bigr]\vspace{.2cm}\\
		\dot E&=&\frac{x(u_2-1)}{(x-1)^2v_2^2}\Bigl[\frac{(\theta v_2-1)(\overline{\theta}v_2-1)}{u_2} +\theta_0v_2-1)\Bigr]\\ 
		%E|_{x=0}&=&\frac{ (u_2-1)(\theta v_2-1)(\overline{\theta}v_2-1)}{u_2v_2^2} +\frac{\theta_1}{v_2}-\theta\overline{\theta}
	\end{array}\right.
	\vspace{.5cm}\\
	\left\{\begin{array}{rcl}
		\dot u_2 &=& \frac{(u_2-1)((\theta+\overline{\theta})v_2-2)}{(1-x)v_2}- \theta_1 u_2-\frac{x}{(1-x)v_2}u_2(u_2-1)(\theta_0v_2-2)\vspace{.2cm}\\
		\dot v_2 &=&-\frac{1}{(1-x)u_2}(\theta v_2-1)(\overline{\theta}v_2-1)-\frac{x}{1-x}(\theta_0v_2-1)u_2 \vspace{.2cm}\\
		\frac{\dot \omega_2}{\omega_2}&=& \frac{(\theta+\overline{\theta})(u_2+1)-2\theta \overline{\theta}v_2}{(1-x)u_2}-\theta_1+2\frac{-1+x(2u_2-1)}{(1-x)v_2}-\frac{x\theta_0(3u_2-1)}{1-x}
	\end{array}\right.\end{array}$$
No new base points.\\ 
Other visible base points:
$$b_x : (u_2,v_2)=\left(\frac{1}{x}, 0\right)$$
$$ b_1 : (u_2,v_2)=\left(1,0\right)$$
$$b_\infty : (u_2,v_2)=\left(0, \frac{1}{\theta}\right)$$
$$b_\infty^- : (u_2,v_2)=\left(0, \frac{1}{\overline{\theta}}\right)$$
No new elliptic base points.\\%No new equilibrium points. 
Visible components of the infinity set: 
{ $H: \{v_2=0\}\, , \quad  D_\infty: \{u_{2}=0\} $}
\\
%\clearpage
\textbf{Estimates near $D_\infty$, i.e. $ u_{2}\longrightarrow 0$}:\\
\begin{eqnarray}
	\omega &= & u_2v_2^2
	\\
	\omega E  & \sim &\frac{(v_3\theta-1)(v_3\overline{\theta}-1)}{x-1}\\
	\frac{\dot E}{E} & \sim & \frac{x}{1-x} \\
	\frac{\dot{\omega}}{\omega} & \sim &\left(\frac{1}{v_2-\frac{1}{\theta}}+\frac{1}{v_2-\frac{1}{\overline{\theta}}} \right)\dot{v}_{2}\\
	\dot{u}_{2}& \sim & \frac{\theta+\overline{\theta}}{x-1}-\frac{2}{(x-1)v_2}\\
	\dot{v}_{2}& \sim &\frac{(v_2\theta-1)(v_2\overline{\theta}-1)}{(x-1)u_2}\\
	\omega E_0  & \sim &-(v_2\theta-1)(v_2\overline{\theta}-1)\\
	\frac{\dot E_0}{E_0} & \sim & x\frac{\theta_0+\theta_1 }{x-1}-\frac{2x}{(x-1)v_2}
\end{eqnarray}
\\
\\
\textbf{Estimates near $H$, i.e. $ v_{2}\longrightarrow 0$}:\\
\begin{eqnarray}
	\omega &= & u_2v_2^2
	\\
	\omega E  & \sim &\frac{(u_2-1)(xu_2-1)}{x-1}\\
	\frac{\dot E}{E} & \sim & \frac{1}{1-x}+\frac{1}{xu_2-1} \\
	\frac{\dot{\omega}}{\omega} & \sim &\left(\frac{1}{u_2-1}+\frac{x}{xu_2-1} \right)\dot{u}_{2}\\
	\dot{u}_{2}& \sim & \frac{2(u_2-1)(xu_2-1)}{(1-x)v_2}\\
	\dot{v}_{2}& \sim &-\frac{xu_2}{x-1}+\frac{1}{(x-1)u_2}\\
	\omega E_0  & \sim &u_2-1\\
	\frac{\dot E_0}{E_0} & \sim & -\frac{x\dot{u}_2}{xu_2-1} 
\end{eqnarray}
%\clearpage 

\subsubsection{Third chart $(u_3,v_3)=\left(\frac{1}{u},uv\right)$}\label{app:33}

$$
\begin{array}{l}\left\{\begin{array}{rcl}
(u,v)&=&\left(\frac{1}{u_3}, u_3{v_3} \right)\vspace{.2cm}\\
\omega_3&=&-u_3\vspace{.2cm}\\
E&=&\left.\frac{1}{x-1}\right[\frac{(xu_3-1)(u_3-1)v_3-(\theta+\overline{\theta})}{u_3}v_3  -\theta\overline{\theta}\frac{xu_3-1}{u_3} -x\theta_0u_3v_3+\\&&\quad \quad\quad\quad \quad \quad\quad \quad \Bigl. +\Bigl((x+1)\theta_0+x\theta_1 +(\theta_x-1)\Bigr)v_3\Bigr]\vspace{.2cm}\\
\dot E&=&\frac{x(u_3-1)}{(x-1)^2}\Bigl[\frac{(v_3-\theta)(v_3-\overline{\theta})}{u_3}-(v_3-\theta_0)v_3\Bigr]\\
%E|_{x=0}&=&\frac{(u_3-1)(v_3-\theta)(v_3-\overline{\theta})}{u_3}+(\theta_1v_3-\theta\overline{\theta}) 
\end{array}\right.
\vspace{.5cm}\\
\left\{\begin{array}{rcl}
\dot u_3 &=& -\frac{(u_3-1)(2v_3-(\theta+\overline{\theta}))}{1-x}- \theta_1 u_3+\frac{x}{1-x}u_3(u_3-1)(2v_3-\theta_0)\vspace{.2cm}\\
\dot v_3 &=&\frac{1}{(1-x)u_3}(v_3-\theta)(v_3-\overline{\theta})-x\frac{(v_3-\theta_0)u_3v_3}{1-x}
\vspace{.2cm}\\
\frac{\dot \omega_3}{\omega_3}&=&   -\frac{(u_3-1)(2v_3-(\theta+\overline{\theta}))}{(1-x)u_3}- \theta_1 +\frac{x}{1-x}(u_3-1)(2v_3-\theta_0)\end{array}\right.
\end{array}$$
New base points: 
$$b_\infty : (u_3,v_3)=\left(0,\theta \right)$$
$$b_\infty^- : (u_3,v_3)=\left(0,\overline{\theta}\right)$$
No other visible base points.\\
New elliptic base points are $b_\infty$ and $b_\infty^-$. \\
Visible components of the infinity set:  
{ $ D_\infty: \{u_{3}=0\} $}
%No new equilibrium points for $x=0$. The two of them which are visible are: $$\begin{array}{lcl}(u_3,v_3)=\left(1-\frac{\theta_1}{\theta_1+\theta-\overline{\theta}}\, , \, \theta\right)&\leadsto &E|_{x=0}=\theta(\theta_1-\overline{\theta}) \vspace{.2cm}\\  (u_3,v_3)=\left(1-\frac{\theta_1}{\theta_1-\theta+\overline{\theta}}\, , \,\overline{ \theta}\right)&\leadsto &E|_{x=0}=\overline{\theta}(\theta_1-\theta)  \end{array} $$
\\
\\
\textbf{Estimates near $D_\infty$, i.e. $ u_{3}\longrightarrow 0$}:\\
\begin{eqnarray}
\omega &= &  -u_3
\\
\omega E  & \sim &\frac{(v_3-\theta)(v_3-\overline{\theta})}{1-x}\\
\frac{\dot E}{E} & \sim & \frac{x}{1-x} \\
\frac{\dot{\omega}}{\omega} & \sim &\frac{1}{u_3}\dot{u}_{3}\\
\dot{u}_{3}& \sim & \frac{2v_3-(\theta+\overline{\theta})}{1-x}\\
\dot{v}_{3}& \sim &\frac{(v_3-\theta)(v_3-\overline{\theta})}{(1-x)u_3}\\
\omega E_0  & \sim & (v_3-\theta)(v_3-\overline{\theta})\\
\frac{\dot E_0}{E_0} & \sim & -\frac{x(\theta_x-1)}{x-1}+x\dot{u}_3
\end{eqnarray}
%\clearpage 

 \section{Okamoto desingularization}\label{app:resolution}
 \subsection{Details for the blow-up procedure}

\subsubsection{Blow up of $\beta_0,\beta_1,\beta_x,\beta_\infty$}\label{s:DetailsA}$ $\\
Let us first blow up the points $\beta_0,\beta_1,\beta_x$. 
\begin{itemize}
 \item[$\star$] Recall that for $i\in \{0,1,x \}$ we have $ \mathcal{V}_i:\{u_0=i\}\cup\{u_1=i\}\cup\{u_2=1/i\}\cup \{u_3=1/i\}$ and that $\beta_i\in \mathcal{V}_i$. Note further that $\mathcal{V}_i\setminus \{\beta_i\} \subset \C^2_{u,v}$. So whenever we remove one of the points $\beta_i$ from a chart other than 
$\C^2_{u,v}$, we may just as well remove the visible part of the whole line $\mathcal{V}_i$, without changing the global picture.
 \item[$\star$] Replace the chart $\C^2_{u_1,v_1}$ by the following six $\C^2$-charts:
 $$\begin{array}{rclcrcl}(u_{01},v_{01})&:=&\left(u_1, \frac{v_1}{u_1} \right) 
&\quad 
&(u_{02},v_{02})&:=&\left(\frac{u_1}{v_1}, {v_1}\right)\vspace{.2cm}\\
 (u_{11},v_{11})&:=&\left(u_1-1, \frac{v_1}{u_1-1} \right) 
&\quad &(u_{12},v_{12})&:=&\left(\frac{u_1-1}{v_1}, {v_1}\right)\vspace{.2cm}\\
 (u_{x 1},v_{x 1})&:=&\left( {u_{01}-x} , \frac{v_{01}}{u_{01}-x}\right)
&\quad &(u_{x 2},v_{x 2})&:=&\left(\frac{u_{01}-x}{v_{01}}, {v_{01}}\right)\vspace{.2cm}\\&=&\left(u_1-x, \frac{v_{1}}{u_1(u_1-x)}\right)&\quad &&=&\left(\frac{u_1(u_1 -x)}{v_{1}}, \frac{v_{1}}{u_1}\right)
\end{array}$$
 \item[$\star$] In each pair of charts $\C^2_{u_{i1},v_{i1}}$,  $\C^2_{u_{i2},v_{i2}}$ (which effectively replaces $\beta_i$ by the exceptional line $\mathcal{D}_i:=\{u_{i1}=0\} \cup \{v_{i2}=0\}$), we have to remove the points $\beta_j$ for $j\in \{0,1,x\}\setminus \{i\}$ if visible. Yet these points are visible only in the $\C^2_{u_{i1},v_{i1}}$ charts.
 By the remark above, we may remove\par 
 the following visible parts of $\mathcal{V}_0:\{u_{11}=-1\}$ \par
  the following visible parts of $\mathcal{V}_1:\{u_{01}=1\}\cup \{u_{x1}=1-x\}$\par
  the following visible parts of $\mathcal{V}_x:\{u_{01}=x\}\cup \{u_{11}=x-1\}$.\par
   \item[$\star$] 
  The three charts $\C^2_{u_{i1},v_{i1}}$ with the removes lines are equivalent to a single $\C^2$-chart, namely 
  $$(\tilde{u}_1,\tilde{v}_1):= \left(u_{11}+1,\frac{v_{11}}{(u_{11}+1)(u_{11}+1-x)}\right)=\left(u_1 ,\frac{v_1}{u_1(u_1-1)(u_1-x) }\right)\, .$$
  For $x=0$, the lines $\mathcal{V}_0$ and $\mathcal{V}_x$ cannot be distinguished. Yet then the pair of charts $\C^2_{u_{x1},v_{x1}}$,$\C^2_{u_{x2},v_{x2}}$ replaces the chart $\C^2_{u_{01},v_{01}}$. Therefore, this coordinate change is still valid.
  \item[$\star$] Similarly, we need to remove the visible points $\beta_x,\beta_1$ from the chart $\C^{2}_{u_2,v_2}$, which can effectively be done by setting 
 $$(\tilde{u}_2,\tilde{v}_2):=  \left(u_2 ,\frac{v_2}{(1-u_2)(1-xu_2) }\right)=  \left(\frac{1}{u_1} ,\frac{v_1u_1}{(u_1-1)(u_1-x) }\right)= \left(\frac{1}{\tilde u_1} , \tilde v_1\tilde u_1^2 \right)\, .$$
 Indeed, removing $\{\tilde u_2\in \{1,1/t\}$ this chart is isomorphic to $\C^{2}_{u_2,v_2}\setminus \mathcal{V}_x\cup \mathcal{V}_1$, removing 
 $\{\tilde u_2=0\}$ this chart is isomorphic to  $\C^{2}_{\tilde{u}_1,\tilde{v}_1}\setminus \{\tilde{u}_1=0\}$. Note that this holds also for $x=0$.
\end{itemize}

For the blow-up of  $\beta_\infty$, we may stick to the standard procedure: 
\begin{itemize}
 \item[$\star$] Remove the point $\beta_\infty$ from the chart $\C^2_{\tilde{u}_2,\tilde{v}_2}$. Note that then it remains visible only in the chart $\C^2_{{u}_3, {v}_3}$, as $\beta_\infty:(0,\theta)$.
 \item[$\star$] Replace the chart $\C^2_{{u}_3, {v}_3}$ by the following pair of $\C^2$-charts.
$$\begin{array}{rclcrcl}(  u_{\infty 1},  v_{\infty 1})&:=&\left(u_3, \frac{v_3-\theta}{u_3}\right)  &\quad 
&(  u_{\infty 2},  v_{\infty 2})&:=&\left(\frac{u_3}{v_3-\theta}, v_3-\theta\right)  
 \end{array}$$
 \item[$\star$]  Note that
 $\C^2_{u_{\infty 1},  -v_{\infty 1}}$ 
 (with a minus sign) corresponds to the classical chart of a certain surface $\Sigma_\theta$ much used in publications concerning the Okamoto desingularisation of the sixth Painlev\'e equation. See for example \cite{Takano}. Hence, for traditional reasons, we denote this chart by $\C^2_{\tilde u_3,  \tilde v_3}:=\C^2_{u_{\infty 1},  -v_{\infty 1}}$ 
 \end{itemize}
 
 \subsubsection{The vector field in the resulting new charts}\label{s:DetailsB}$ $\\
In our seven new charts, that we have to add to the chart $\C^2_{u_0,v_0}$ to obtain the global picture after blow-up of $\beta_0,\beta_1,\beta_x,\beta_\infty$, the vector field respectively reads as follows.
 
 $$
\left\{
\begin{array}{ccl}
 \dot{ \tilde u}_1&=& 
 \frac{2}{(x-1)\tilde v_1}  - \frac{\tilde u_1 (\tilde u_1 -1)(\tilde u_1 -x)}{(x-1)}\left(\frac{\theta_0}{\tilde u_1 }+\frac{\theta_1}{\tilde u_1 -1}+\frac{\theta_x -1}{\tilde u_1 -x}\right)
 \vspace{.2cm}\\
\dot{\tilde v}_1&=&  \frac{x\theta_0 \tilde v_1 -1}{(x-1)\tilde u_1}-\frac{(x-1)\theta_1\tilde v_1+1}{(x-1)(\tilde u_1-1)}+\frac{x(x-1)\theta_x \tilde v_1-1}{(x-1)(\tilde u_1-x)}+\left(\frac{ \theta+\overline{\theta}}{x-1}(\tilde u_1+x-1) -2\frac{x\theta_0+(x-1)\theta_1}{x-1} \right)\tilde v_1+\vspace{.2cm}
\\
&&+\frac{ \theta\overline{\theta}}{x-1}\tilde u_1(\tilde u_1-1)(\tilde u_1-x)\tilde{v}_1^2\, ,
\end{array}
\right.
$$

$$
 \left\{\begin{array}{ccl}
 \dot{\tilde u}_2&=&- \frac{2 }{(x-1)\tilde v_2}  + \frac{  (1-{\tilde u}_2)(1 -x{\tilde u}_2)}{x-1}\left( \theta_0 +\frac{\theta_1}{1-{\tilde u}_2}+\frac{\theta_x -1}{1-x{\tilde u}_2}\right)\vspace{.2cm}\\
\dot{\tilde v}_2&=& \frac{(\theta \tilde v_2-1)(\overline{\theta}\tilde v_2-1)}{(x-1)\tilde u_2}+\frac{(x-1)\theta_x\tilde v_2-x}{(x-1)(1-x\tilde u_2)}-\frac{(x-1)\theta_1\tilde v_2+1}{(x-1)(1-\tilde u_2)}+\left(2\frac{x\theta_1+\theta_x-1}{x-1}-1\right) \tilde v_2+\frac{\theta \overline{\theta} \tilde v_2-\theta_0}{x-1}(x \tilde u_2-x-1) \tilde v_2\, ,
\end{array}\right.
$$

$$
\left\{\begin{array}{ccl} \dot{\tilde u}_3&=&  \frac{\tilde{u}_3(1-\tilde{u}_3)(1-x\tilde{u}_3)}{(x-1)}\left(2\tilde{v}_3-\frac{ {\theta}_\infty}{\tilde{u}_3}+\frac{x(\theta_x -1)}{1-x\tilde{u}_3}+\frac{\theta_1}{1-\tilde{u}_3} \right)
 \vspace{.2cm}\\ 
\dot{\tilde v}_3&=& -\frac{3x\tilde{u}_3^2-2(x+1)\tilde{u}_3+1}{x-1}\tilde{v}_3^2+\frac{x(2\theta-\theta_0)(2\tilde{u}_3-1)+(x-1)\theta_1-\theta_\infty}{x-1}\tilde{v}_3-x\frac{\theta(\theta-\theta_0)}{x-1}\, ,
 \vspace{.2cm}\\
\end{array}\right.
$$

$$
\left\{\begin{array}{ccl}
\dot{u}_{02}  &=&  x\frac{u_{02}-\theta_0}{(x-1)v_{02} } -\frac{u_{02}^3v_{02}  }{x-1}  +\frac{ \theta+\overline{\theta}}{x-1}u_{02}^2v_{02}  -\frac{ \theta\overline{\theta} }{x-1}u_{02}v_{02}
\vspace{.2cm}\\
\dot{v}_{02} &=& \frac{3(u_{02}v_{02})^2-2(x+1)u_{02}v_{02} +x}{x-1}-2\frac{ \theta+\overline{\theta} }{x-1}u_{02}v_{02}^2  +\left( \frac{x\theta_0}{ x-1}+{\theta_1}+\frac{\theta+\overline{\theta}}{x-1} \right)v_{02} +\frac{ \theta\overline{\theta}}{x-1}v_{02}^2\, , 
\end{array}\right.
$$

$$
\left\{\begin{array}{ccl}
\dot{u}_{x2}  &=&-\frac{(\theta_x-1)(u_{x2}v_{x2}+x)-u_{x2}+x}{v_{x2}}-\frac{u_{x2}(u_{x2}+x\theta_0)}{(x-1)(u_{x2}v_{x2}+x)}+\frac{u_{x2}(u_{x2}+\theta_0)}{x-1} -\frac{\theta \overline{\theta}u_{x2}v_{x2}(u_{x2}v_{x2}+x)}{x-1} \vspace{.2cm}\\
\dot{v}_{x2} &=& \frac{(u_{x2}v_{x2}+x)(\theta v_{x2}-1)(\overline{\theta}v_{x2}-1)}{x-1}+x\frac{\theta_0v_{x2}-1}{(x-1)(u_{x2}v_{x2}+x)}
\, , 
\end{array}\right.
$$

$$
\left\{\begin{array}{ccl}
\dot{u}_{12} &=& -\frac{  u_{12}-\theta_1}{v_{12}} -\frac{ u_{12}^3v_{12}}{x-1}+\frac{\theta+\overline{\theta} }{x-1}u_{12}^2v_{12}  -\frac{\theta\overline{\theta} }{x-1}u_{12}v_{12}
\vspace{.2cm}\\
\dot{v}_{12} &=&    \frac{3u_{12}^2v_{12}^2 -2(x-2)(u_{12}v_{12}) -x+1}{ x-1}-2\frac{\theta+\overline{\theta} }{x-1}u_{12}v_{12}^2  +\left( \frac{x\theta_0}{ x-1}+ \theta_1-\frac{\theta+\overline{\theta}}{x-1} \right)v_{12} +\frac{\theta\overline{\theta}}{x-1}v_{12}^2\, ,\end{array}\right. 
$$

$$
\left\{\begin{array}{ccl}
\dot{u}_{\infty 2}  &=&-\frac{3x(u_{\infty 2} v_{\infty 2} )^2-2(x+1)u_{\infty 2} v_{\infty 2} +1}{x-1}-2x\frac{2\theta-\theta_0}{x-1}u_{\infty 2}^2v_{\infty 2} +\left(\frac{\theta_\infty}{x-1}-\theta_1+\frac{x(2\theta-\theta_0)}{x-1}\right)u_{\infty 2} -\vspace{.2cm}\\&&-x\frac{\theta(\theta-\theta_0)}{x-1}u_{\infty 2} ^2 \vspace{.2cm}\\
%&=&\frac{u_{\infty 2}v_{\infty 2}-1}{x-1}+\frac{u_{\infty 2}(v_{\infty 2}+\theta_\infty)}{x-1}-\theta_1u_{\infty 2}+x\frac{u_{\infty 2}(u_{\infty 2}v_{\infty 2}-1)(2v_{\infty 2}+2\theta-\theta_0)}{1-x}+\vspace{.2cm}\\&&+x\frac{u_{\infty 2}^2(v_{\infty 2}+\theta-\theta_0)(v_{\infty 2}+\theta)}{1-x}\, , 
%\vspace{.2cm}\\
\dot{v}_{\infty 2} &=& \frac{v_{\infty 2} +\theta_\infty}{(1-x)u_{\infty 2} }-x\frac{u_{\infty 2} v_{\infty 2} (v_{\infty 2} +\theta-\theta_0)(v_{\infty 2} +\theta)}{1-x} \, . \end{array}\right.
$$

 \subsection{Detailed charts of Okamoto's space}\label{s:DetailsC}
 
  \subsubsection{The chart $(u_{0} , v_{0})=\left(u,v \right)$}\label{app:00}
$ $\\  Domain of definition: $\C^2$.\\
 Visible components of the infinity set:   $\emptyset$\\
 Visible exceptional lines: $\emptyset$ \vspace{.2cm}
     $$\begin{array}{l}\left\{\begin{array}{rcl}
     (u_{0} , v_{0})&=&\left(u,v \right)\vspace{.2cm}\\
(u,v)&=&\left(u_0,v_0\right)\vspace{.2cm}\\
\omega_{0}&=&  1
\vspace{.2cm}\\
E&=&\frac{u_0(u_0-1)(u_0-x)}{x-1}\Bigl\{ v_0^2-v_0\left(\frac{\theta_0}{u_0} +\frac{\theta_1}{u_0-1}+\frac{\theta_x-1}{u_0-x}\right)  +{\frac{\theta \overline{\theta}}{u_0(u_0-1)}} \Bigr\}.
   \end{array}\right.
\vspace{.5cm}\\

\left\{\begin{array}{rcl}
\dot{u}_0 &=& \frac{u_0 (u_0 -1)(u_0 -x )}{x-1}\left(2v_0 -\frac{\theta_0}{u_0 }-\frac{\theta_1}{u_0 -1}-\frac{\theta_x-1}{u_0 -x }\right)\vspace{.2cm}\\
\dot{v}_0 &=& -\frac{3u_0^2-2(x+1)u_0 +x}{x-1}v_0^2+2\frac{\theta+\overline{\theta} }{x-1}u_0 v_0 -\left( \frac{x\theta_0}{ x-1}+  \theta_1 +\frac{\theta+\overline{\theta}}{x-1} \right)v_0 -\frac{\theta \overline{\theta} }{x-1}\, . 
  \end{array}\right.\end{array}$$
  
\subsubsection{The chart $(\tilde u_{1} , \tilde v_{1})=\left(u,\frac{1}{u(u-x)(u-1)v} \right)$}\label{app:1tilde}
$ $\\  Domain of definition: $\C^2\setminus \{\gamma_0,\gamma_x,\gamma_1\}$, where $$\gamma_0:\left(0,\frac{1}{x\theta_0} \right)\, , \gamma_x:\left(x,\frac{1}{x(x-1)\theta_x} \right)\, , \gamma_1:\left(1,\frac{1}{(1-x)\theta_1} \right)\, .$$
 Visible components of the infinity set:   $$\mathcal{H}^*:\{\tilde{v}_1=0\}\, ,  \mathcal{D}_0^*:\{\tilde{u}_1=0\}\, , \mathcal{D}_x^*:\{\tilde{u}_1=x\}\, , \mathcal{D}_1^*:\{\tilde{u}_1=1\}\, .$$
 Visible exceptional lines: $\emptyset$ \vspace{.2cm}
     $$\begin{array}{l}\left\{\begin{array}{rcl}
     (\tilde u_{1} , \tilde v_{1})&=&\left(u,\frac{1}{u(u-x)(u-1)v} \right)\vspace{.2cm}\\
(u,v)&=&\left(\tilde u_{1},\frac{1}{\tilde u_{1}(\tilde u_{1}-x)(\tilde u_{1}-1)\tilde v_{1}}\right)\vspace{.2cm}\\
\tilde \omega_{1}&=& - \tilde u_{1}(\tilde u_{1}-x)(\tilde u_{1}-1)\tilde v_{1}^2
\vspace{.2cm}\\
\dot E&=&-\frac{x}{(x-1)^2}\Bigl\{ \frac{1}{\tilde u_{1}(\tilde u_{1}-x)^2\tilde v_{1}^2} -\frac{(\theta+\overline{\theta})\tilde u_{1}-\theta_0} {\tilde v_{1}\tilde u_{1}(\tilde u_{1}-x)}   + \theta \overline{\theta}(\tilde{u}_1-1) \Bigr\}\vspace{.2cm}\\
E&=&\frac{1}{x-1}\Bigl\{ \frac{1}{\tilde u_{1}(\tilde u_{1}-x)(\tilde u_{1}-1)\tilde v_{1}^2} -\frac{1} {\tilde v_{1}}\left(\frac{\theta_0}{\tilde{u}_1} +\frac{\theta_1}{\tilde{u}_1-1}+\frac{\theta_x-1}{\tilde{u}_1-x}\right)  + \theta \overline{\theta}(\tilde{u}_1-x) \Bigr\}.
   \end{array}\right.
\vspace{.5cm}\\

\left\{\begin{array}{rcl}
\dot{\tilde{u}}_1 &=& \frac{ 2}{(x-1)\tilde v_{1}} - \frac{\tilde{u}_1 (\tilde{u}_1 -1)(\tilde{u}_1 -x )}{x-1}\left(  \frac{\theta_0}{\tilde{u}_1 }+\frac{\theta_1}{\tilde{u}_1 -1}+\frac{\theta_x-1}{\tilde{u}_1 -x }\right)\vspace{.2cm}\\
\dot{\tilde v}_1 &=& \frac{x\theta_0 \tilde v_1 -1}{(x-1)\tilde u_1}-\frac{(x-1)\theta_1\tilde v_1+1}{(x-1)(\tilde u_1-1)}+\frac{x(x-1)\theta_x \tilde v_1-1}{(x-1)(\tilde u_1-x)}+\left(\frac{ \theta+\overline{\theta}}{x-1}(\tilde u_1+x-1) -2\frac{x\theta_0+(x-1)\theta_1}{x-1} \right)\tilde v_1+\vspace{.2cm}\\
&&+\frac{ \theta\overline{\theta}}{x-1}\tilde u_1(\tilde u_1-1)(\tilde u_1-x)\tilde{v}_1^2\, . 
  \end{array}\right.\end{array}$$

  \subsubsection{The chart $(\tilde u_{2} , \tilde v_{2})=\left(\frac{1}{u},\frac{u}{(u-x)(u-1)v} \right)$}\label{app:2tilde}
$ $\\  Domain of definition: $\C^2\setminus \{ \gamma_x,\gamma_1,\beta_\infty,\beta_\infty^-\}$, where $$\gamma_x:\left(\frac{1}{x},\frac{x}{(x-1)\theta_x} \right)\, ,  \gamma_1:\left(1,\frac{1}{(1-x)\theta_1} \right)\, ,  \beta_\infty:\left(0,\frac{1}{ \theta } \right)\, ,  \beta_\infty^-:\left(0,\frac{1}{ \overline{\theta} } \right)\, .$$
 Visible components of the infinity set:   $$\mathcal{H}^*:\{\tilde{v}_2=0\}\, ,  \mathcal{D}_\infty^{**}:\{\tilde{u}_2=0\}\, , \mathcal{D}_x^*:\{\tilde{u}_2=1/x\}\, , \mathcal{D}_1^*:\{\tilde{u}_2=1\}\, .$$
 Visible exceptional lines: $\emptyset$ \vspace{.2cm}
     $$\begin{array}{l}\left\{\begin{array}{rcl}
     (\tilde u_{2} , \tilde v_{2})&=&\left(\frac{1}{u},\frac{u}{(u-x)(u-1)v} \right)\vspace{.2cm}\\
(u,v)&=&\left(\frac{1}{\tilde u_2},\frac{\tilde u_{2}}{ (1-x\tilde u_2)(1-\tilde u_2)\tilde v_{2}}\right)\vspace{.2cm}\\
\tilde \omega_{2}&=&  \tilde u_2(1-x\tilde u_2)(1-\tilde u_2)\tilde v_{2}^2
\vspace{.2cm}\\
\dot E&=&-\frac{x}{(x-1)^2}\Bigl\{ \frac{1}{\tilde u_{2}(1-x\tilde u_{2})^2\tilde v_{2}^2} -\frac{(\theta+\overline{\theta}) -\theta_0\tilde u_{2}} {\tilde v_{2}\tilde u_{2}(1-x\tilde u_{2})}   + \frac{\theta \overline{\theta}}{\tilde u_{2}}(1-\tilde{u}_2) \Bigr\}\vspace{.2cm}\\
E&=& \frac{(\theta \tilde v_{2}-1)\left(\overline{\theta}\tilde v_{2}-1\right)}{(x-1)\tilde u_{2}\tilde v_{2}^2}- \frac{(x-1)\theta_1 \tilde v_{2}+1 }{(x-1)^2(1-\tilde u_{2})\tilde v_{2}^2}-x\frac{(x-1)(\theta_x-1) \tilde v_{2}-x }{(x-1)^2(1-x\tilde u_{2})\tilde v_{2}^2}-\frac{x\theta\overline{\theta}}{x-1}
   \end{array}\right.
\vspace{.5cm}\\

\left\{\begin{array}{rcl}
 \dot{\tilde u}_2&=&- \frac{2 }{(x-1)\tilde v_2}  + \frac{  (1-{\tilde u}_2)(1 -x{\tilde u}_2)}{x-1}\left( \theta_0 +\frac{\theta_1}{1-{\tilde u}_2}+\frac{\theta_x -1}{1-x{\tilde u}_2}\right)\vspace{.2cm}\\
\dot{\tilde v}_2&=& \frac{(\theta \tilde v_2-1)(\overline{\theta}\tilde v_2-1)}{(x-1)\tilde u_2}+\frac{(x-1)\theta_x\tilde v_2-x}{(x-1)(1-x\tilde u_2)}-\frac{(x-1)\theta_1\tilde v_2+1}{(x-1)(1-\tilde u_2)}+\left(2\frac{x\theta_1+\theta_x-1}{x-1}-1\right) \tilde v_2+\frac{\theta \overline{\theta} \tilde v_2-\theta_0}{x-1}(x \tilde u_2-x-1) \tilde v_2\, .
  \end{array}\right.\end{array}$$

  \subsubsection{The chart $(\tilde u_{3} , \tilde v_{3})=\left(u_3,-\frac{v_3-\theta}{u_3} \right)$}\label{app:3tilde}
$ $\\  Domain of definition: $\C^2.$\\
 Visible components of the infinity set:  $\emptyset$.\\
 Visible exceptional lines: $\mathcal{E}_\infty:\{\tilde u_{3}=0\}$ \vspace{.2cm}
     $$\begin{array}{l}\left\{\begin{array}{rcl}
     (\tilde u_{3} , \tilde v_{3})&=&\left(\frac{1}{u}, -u(uv-\theta)  \right)\vspace{.2cm}\\
(u,v)&=&\left(\frac{1}{\tilde u_3},-\tilde u_{3}^2\tilde v_{3}+\theta\tilde u_{3}\right)\vspace{.2cm}\\
\tilde \omega_{3}&=&  1\vspace{.2cm}\\
E&=&(1-\tilde u_{3})(1-x\tilde u_{3})\left(\frac{\tilde v_{3}(\tilde u_{3}\tilde v_{3}-2\theta+\theta_0)}{x-1}   +\frac{\theta\overline{\theta}+\theta_1(\tilde v_{3}-\theta)}{(x-1)(1-\tilde u_{3})}+\frac{(\theta_x-1)(\tilde v_{3}-x\theta)}{(x-1)(1-x\tilde u_{3})}\right)
   \end{array}\right.
\vspace{.5cm}\\

\left\{\begin{array}{rcl}
\dot{\tilde u}_3&=&  \frac{\tilde{u}_3(1-\tilde{u}_3)(1-x\tilde{u}_3)}{(x-1)}\left(2\tilde{v}_3-\frac{ {\theta}_\infty}{\tilde{u}_3}+\frac{x(\theta_x -1)}{1-x\tilde{u}_3}+\frac{\theta_1}{1-\tilde{u}_3} \right)
 \vspace{.2cm}\\ 
\dot{\tilde v}_3&=& -\frac{3x\tilde{u}_3^2-2(x+1)\tilde{u}_3+1}{x-1}\tilde{v}_3^2+\frac{x(2\theta-\theta_0)(2\tilde{u}_3-1)+(x-1)\theta_1-\theta_\infty}{x-1}\tilde{v}_3-x\frac{\theta(\theta-\theta_0)}{x-1}\, .
  \end{array}\right.\end{array}$$

  \subsubsection{The chart $(u_{03} , v_{03})=\left(u_{02} -\theta_0, \frac{v_{02}}{u_{02}-\theta_0} \right)$}\label{app:0303}
$ $\\  Domain of definition: $\C^2$.\\
 Visible components of the infinity set:   $ \mathcal D_0^* : \{v_{03}=0\} $\\
 Visible exceptional lines: $ \mathcal E_0 : \{u_{03}=0\} $\vspace{.2cm}
     $$\begin{array}{l}\left\{\begin{array}{rcl}
     (u_{03} , v_{03})&=&\left(uv-\theta_0, \frac{1}{uv^2-\theta_0 v} \right)\vspace{.2cm}\\
(u,v)&=&\left(u_{03}^2v_{03}+ \theta_0u_{03}v_{03} , \frac{1}{u_{03}v_{03}   }\right)\vspace{.2cm}\\
\omega_{03}&=&  -v_{03}
\vspace{.2cm}\\
E&=&\left.\frac{1}{x-1}\right[\frac{((u_{03}+\theta_0)u_{03}v_{03}-x)((u_{03}+\theta_0)u_{03}v_{03}-1)}{v_{03}}-(\theta+\overline{\theta}-\theta_0)(u_{03}+\theta_0)^2u_{03}v_{03}  +\\&&\quad \quad\quad\quad \quad \quad\quad \quad \Bigl.+\theta\overline{\theta}((u_{03}+\theta_0)u_{03}v_{03}-x)   +\Bigl( x\theta_1 +(\theta_x-1)\Bigr)(u_{03}+\theta_0)\Bigr] 
   \end{array}\right.
\vspace{.5cm}\\

\left\{\begin{array}{rcl}
\dot u_{03} &=& \frac{1}{1-x}(u_{03}+\theta_0-\theta)(u_{03}+\theta_0-\overline{\theta}) (u_{03}+\theta_0)u_{03}v_{03}
  -\frac{x}{(1-x)v_{03}}
\vspace{.2cm}\\
\dot v_{03} &=&- \frac{1}{1-x}(u_{03}+\theta_0-\theta)(u_{03}+\theta_0-\overline{\theta}) (2u_{03}+\theta_0) v_{03}^2
 -\frac{2(u_{03}+\theta_0)-(\theta+\overline{\theta})}{1-x}(u_{03}+\theta_0)u_{03}v_{03}^2 +\vspace{.2cm}\\
&&+\frac{2u_{03}+\theta_0-\theta_x+1  }{1-x}v_{03}+\frac{x(2u_{03}+\theta_0-\theta_1)}{1-x}v_{03}  \end{array}\right.\end{array}$$

We have:
{$$\omega_{x3}= -\frac{1}{((u-x)uv-x\theta_x)v}=-\frac{u_{03}v_{03}}{u_{03}v_{03} (u_{03} + \theta_0 )^2-x(u_{03} + \theta_0+\theta_x)}$$
  $$\omega_{x3}\sim \frac{u_{03}v_{03}}{ x(u_{03} + \theta_0+\theta_x)}$$
  $$E\sim  \frac{1}{x-1} \frac{x}{v_{03}} $$
    $$E\omega_{x3}\sim \frac{1}{x-1} \frac{u_{03} }{ (u_{03} + \theta_0+\theta_x)}\longrightarrow \frac{1}{x-1}  $$
   $$(x-1)\omega_{x3} \longrightarrow \frac{1}{E}  $$}
 
Further formulae:

      $$ \left\{\begin{array}{rcl}
 -\frac{(x-1)\omega_{03}}{x}E&=&  1 - \frac{v_{03}}{x} \Bigl[(\theta+\overline{\theta}-\theta_0-u_{03})(u_{03}+\theta_0)^2u_{03}v_{03}  +(u_{03}+\theta_0)u_{03}(\theta\overline{\theta}v_{03}+x+1)-\\&&\quad \quad\quad\quad \quad \quad\quad \quad \quad\quad \quad \quad\quad \quad \quad\quad \quad \quad\quad \quad \Bigl.-\theta\overline{\theta}x   +\Bigl( x\theta_1 +\theta_x-1\Bigr)(u_{03}+\theta_0)\Bigr]\vspace{.2cm}\\
 \frac{\dot \omega_{03}}{\omega_{03}}&=&- \frac{1}{1-x}(u_{03}+\theta_0-\theta)(u_{03}+\theta_0-\overline{\theta}) (2u_{03}+\theta_0) v_{03}
 -\frac{2(u_{03}+\theta_0)-(\theta+\overline{\theta})}{1-x}(u_{03}+\theta_0)u_{03}v_{03} -\vspace{.2cm}\\
&&+\frac{2u_{03}+\theta_0-\theta_x+1  }{1-x}+\frac{x(2u_{03}+\theta_0-\theta_1)}{1-x} \end{array}\right. $$
 
    \subsubsection{The chart $(u_{04} , v_{04})=\left( \frac{1}{v_{03}} ,  u_{03}v_{03}  \right)$}\label{app:0404}
 $ $\\  Domain of definition: $\C^2$.\\
 Visible components of the infinity set:   $\emptyset $\\
 Visible exceptional lines: $ \mathcal E_0 : \{v_{04}=0\} $\vspace{.2cm}   $$\begin{array}{l}\left\{\begin{array}{rcl}
   (u_{04} , v_{04})&=& \left(uv^2-\theta_0 v  , \frac{1}{v} \right)\vspace{.2cm}\\
(u,v)&=&\left((u_{04}v_{04}+ {\theta_0})v_{04} , \frac{1}{ v_{04}  } \right)\vspace{.2cm}\\
\omega_{04}&=&  -1
\vspace{.2cm}\\
E&=&\left.\frac{1}{x-1}\right[((u_{04}v_{04}+\theta_0)v_{04}-x)((u_{04}v_{04}+\theta_0)v_{04}-1)u_{04}-\\&&\quad \quad\quad\quad \quad \quad\quad \quad -(\theta+\overline{\theta}-\theta_0)(u_{04}v_{04}+\theta_0)^2v_{04} +\theta\overline{\theta}((u_{04}v_{04}+\theta_0)v_{04}-x) +\\&&\hspace{4.5cm}\quad \quad\quad\quad \quad \quad\quad \quad \Bigl.   +\Bigl( x\theta_1 +(\theta_x-1)\Bigr)(u_{04}v_{04}+\theta_0)\Bigr]
   \end{array}\right.
\vspace{.5cm}\\

\left\{\begin{array}{rcl}
\dot u_{04}&=&\frac{1}{1-x}(u_{04}v_{04}+\theta_0-\theta)(u_{04}v_{04}+\theta_0-\overline{\theta}) (2u_{04}v_{04}+\theta_0)  
  -\theta_1u_{04}   -\vspace{.2cm}\\
&&-\frac{(\theta+\overline{\theta})-2(u_{04}v_{04}+\theta_0) }{1-x}[(u_{04}v_{04}+\theta_0)v_{04}-1]u_{04} -\frac{x(2u_{04}v_{04}+\theta_0 )}{1-x}u_{04}\vspace{.2cm}\\

\dot v_{04}  
&=&-\frac{1}{1-x}\Bigl[ \left[2(u_{04}v_{04}+\theta_0)-\theta\right]v_{04}-1\Bigr]\Bigl[\left[2(u_{04}v_{04}+\theta_0)-\overline{\theta}\right]v_{04}-1\Bigr]+
\vspace{.2cm}\\
&&+\frac{1}{1-x}\left[(u_{04}v_{04}+\theta_0)v_{04}-1\right]^2+\theta_1v_{04}+\frac{x}{1-x}\left[(2u_{04}v_{04}+\theta_0)v_{04}-1\right]
\end{array}\right.\end{array}$$
  
 \subsubsection{The chart $(u_{x3} , v_{x3})=\left(u_{x2}-x\theta_x , \frac{v_{x2}}{u_{x2}-x\theta_x} \right) $}\label{app:x3x3}
  $ $\\  Domain of definition: $\C^2\setminus \{\gamma_0,\alpha \}$, where $$\gamma_0:\left(-x(\theta_0+\theta_x)\,  ,\,  -\frac{1}{x\theta_0(\theta_0+\theta_x)}\right)\, , \quad \alpha :\left(-x\frac{\theta_x}{2} \,  ,\,  \frac{4}{x\theta_x^2}\right)\, .$$ Here $\alpha $ is an apparent base point on $\mathcal D_0\setminus \mathcal{H}$, not a base point in the charts $\C^2_{u_{03},v_{03}}$ and $\C^2_{u_{04},v_{04}}$.
  \\
 Visible components of the infinity set:   $\mathcal D_0\setminus \gamma_0 : \{u_{x3}v_{x3}(u_{x3}+x\theta_x)=-x\} \,,  \mathcal D_x^* : \{v_{x3}=0\}$\\
 Visible exceptional lines: $ \mathcal E_x : \{u_{x3}=0\} $\vspace{.2cm}
   $$\begin{array}{l}\left\{\begin{array}{rcl}
 (u_{x3} , v_{x3})&=&\left((u-x)uv-x\theta_x , \frac{1}{((u-x)uv-x\theta_x)uv} \right)\vspace{.2cm}\\
(u,v)&=&\left((u_{x3}+x\theta_x)u_{x3}v_{x3}+x , \frac{1}{((u_{x3}+x\theta_x)u_{x3}v_{x3}+x)u_{x3}v_{x3}} \right)\vspace{.2cm}\\
\omega_{x3}&=& -((u_{x3}+x\theta_x)u_{x3}v_{x3}+x) v_{x3} 
\vspace{.2cm}\\
  E&=&\left.\frac{1}{(x-1) u_{x3}v_{x3}}\right[-\frac{u_{x3}+x\theta_x+x\theta_0}{(u_{x3}+x\theta_x)u_{x3}v_{x3}+x}-(\theta u_{x3}v_{x3}-1)(\overline{\theta}u_{x3}v_{x3}-1)(u_{x3}+x\theta_x) +\Bigr.\\&&
  \quad   \quad   \quad   \quad   \quad   \quad   \quad   \quad   \quad   \quad   \quad  \quad   \quad   \quad   \quad   \quad   \quad   \quad   \quad   \quad   \quad   \quad \Bigl.+  \theta_0-(x-1)(\theta_x-1) \Bigr]\vspace{.2cm}\\
  \end{array}\right.
\vspace{.5cm}\\

\left\{\begin{array}{rcl}
\dot u_{x3} &=& -u_{x3}(\theta_x-1)-x\theta_x^2  -\frac{(u_{x3}+x\theta_x+\theta_0)(u_{x3}+x\theta_x)}{1-x}+\frac{1}{ v_{x3}}+\frac{(u_{x3}+x(\theta_0+\theta_x))(u_{x3}+x\theta_x)}{(1-x)((u_{x3}+x\theta_x)u_{x3}v_{x3}+x)}+\vspace{.2cm}\\
&& +\frac{\theta\overline{\theta}}{1-x} ((u_{x3}+x\theta_x)u_{x3}v_{x3}+x)(u_{x3}+x\theta_x)u_{x3}v_{x3} 
\vspace{.2cm}\\
 \dot v_{x3} &=&\frac{((u_{x3}+x\theta_x)u_{x3}v_{x3}+x)v_{x3}}{1-x}\left[-\theta\overline{\theta}(2u_{x3}+x\theta_x)v_{x3}+(\theta+\overline{\theta})\right]-
\vspace{.2cm}\\
&&-\frac{(2u_{x3}+x\theta_x)v_{x3}}{(1-x)((u_{x3}+x\theta_x)u_{x3}v_{x3}+x)}+\frac{(\theta_0+\theta_x)v_{x3}}{1-x}\left(1+\frac{x\theta_x(u_{x3}+x\theta_x)v_{x3}-2x}{(u_{x3}+x\theta_x)u_{x3}v_{x3}+x}\right)
-v_{x3}
\end{array}\right.\end{array}$$
Further formulae:
 $$\begin{array}{l}\left\{\begin{array}{rcl}
 - \omega_{x3} E&=&-\left[\frac{x(\theta \overline{\theta}u_{x3}v_{x3}-(\theta+\overline{\theta}+\theta_x-1))+(\theta+\overline{\theta}-\theta_1)+(\theta u_{x3}v_{x3}-1)(\overline{\theta}u_{x3}v_{x3}-1)(u_{x3}+x\theta_x)}{1-x}\right] (u_{x3}+x\theta_x)v_{x3}  +\vspace{.2cm}\\&&
+\left(1+ \frac{x}{u_{x3} }\right)
\vspace{.2cm}\\
\frac{\dot \omega_{x3}}{\omega_{x3}}&=&-\frac{((u_{x3}+x\theta_x)u_{x3}v_{x3}+x)(2u_{x3}+x\theta_x) v_{x3}}{1-x}\cdot\theta\overline{\theta} +\frac{2u_{x3}+x\theta_x}{x-1}+
\vspace{.2cm}\\
&&+ \frac{\theta+\overline{\theta}}{1-x}(2(u_{x3}+x\theta_x)u_{x3}v_{x3}+x)-x(\theta_0+\theta_x)\frac{\theta_x(u_{x3}+x\theta_x)v_{x3}-1}{(x-1)((u_{x3}+x\theta_x)u_{x3}v_{x3}+x)}\vspace{.2cm}\\
&=&\frac{2u_{x3}+x\theta_x }{1-x}\left(\theta\overline{\theta}\omega_{x3}-(\theta_0+\theta_x)\theta_x\frac{xv_{x3} }{\omega_{x3}}v_{x3}+1 \right)  -x \frac{\theta+\overline{\theta}}{1-x}(2 \frac{\omega_{x3}}{xv_{x3}} +1) - \frac{ (\theta_0+\theta_x)}{x-1 } \frac{xv_{x3}}{\omega_{x3}}
%\vspace{.2cm}\\ \omega_{x3}&=& -((u_{x3}+x\theta_x)u_{x3}v_{x3}+x) v_{x3} 
\end{array}\right.\end{array}$$

 \subsubsection{The chart $(u_{x4} , v_{x4}) =\left( \frac{u_{x2}-x\theta_x}{v_{x2}} , {v_{x2}} \right) $}\label{app:x4x4}
  $ $\\  Domain of definition: $\C^2\setminus \{\gamma_0,\alpha \}$, where $$\gamma_0:\left(-x(\theta_0+\theta_x)\theta_0\,  ,\,  \frac{1}{\theta_0}\right) \, , \quad \alpha :\left(x\frac{\theta_x^2}{4} \,  ,\, - \frac{2}{\theta_x}\right)\, .$$ Visible components of the infinity set:   $\mathcal D_0\setminus \gamma_0 :  \{ v_{x4}(u_{x4}v_{x4}+x\theta_x)=-x\}$\\
 Visible exceptional lines: $ \mathcal E_x : \{v_{x4}=0\} $\vspace{.2cm}
 $$\begin{array}{l}\left\{\begin{array}{rcl}
 (u_{x4} , v_{x4})&=& \left(((u-x)uv-x\theta_x)uv , \frac{1}{uv} \right)\vspace{.2cm}\\
(u,v)&=&\left((u_{x4}v_{x4}+x\theta_x)v_{x4}+x , \frac{1}{((u_{x4}v_{x4}+x\theta_x)v_{x4}+x)v_{x4}} \right)\vspace{.2cm}\\
\omega_{x4}&=& -((u_{x4}v_{x4}+x\theta_x)v_{x4}+x) 
\vspace{.2cm}\\
  E&=&\left.\frac{1}{(x-1) v_{x4}}\right[-\frac{u_{x4}v_{x4}+x\theta_x+x\theta_0}{(u_{x4}v_{x4}+x\theta_x) v_{x4}+x}-(\theta v_{x4}-1)(\overline{\theta}v_{x4}-1)(u_{x4}v_{x4}+x\theta_x) +\Bigr.\\&&
  \quad   \quad   \quad   \quad   \quad   \quad   \quad   \quad   \quad   \quad   \quad  \quad   \quad   \quad   \quad   \quad   \quad   \quad   \quad   \quad   \quad   \quad \Bigl.+  \theta_0-(x-1)(\theta_x-1) \Bigr] 
  \end{array}\right.
\vspace{.5cm}\\

\left\{\begin{array}{rcl}
\dot u_{x4} &=& \frac{((u_{x4}v_{x4}+x\theta_x)v_{x4}+x)}{1-x}\left[\theta\overline{\theta}(2u_{x4}v_{x4}+x\theta_x)-(\theta+\overline{\theta})u_{x4}\right]+
\vspace{.2cm}\\
&&+\frac{(2u_{x4}v_{x4}+x\theta_x)u_{x4}}{(1-x)((u_{x4}v_{x4}+x\theta_x)v_{x4}+x)}-\frac{\theta_0+\theta_x}{1-x}\left(u_{x4}+\frac{x\theta_x(u_{x4}v_{x4}+x\theta_x)-2xu_{x4}}{(u_{x4}v_{x4}+x\theta_x)v_{x4}+x}\right)
+u_{x4}
 \vspace{.2cm}\\
\dot v_{x4} &=&-\frac{(u_{x4}v_{x4}+x\theta_x)v_{x4}+x}{1-x}(\theta v_{x4}-1)(\overline{\theta}v_{x4}-1)+x\frac{\theta_0v_{x4}-1}{(x-1)((u_{x4}v_{x4}+x\theta_x)v_{x4}+x)}
 \end{array}\right.\end{array}$$
We also have:
{$$\omega_{x3}= -((u_{x4}v_{x4}+x\theta_x) v_{x4}+x) \frac{1}{u_{x4}} $$}
  
\subsubsection{The chart $(u_{13},v_{13})=\left(u_{12}-\theta_1 , \frac{v_{12}}{u_{12}-\theta_1}\right) $}\label{app:1313}
  $ $\\  Domain of definition: $\C^2$.\\
   Visible components of the infinity set:   $\mathcal D_1^* : \{v_{13}=0\}$\\
 Visible exceptional lines: $ \mathcal E_1 : \{u_{13}=0\} $\vspace{.2cm}
 $$\begin{array}{l}\left\{\begin{array}{rcl}
 (u_{13},v_{13})&=& \left((u-1)v-\theta_1 , \frac{1}{(u-1)v^2-\theta_1v} \right)\vspace{.2cm}\\
(u,v)&=&\left(u_{13}^2v_{13}+\theta_1 u_{13}v_{13}+1 , \frac{1}{ u_{13}v_{13}} \right)\vspace{.2cm}\\
\omega_{13}&=& -v_{13}\vspace{.2cm}\\
\dot E&=&-x\frac{u_{13}+\theta_1 }{(x-1)^2  }\Bigl[(u_{13}+\theta_1-\theta )(u_{13}+\theta_1-\overline{\theta}  )u_{13}v_{13}+u_{13}-(\theta_x-1)  \Bigr]\vspace{.2cm}\\
E &=&\frac{(u_{13}+\theta_1-\theta )(u_{13}+\theta_1-\overline{\theta}  )}{x-1 }((u_{13}+\theta_1)u_{13}v_{13}+1)-x\frac{\theta \overline{\theta} -\theta_0(u_{13}+\theta_1)}{x-1}-\frac{(u_{13}+\theta_1)u_{13}v_{13}+1}{v_{13} }  \end{array}\right.
\vspace{.5cm}\\

\left\{\begin{array}{rcl}
\dot u_{13} &=&   -\frac{1}{v_{13}}+\frac{1}{1-x} \left(u_{13} +\theta_1-\theta\right)\left(u_{13} +\theta_1-\overline{\theta}\right)(u_{13}+\theta_1) u_{13}v_{13}\vspace{.2cm}\\
\dot v_{13} &=&-(2u_{13}+\theta_1)v_{13}-\frac{1}{1-x} \left(u_{13} +\theta_1-\theta\right)\left(u_{13} +\theta_1-\overline{\theta}\right)(2u_{13}+\theta_1)v_{13}^2-\vspace{.2cm}\\
&&\quad - \frac{1}{1-x}\left(2u_{13} +2\theta_1-\theta-\overline{\theta}\right)( u_{13}^2v_{13}+\theta_1 u_{13}v_{13} +1)v_{13}-\frac{x\theta_0}{1-x}v_{13}
 \end{array}\right.\end{array}$$
 
 \subsubsection{The chart $(u_{14},v_{14})=\left(\frac{u_{12}-\theta_1}{v_{12}} , v_{12}\right) $}\label{app:1414}
 $ $\\  Domain of definition: $\C^2.$\\
  Visible components of the infinity set:   $\emptyset$\\
 Visible exceptional lines: $ \mathcal E_1 : \{v_{14}=0\} $\vspace{.2cm}
 $$\begin{array}{l}\left\{\begin{array}{rcl}
 (u_{14},v_{14})&=& \left((u-1)v^2-\theta_1v , \frac{1}{ v} \right)\vspace{.2cm}\\
(u,v)&=&\left(u_{14}v_{14}^2+\theta_1  v_{14}+1 , \frac{1}{ v_{14}} \right)\vspace{.2cm}\\
\omega_{14}&=& -1
\vspace{.2cm}\\
E &=&\frac{(u_{14}v_{14}+\theta_1-\theta )(u_{14}v_{14}+\theta_1-\overline{\theta}  )}{x-1 }((u_{14}v_{14}+\theta_1) v_{14}+1)-x\frac{\theta \overline{\theta} -\theta_0(u_{14}v_{14}+\theta_1)}{x-1}-\Bigr.\\&&   \Bigl.-((u_{14}v_{14}+\theta_1) v_{14}+1)u_{14}\vspace{.2cm}\\
 \end{array}\right.
\vspace{.5cm}\\

\left\{\begin{array}{rcl}
\dot u_{14} &=&  (2u_{14}v_{14}+\theta_1)u_{14}+\frac{(2u_{14}v_{14}+\theta_1) }{1-x}\left(u_{14}v_{14} +\theta_1-\theta\right)\left(u_{14}v_{14} +\theta_1-\overline{\theta}\right)+\vspace{.2cm}\\
&&\quad +2\frac{(u_{14}v_{14}^2+\theta_1  v_{14}+1)u_{14}}{1-x}\left(u_{14}v_{14} +\theta_1-\frac{\theta+\overline{\theta}}{2} \right)+\frac{x\theta_0u_{14}}{1-x}
\vspace{.2cm}\\
\dot v_{14} &=&-\frac{v_{14}^2}{1-x}\left(u_{14}v_{14} +\theta_1-\theta\right)\left(u_{14}v_{14} +\theta_1-\overline{\theta}\right)-1-2\left(u_{14}v_{14} -\frac{\theta_0-\theta_1}{2}\right)v_{14}-
\vspace{.2cm}\\&&\quad - \frac{\theta_0 v_{14}}{1-x}-2\frac{(u_{14}v_{14}^2+\theta_1  v_{14}+1)v_{14}}{1-x}\left(u_{14}v_{14}+\theta_1-\frac{\theta+\overline{\theta}}{2}  \right)
 
 \end{array}\right.\end{array}$$
 \subsubsection{The chart $(u_{\infty 3} , v_{\infty 3}) =\left(u_{\infty 2}, \frac{v_{\infty 2}+\theta_\infty}{u_{\infty 2}}\right) $}\label{app:inf3inf3}
 $ $\\  Domain of definition: $\C^2.$\\
  Visible components of the infinity set:   $\emptyset$\\
 Visible exceptional lines: $ \mathcal E_\infty^- : \{u_{\infty 3}=0\} $, $ \mathcal E_\infty : \{u_{\infty 3}v_{\infty 3}=\theta_\infty \} $\vspace{.2cm}
$$
\begin{array}{l}\left\{\begin{array}{rcl}
 (u_{\infty 3} , v_{\infty 3})&=&\left(\frac{1}{u^2v-\theta u},(uv-\theta)(uv-\overline{\theta})u\right)\vspace{.2cm}\\
(u,v)&=&\left(\frac{1}{(u_{\infty 3}v_{\infty 3}-\theta_\infty)u_{\infty 3}}, \left(u_{\infty 3}v_{\infty 3}+\overline{\theta}\right)(u_{\infty 3}v_{\infty 3}-\theta_\infty)u_{\infty 3} \right)\vspace{.2cm}\\
\omega_{\infty 3}&=& -1\vspace{.2cm}\\
E &=&\theta_1(u_{\infty 3}v_{\infty 3}+\overline{\theta})+\frac{x}{x-1}(u_{\infty 3}v_{\infty 3}+\overline{\theta}-\theta_0)(u_{\infty 3}v_{\infty 3}+\overline{\theta})[u_{\infty 3}(u_{\infty 3}v_{\infty 3}-\theta+\overline{\theta})-1]  -\vspace{.2cm}\\&&-\theta \overline{\theta}-\frac{v_{\infty 3}}{x-1}[u_{\infty 3}(u_{\infty 3}v_{\infty 3}-\theta+\overline{\theta})-1] 
 \end{array}\right.
\vspace{.5cm}\\
\left\{\begin{array}{rcl}
\dot u_{\infty 3}&=&- \left(\theta_1+\frac{\theta_\infty}{1-x}\right)u_{\infty 3}-\frac{2(u_{\infty 3}v_{\infty 3}-\theta_\infty)u_{\infty 3} -1}{1-x}-x\frac{(u_{\infty 3}^2v_{\infty 3}-\theta_\infty u_{\infty 3}-1)^2}{1-x} +
\vspace{.2cm}\\
&&+\frac{x}{1-x }\Bigl((2u_{\infty 3}v_{\infty 3}-\theta-\theta_0+2\overline{\theta})u_{\infty 3}-1\Bigr)\Bigl((2u_{\infty 3}v_{\infty 3}-\theta+2\overline{\theta})u_{\infty 3}-1\Bigr)
\vspace{.2cm}\\
\dot v_{\infty 3} &=&v_{\infty 3}(2u_{\infty 3}v_{\infty 3}+\theta_1-\theta_\infty)- x\frac{ \overline{\theta}(\overline{\theta}-\theta_0)}{1-x}   (2u_{\infty 3}v_{\infty 3}-\theta_\infty)  -\vspace{.2cm}\\
&& -\frac{x}{1-x}\Bigl[    2(2u_{\infty 3}v_{\infty 3}-\theta_\infty)(u_{\infty 3}^2v_{\infty 3}-\theta_\infty u_{\infty 3}-1)    -\vspace{.2cm}\\
&& -(\theta+ \overline{\theta}-\theta_0)\bigl(1-3u_{\infty 3}^2v_{\infty 3}+2\theta_\infty u_{\infty 3}\bigr) 
\Bigr]v _{\infty 3}
   \end{array}\right.
\end{array}$$
    
  \subsubsection{The chart $(u_{\infty 4} , v_{\infty 4})=\left(\frac{u_{\infty 2}}{v_{\infty 2}+\theta_\infty}, v_{\infty 2}+\theta_\infty\right) $}\label{app:inf4inf4}
  $ $\\  Domain of definition: $\C^2.$\\
  Visible components of the infinity set:   $\mathcal{D}_\infty^{**}:\{u_{\infty 4}=0\}$\\
 Visible exceptional lines: $ \mathcal E_\infty^- : \{v_{\infty 4}=0\} $, $ \mathcal E_\infty : \{ v_{\infty 4}=\theta_\infty \} $\vspace{.2cm}
$$
\begin{array}{l}\left\{\begin{array}{rcl}
 (u_{\infty 4} , v_{\infty 4})&=&\left(\frac{1}{(uv-\theta)(uv-\overline{\theta})u  }, uv-\overline{\theta} \right)\vspace{.2cm}\\
(u,v)&=&\left(\frac{1}{(v_{\infty 4}-\theta_\infty)u_{\infty 4}v_{\infty 4}}, \left(v_{\infty 4}+\overline{\theta}\right)(v_{\infty 4}-\theta_\infty)u_{\infty 4}v_{\infty 4} \right)\vspace{.2cm}\\
\omega_{\infty 4}&=& -u_{\infty 4}\vspace{.2cm}\\
  \dot E&=&x\frac{(v_{\infty 4}-\theta_\infty)u_{\infty 4}v_{\infty 4} -1}{(x-1)^2  }\Bigl[\frac{1}{u_{\infty 4}}-(v_{\infty 4}+\overline{\theta}-\theta_0)(v_{\infty 4}+\overline{\theta})    \Bigr]\vspace{.2cm}\\
E &=&\theta_1( v_{\infty 4}+\overline{\theta})+\frac{x}{x-1}(v_{\infty 4}+\overline{\theta}-\theta_0)(v_{\infty 4}+\overline{\theta})[u_{\infty 4}v_{\infty 4}(v_{\infty 4}-\theta+\overline{\theta})-1]  -\vspace{.2cm}\\&&-\theta \overline{\theta}-\frac{1}{(x-1)u_{\infty 4}}[u_{\infty 4}v_{\infty 4}(v_{\infty 4}-\theta+\overline{\theta})-1] 
 \end{array}\right.
\vspace{.5cm}\\
\left\{\begin{array}{rcl}
\dot u_{\infty 4}&=&-u_{\infty 4}(2v_{\infty 4}+\theta_1-\theta_\infty)+ x\frac{u_{\infty 4}}{1-x}\Bigl[    \overline{\theta}(\overline{\theta}-\theta_0) (2v_{\infty 4}-\theta_\infty)u_{\infty 4}\Bigr. +\vspace{.5cm}\\
&&\left.\quad \quad \quad \quad +(\theta_x+\theta_1-1)(3u_{\infty 4}v_{\infty 4}^2-2\theta_\infty u_{\infty 4}v_{\infty 4} -1)\right. 
+\vspace{.5cm}\\
&&\Bigl.\quad \quad \quad \quad +2(2v_{\infty 4}-\theta_\infty)(u_{\infty 4}v_{\infty 4}^2-\theta_\infty u_{\infty 4}v_{\infty 4}-1)
\Bigr]\vspace{.2cm}\\
\dot v_{\infty 4} &=&\frac{1}{(1-x)u_{\infty 4}}-\frac{x}{1-x}(v_{\infty 4}+\overline{\theta})(v_{\infty 4}+\overline{\theta}-\theta_0)(v_{\infty 4}-\theta_\infty)u_{\infty 4}v_{\infty 4}
 \end{array}\right.
\end{array}$$

\section{Estimates near $\mathcal{D}_\infty^{**}\setminus \mathcal{H}^*$ and $\mathcal{D}_1^{*}\setminus \mathcal{H}^*$}\label{app:EstD1Dinfty}
\begin{lemma}[Behaviour near $\mathcal{D}_1^{*}\setminus \mathcal{H}^*$]\label{lemma:nearD1}
If a solution at a complex time $t$ is sufficiently close to $\mathcal{D}_1^{*}\setminus \mathcal{H}^*$, then there exists unique $\tau\in\mathbf{C}$ such that $(u(\tau),v(\tau))$ belongs to line $\mathcal{E}_1$.
In other words, the pair $(u(t),v(t))$ has a pole at $t=\tau$.

Moreover $|t-\tau|=O(|d(t)||u_{13}(t)|)$ for sufficiently small $d(t)$ and bounded $u_{13}$.

For large $R_1>0$, consider the set $\{t\in\mathbb{C}\mid |u_{13}(t)|\le R_1\}$.
Its connected component containing $\tau$ is an approximate disk $\Delta_1$ with centre $\tau$ and radius $|d(\tau)|R_1$, and
$t\mapsto u_{13}(t)$ is a complex analytic diffeomorphism from that approximate disk onto $\{u\in\mathbb{C}\mid|u|\le R_1\}$.
\end{lemma}
\begin{proof}
For the study of the solutions near $\mathcal{D}_1^{*}\setminus \mathcal{H}^*$, we use coordinates $(u_{13},v_{13})$, see Section \ref{app:1313}.
In this chart, the set $\mathcal{D}_1^{*}\setminus \mathcal{H}^*$  is given by $\{v_{13}=0\}$ and parametrized by $u_{13}\in\mathbb{C}$.
Moreover, $\mathcal{E}_1$ is given by $u_{13}=0$ and parametrized by $v_{13}$.

Asymptotically, for $v_{13}\to0$, bounded $u_{13}$, and $x=e^t$ bounded away from $0$ and $1$, we have:
\begin{subequations}
\begin{align}
\label{u13dot}&\dot u_{13}\sim -\frac{1}{v_{13}}\\
\label{v13dot}&\dot v_{13}\sim -2u_{13}v_{13}\frac{2-x}{1-x}-(3\theta_1-\theta-\bar{\theta}-\theta_0)v_{13}-\frac{\theta_0}{1-x}v_{13}\\
\label{omega13}&\omega_{13}=-v_{13}\\
\label{omega13dot}&\frac{\dot\omega_{13}}{\omega_{13}}= -(2u_{13}+\theta_1)
 - \frac{1}{1-x}\left(2u_{13} +2\theta_1-\theta-\overline{\theta}\right) -\frac{x\theta_0}{1-x}+O(\omega_{13}) \\
\label{e13omega13}&E\omega_{13}\sim 1
\end{align}
\end{subequations}

Integrating \eqref{omega13dot} from $\tau$ to $t$, we get:
$$
\omega_{13}(t)=\omega_{13}(\tau)e^{-\theta_1(t-\tau)}e^{K(t-\tau)}(1+o(1)),
$$
with 
$$
K=-2u_{13}(\tilde\tau)
-\frac{1}{1-e^{\tilde\tau}}\left(2u_{13}(\tilde\tau)+2\theta_1-\theta-\overline{\theta}\right) 
-\frac{e^{\tilde\tau}\theta_0}{1-e^{\tilde\tau}},
$$
and $\tilde\tau$ being on the integration path.

Arguments similar to those in the proof of Lemma \ref{lemma:nearDx} show that $v_{13}$ is approximately equal to a small constant, from \eqref{u13dot} follows that: 
$$
u_{13}\sim u_{13}(\tau)-\frac{t-\tau}{v_{13}}.
$$
Thus, if $t$ runs over an approximate disk $\Delta$ centred at $\tau$ with radius $|v_{13}|R$, then $u_{13}$ fills an approximate disk centred at $u_{13}(\tau)$ with radius $R$.
Therefore, if $|v_{13}|\ll|\tau|$, the solution has the following properties for $t\in\Delta$:
$$
\frac{v_{13}(t)}{v_{13}(\tau)}\sim1,
$$
and $u_{13}$ is a complex analytic diffeomorphism from $\Delta$ onto an approximate disk with centre $u_{13}(\tau)$ and radius $R$.
If $R$ is sufficiently large, we will have $0\in u_{13}(\Delta)$, i.e.~the solution of the Painlev\'e equation will have a pole at a unique point in $\Delta$.

Now, it is possible to take $\tau$ to be the pole point.
For $|t-\tau|\ll|\tau|$, we have:
$$
\frac{d(t)}{d(\tau)}\sim1,
\quad\text{i.e.}\quad
\frac{v_{13}(t)}{d(\tau)}\sim-\frac{\omega_{13}(t)}{d(\tau)}\sim-1,
\quad
u_{13}(t)\sim -\frac{t-\tau}{v_{13}}\sim\frac{t-\tau}{d(\tau)}.
$$
Let $R_1$ be a large positive real number.
Then the equation $|u_{13}(t)|=R_1$ corresponds to $|t-\tau|\sim|d(\tau)|R_1$, which is still small compared to $|\tau|$ if $|d(\tau)|$ is sufficiently small.
Denote by $\Delta_1$ the connected component of the set of all $t\in\mathbb{C}$ such that $\{t\mid |u_{13}(t)|\le R_1\}$ is an approximate disk with centre $\tau$ and radius $2|d(\tau)|R_1$.
More precisely, $u_{13}$ is a complex analytic diffeomorphism from $\Delta_1$ onto $\{u\in\mathbb{C}\mid |u|\le R_1\}$, and
$$
\frac{d(t)}{d(\tau)}\sim1
\quad\text{for all}\quad
t\in\Delta_1.
$$
From \eqref{e13omega13}, $E(t)\omega_{13}(t)\sim1$ for the annular disk $\Delta_1\setminus\Delta_1'$, where $\Delta_1'$ is a disk centred at $\tau$ with small radius compared to radius of $\Delta_1$.
\end{proof}

 \begin{lemma}[Behaviour near $\mathcal{D}_\infty^{**}\setminus \mathcal{H}^*$]\label{lemma:nearDinf}
If a solution at a complex time $t$ is sufficiently close to $\mathcal{D}_\infty^{**}\setminus \mathcal{H}^*$, then there exists unique $\tau\in\mathbf{C}$ such that $(u(t),v(t))$ has a pole at $t=\tau$.
Moreover $|t-\tau|=O(|d(t)||v_{\infty4}(t)|)$ for sufficiently small $d(t)$ and bounded $v_{\infty4}$.

For large $R_{\infty}>0$, consider the set $\{t\in\mathbb{C}\mid |v_{\infty 4}|\le R_\infty\}$.
Its connected component containing $\tau$ is an approximate disk $\Delta_\infty$ with centre $\tau$ and radius $|d(\tau)|R_\infty$, and
$t\mapsto v_{\infty4}(t)$ is a complex analytic diffeomorphism from that approximate disk onto $\{u\in\mathbb{C}\mid|u|\le R_\infty\}$.
\end{lemma}
\begin{proof}
For the study of the solutions near $\mathcal{D}_\infty^{**}\setminus \mathcal{H}^*$, we use coordinates $(u_{\infty4},v_{\infty4})$, see Section \ref{app:inf4inf4}.
In this chart, the set $\mathcal{D}_\infty^{**}\setminus \mathcal{H}^*$  is given by $\{u_{\infty4}=0\}$ and parametrized by $v_{\infty4}\in\mathbb{C}$.
Moreover, $\mathcal{E}_{\infty}$ is given by $\{v_{\infty4}=\theta_{\infty}\}$ and parametrised by $u_{\infty4}$, while $\mathcal{E}_{\infty}^-$ is given by $\{v_{\infty4}=0\}$ and also parametrised by $u_{\infty4}$.

Asymptotically, for $u_{\infty4}\to0$,  $v_{\infty4}$ bounded,  and $x=e^t$ bounded away from $0$ and $1$, we have:
\begin{subequations}
\begin{align}
\label{uinf4dot}
&\dot{u}_{\infty 4} \sim \frac{(2v_{\infty 4}-\theta_\infty)(x+1)+\theta_1+x(\theta_x-1)}{x-1}\cdot u_{\infty 4},
\\
\label{vinf4dot}
&\dot{v}_{\infty 4} \sim  -\frac{1}{(x-1)u_{\infty 4}},
\\
\label{omegainf4}
&\omega_{\infty4}= - u_{\infty4}\\
\label{omegainf4dot}
&\frac{\dot\omega_{\infty4}}{\omega_{\infty4}}\sim2v_{\infty4}-\theta_{\infty}+\theta_x-1+\frac{4v_{\infty 4}-2\theta_\infty+\theta_1+\theta_x-1}{x-1} \\
\label{einf4omegainf4}
&E\omega_{\infty4}\sim -\frac{1}{x-1}.
\end{align}
\end{subequations}

Integrating \eqref{omegainf4dot} from $\tau$ to $t$, we get:
$$
\omega_{\infty4}(t)=\omega_{\infty4}(\tau)e^{(\theta_x-\theta_{\infty}-1)(t-\tau)}e^{K(t-\tau)}(1+o(1)),
$$
with 
$$
K=2v_{\infty_4}(\tilde\tau)+\frac{4v_{\infty 4}(\tilde\tau)-2\theta_\infty+\theta_1+\theta_x-1}{e^{\tilde\tau}-1}
$$
and $\tilde\tau$ being on the integration path.

Arguments similar to those in the proof of Lemma \ref{lemma:nearDx} show that $u_{\infty4}$ is approximately equal to a small constant, and from \eqref{vinf4dot} follows that: 
$$
v_{\infty4}\sim v_{\infty4}(\tau)+\frac{t-\tau-\log\frac{1-e^t}{1-e^{\tau}}}{u_{\infty4}}.
$$
Thus, for large $R$, if $t$ runs over an approximate disk $\Delta$ centred at $\tau$ with radius $|u_{\infty4}|R$, then $v_{\infty4}$ fills an approximate disk centred at $v_{\infty4}(\tau)$ with radius $R$.
Therefore, if $|u_{\infty4}|\ll|\tau|$, the solution has the following properties for $t\in\Delta$:
$$
\frac{u_{\infty4}(t)}{u_{\infty4}(\tau)}\sim1,
$$
and $v_{\infty4}$ is a complex analytic diffeomorphism from $\Delta$ onto an approximate disk with centre $v_{\infty4}(\tau)$ and radius $R$.
If $R$ is sufficiently large, we will have $0\in v_{\infty4}(\Delta)$, i.e.~the solution of the Painlev\'e equation will have a pole at a unique point in $\Delta$.

Now, it is possible to take $\tau$ to be the pole point.
For $|t-\tau|\ll|\tau|$, we have that $d(t)\sim d(\tau)$ implies:
$$
1\sim
\frac{(1-x)\omega_{\infty4}(t)}{d(\tau)}
\sim
\frac{(e^t-1) u_{\infty4}(t)}{d(\tau)},
\quad
v_{\infty4}(t)\sim \frac{t-\tau}{u_{\infty4}}\sim\frac{(t-\tau)(e^t-1)}{d(\tau)}.
$$
Let $R_{\infty}$ be a large positive real number.
Then the equation $|v_{\infty4}(t)|=R_{\infty}$ corresponds to $|(e^t-1)(t-\tau)|\sim|d(\tau)|R_{\infty}$, which is still small compared to $|\tau|$ if $|d(\tau)|$ is sufficiently small.
Denote by $\Delta_{\infty}$ the connected component of the set of all $t\in\mathbb{C}$ such that $\{t\mid |v_{\infty4}(t)|\le R_{\infty}\}$ is an approximate disk with centre $\tau$ and radius $2|d(\tau)|R_{\infty}$.
More precisely, $v_{\infty4}$ is a complex analytic diffeomorphism from $\Delta_{\infty}$ onto $\{v\in\mathbb{C}\mid |v|\le R_{\infty}\}$, and
$$
\frac{d(t)}{d(\tau)}\sim1
\quad\text{for all}\quad
t\in\Delta_{\infty}.
$$
From \eqref{einf4omegainf4}, $E(t)\omega_{\infty4}(t)\sim1/(1-e^t)$ for the annular disk $\Delta_{\infty}\setminus\Delta_{\infty}'$, 
where $\Delta_{\infty}'$ is a disk centred at $\tau$ with small radius compared to radius of $\Delta_{\infty}$.
\end{proof}

 \section{The vector field in the limit space}\label{s:Details0B}
 $ 
  \left\{\begin{array}{ccl}
E_0&=&-u \{ (u-1)v\left(uv-2\theta+\theta_\infty\right)-\theta_1v  + \theta (\theta-\theta_\infty)  \}  \vspace{.2cm}\\
\dot{u} &=& - 2u^2 (u -1)v  +(\theta+\overline{\theta} )u (u -1)+\theta_1u \vspace{.2cm}\\
\dot{v} &=& (3u-2)uv^2-2(\theta+\overline{\theta})u v +\left(  \theta+\overline{\theta}   -\theta_1 \right)v +\theta \overline{\theta} \, . 
\end{array}\right.$\vspace{.4cm}\\
 $
\begin{array}{l}\left\{\begin{array}{rcl}
     (\tilde u_{1} , \tilde v_{1})&=&\left(u,\frac{1}{u^2 (u-1)v} \right)\vspace{.2cm}\\
 E_0&=&-  \frac{1}{\tilde u_{1}^2 (\tilde u_{1}-1)\tilde v_{1}^2} +\frac{1} {\tilde v_{1}}\left(\frac{\theta+\overline{\theta}-\theta_1}{\tilde{u}_1} +\frac{\theta_1}{\tilde{u}_1-1} \right)  - \theta \overline{\theta}\tilde{u}_1  
   \vspace{.2cm}\\
\dot{\tilde{u}}_1 &=& -\frac{ 2}{ \tilde v_{1}} +  \tilde{u}_1^2 (\tilde{u}_1 -1)  \left(  \frac{\theta+\overline{\theta}-\theta_1}{\tilde{u}_1 }+\frac{\theta_1}{\tilde{u}_1 -1} \right)\vspace{.2cm}\\
\dot{\tilde v}_1 &=& \frac{  1}{ \tilde u_1}-\frac{\theta_1\tilde v_1-1}{\tilde u_1-1} -\left((\theta+\overline{\theta})(\tilde u_1-1) +2  \theta_1 \right)\tilde v_1 -  \theta\overline{\theta} \tilde u_1^2(\tilde u_1-1) \tilde{v}_1^2\, . 
  \end{array}\right.\end{array}$
  \vspace{.4cm}\\
  $\begin{array}{l}\left\{\begin{array}{rcl}
     (\tilde u_{2} , \tilde v_{2})&=&\left(\frac{1}{u},\frac{1}{ (u-1)v} \right)\vspace{.2cm}\\
 E_0&=& -\frac{(\theta \tilde v_{2}-1)\left(\overline{\theta}\tilde v_{2}-1\right)}{ \tilde u_{2}\tilde v_{2}^2}+ \frac{ \theta_1 \tilde v_{2}-1 }{ (1-\tilde u_{2})\tilde v_{2}^2}   \vspace{.2cm}\\
 \dot{\tilde u}_2&=& \frac{2 }{\tilde v_2} +\theta_1 -    (1-{\tilde u}_2)  \left( \theta+\overline{\theta}-\theta_1  \right)\vspace{.2cm}\\
\dot{\tilde v}_2&=& -\frac{(\theta \tilde v_2-1)(\overline{\theta}\tilde v_2-1)}{ \tilde u_2} -\frac{ \theta_1\tilde v_2-1}{ (1-\tilde u_2)}-(\theta+\overline{\theta}-\theta_1)\tilde v_2+\theta \overline{\theta} \tilde v_2^2\, .
  \end{array}\right.\end{array}$
    \vspace{.4cm}\\
   $\begin{array}{l}\left\{\begin{array}{rcl}
     (\tilde u_{3} , \tilde v_{3})&=&\left(\frac{1}{u}, -u(uv-\theta)  \right)\vspace{.2cm}\\
 E_0&=&  -(1-\tilde u_{3})\tilde u_{3}^2\tilde v_{3}-(\theta_\infty+\theta_1 ) \tilde u_{3}\tilde v_{3}  + \theta_\infty  \tilde v_{3}- \theta(\overline{\theta}-\theta_1 )
  \vspace{.2cm}\\
\dot{\tilde u}_3&=&  -   2\tilde{u}_3(1-\tilde{u}_3)\tilde{v}_3+  {\theta}_\infty(1-\tilde{u}_3) - \theta_1 \tilde{u}_3   \vspace{.2cm}\\ 
\dot{\tilde v}_3&=& ( -2\tilde{u}_3+1) \tilde{v}_3^2 -(\theta_1+\theta_\infty) \tilde{v}_3 \, .
  \end{array}\right.\end{array}$  \vspace{.4cm}\\
$\left\{\begin{array}{rcl}
     (u_{03} , v_{03})&=&\left(uv-\theta_0, \frac{1}{uv^2-\theta_0 v} \right)\vspace{.2cm}\\
(u,v)&=&\left(u_{03}^2v_{03}+ \theta_0u_{03}v_{03} , \frac{1}{u_{03}v_{03}   }\right)\vspace{.2cm}\\
 E_0&=&\Bigl[- (u_{03}+\theta_0)u_{03} v_{03}\left(u_{03}+\theta_0  -(\theta+\overline{\theta}-\theta_0) \right)  - u_{03}(\theta\overline{\theta}v_{03}-1)    -(\theta_x-1)\Bigr] (u_{03}+\theta_0)
 \vspace{.2cm}\\
 \dot u_{03} &=&  (u_{03}+\theta_0-\theta)(u_{03}+\theta_0-\overline{\theta}) (u_{03}+\theta_0)u_{03}v_{03}
  \vspace{.2cm}\\
\dot v_{03} &=&-  (u_{03}+\theta_0-\theta)(u_{03}+\theta_0-\overline{\theta}) (2u_{03}+\theta_0) v_{03}^2
 -(2(u_{03}+\theta_0)-(\theta+\overline{\theta}))(u_{03}+\theta_0)u_{03}v_{03}^2 +\vspace{.2cm}\\
&&+(2u_{03}+\theta_0-\theta_x+1  )v_{03}   \end{array}\right.$
 \vspace{.4cm}\\
 $\left\{\begin{array}{rcl}
   (u_{04} , v_{04})&=& \left(uv^2-\theta_0 v  , \frac{1}{v} \right)\vspace{.2cm}\\
E_0&=&\Bigl[u_{04}v_{04} -(u_{04}v_{04}+\theta_0-(\theta+\overline{\theta}))(u_{04}v_{04}+\theta_0)v_{04}   -\\&&\hspace{4.5cm}\quad \quad\quad\quad \quad \quad\quad \quad \Bigl.   -\theta\overline{\theta} v_{04}- (\theta_x-1)\Bigr] (u_{04}v_{04}+\theta_0)
 \vspace{.2cm}\\
\dot u_{04}&=& (u_{04}v_{04}+\theta_0-\theta)(u_{04}v_{04}+\theta_0-\overline{\theta}) (2u_{04}v_{04}+\theta_0)  
  -\theta_1u_{04}   -\vspace{.2cm}\\
&&-[(\theta+\overline{\theta})-2(u_{04}v_{04}+\theta_0) ][(u_{04}v_{04}+\theta_0)v_{04}-1]u_{04}  \vspace{.2cm}\\

\dot v_{04} &=&- \Bigl[ \left(2(u_{04}v_{04}+\theta_0)-\theta\right)v_{04}-1\Bigr]\Bigl[\left(2(u_{04}v_{04}+\theta_0)-\overline{\theta}\right)v_{04}-1\Bigr]+
\vspace{.2cm}\\
&&+ \left[(u_{04}v_{04}+\theta_0)v_{04}-1\right]^2+\theta_1v_{04} 
 \end{array}\right.$
  \vspace{.4cm}\\
   $ \left\{\begin{array}{rcl}
 (u_{x3} , v_{x3})&=&(u_{x2},v_{x2}/u_{x2})=\left(u^2v  , \frac{1}{u^3v^2} \right)\vspace{.2cm}\\
   E_0&=& \frac{1-   (\theta+\overline{\theta}-\theta_1) u_{x3} v_{x3} +u_{x3}^2v_{x3} }{(u_{x3}v_{x3})^2}+\theta\overline{\theta}u_{x3}^2v_{x3}-(\theta+\overline{\theta})u_{x3} \vspace{.2cm}\\
\dot u_{x3} &=&   - (u_{x3} +\theta+\overline{\theta}-\theta_1)u_{x3}  +\frac{2}{ v_{x3}}+
 \theta\overline{\theta}  u_{x3}^4v_{x3}^2 
\vspace{.2cm}\\
 \dot v_{x3} &=& -2\theta\overline{\theta}u_{x3}^3 v_{x3}^3+(\theta+\overline{\theta})u_{x3}^2v_{x3}^2 -\frac{2 }{ u_{x3} }+ (\theta+\overline{\theta}-\theta_1)v_{x3} 
 
\end{array}\right. $ \vspace{.4cm}\\
 $ \left\{\begin{array}{rcl}
 (u_{x4} , v_{x4})&=&(u_{x2}/v_{x2},v_{x2})= \left(u^3v^2 , \frac{1}{uv} \right)\vspace{.2cm}\\
   E_0&=&  \frac{1-(\theta+\overline{\theta}-\theta_1)v_{x4}}{ v_{x4}^2 }+(\theta v_{x4}-1)(\overline{\theta}v_{x4}-1)u_{x4}  
  \vspace{.2cm}\\
\dot u_{x4} &=&  \theta\overline{\theta}u_{x4}^2v_{x4}^3 -(\theta+\overline{\theta})u_{x4}^2v_{x4}^2+\frac{2u_{x4} }{  v_{x4}} -(\theta+\overline{\theta}-\theta_1)u_{x4}
 \vspace{.2cm}\\
\dot v_{x4} &=&-u_{x4}v_{x4}^2(\theta v_{x4}-1)(\overline{\theta}v_{x4}-1) 
 \end{array}\right. $
    \vspace{.4cm}\\
 $ \left\{\begin{array}{rcl}
 (u_{13},v_{13})&=& \left((u-1)v-\theta_1 , \frac{1}{(u-1)v^2-\theta_1v} \right)\vspace{.2cm}\\
 E_0 &=&- (u_{13}+\theta_1-\theta )(u_{13}+\theta_1-\overline{\theta}  ) ((u_{13}+\theta_1)u_{13}v_{13}+1) - (u_{13}+\theta_1)u_{13}  -\frac{ 1}{v_{13} }  \vspace{.2cm}\\
\dot u_{13} &=&   -\frac{1}{v_{13}}+ \left(u_{13} +\theta_1-\theta\right)\left(u_{13} +\theta_1-\overline{\theta}\right)(u_{13}+\theta_1) u_{13}v_{13}\vspace{.2cm}\\
\dot v_{13} &=&-(2u_{13}+\theta_1)v_{13}-  \left(u_{13} +\theta_1-\theta\right)\left(u_{13} +\theta_1-\overline{\theta}\right)(2u_{13}+\theta_1)v_{13}^2-\vspace{.2cm}\\
&&\quad - \left(2u_{13} +2\theta_1-\theta-\overline{\theta}\right)( u_{13}^2v_{13}+\theta_1 u_{13}v_{13} +1)v_{13} 
 \end{array}\right. $  \vspace{.4cm}\\
 $ \left\{\begin{array}{rcl}
 (u_{14},v_{14})&=& \left((u-1)v^2-\theta_1v , \frac{1}{ v} \right)\vspace{.2cm}\\
 E_0 &=&- \Bigl[(u_{14}v_{14}+\theta_1-\theta )(u_{14}v_{14}+\theta_1-\overline{\theta}  )+u_{14}\Bigr] ((u_{14}v_{14}+\theta_1) v_{14}+1)  \vspace{.2cm}\\
\dot u_{14} &=&  (2u_{14}v_{14}+\theta_1)\Bigl[ \left(u_{14}v_{14} +\theta_1-\theta\right)\left(u_{14}v_{14} +\theta_1-\overline{\theta}\right)+u_{14}\Bigr]+\vspace{.2cm}\\
&&\quad +2 (u_{14}v_{14}^2+\theta_1  v_{14}+1)u_{14}\left(u_{14}v_{14} +\theta_1-\frac{\theta+\overline{\theta}}{2} \right) 
\vspace{.2cm}\\
\dot v_{14} &=&- v_{14}^2 \left(u_{14}v_{14} +\theta_1-\theta\right)\left(u_{14}v_{14} +\theta_1-\overline{\theta}\right)-1-2\left(u_{14}v_{14} -\frac{\theta_0-\theta_1}{2}\right)v_{14}-
\vspace{.2cm}\\&&\quad -  \theta_0 v_{14} -2 (u_{14}v_{14}^2+\theta_1  v_{14}+1)v_{14}\left(u_{14}v_{14}+\theta_1-\frac{\theta+\overline{\theta}}{2}  \right)
 
 \end{array}\right.$
    \vspace{.4cm}\\$
\left\{\begin{array}{rcl}
 (u_{\infty 3} , v_{\infty 3})&=&\left(\frac{1}{u^2v-\theta u},(uv-\theta)(uv-\overline{\theta})u\right)\vspace{.2cm}\\
 E_0 &=&(\theta_1  -\theta) \overline{\theta}+ u_{\infty 3}v_{\infty 3}(\theta_1-\theta_\infty)+(u_{\infty 3}v_{\infty 3})^2-v_{\infty 3}\vspace{.2cm}\\
 \dot u_{\infty 3}&=& \left(\theta_\infty -\theta_1\right)u_{\infty 3}- 2u_{\infty 3}^2v_{\infty 3}  +1 
\vspace{.2cm}\\
\dot v_{\infty 3} &=&v_{\infty 3}(2u_{\infty 3}v_{\infty 3}+\theta_1-\theta_\infty) 
   \end{array}\right.$
  \vspace{.4cm}\\
  $\left\{\begin{array}{rcl}
 (u_{\infty 4} , v_{\infty 4})&=&\left(\frac{1}{(uv-\theta)(uv-\overline{\theta})u  }, uv-\overline{\theta} \right)\vspace{.2cm}\\
 E_0 &=&(\theta_1  -\theta )\overline{\theta}+ v_{\infty 4}(v_{\infty 4}-\theta_\infty+\theta_1)-\frac{1}{ u_{\infty 4}} \vspace{.2cm}\\
\dot u_{\infty 4}&=&-u_{\infty 4}(2v_{\infty 4}+\theta_1-\theta_\infty) \vspace{.2cm}\\
\dot v_{\infty 4} &=&\frac{1}{u_{\infty 4}} 
 \end{array}\right.
$

\bibliographystyle{amsplain}

\begin{thebibliography}{99}

\bibitem{beauville}  A. Beauville, Complex algebraic surfaces. London Mathematical Society Student Texts, volume {\bf 34}. Cambridge University Press, Cambridge, 1996. x+132 pp.

\bibitem{Boutroux} P. Boutroux, Recherche sur les transcendants de M. Painlev\'e et l'\'etude asymptotique des \'equations diff\'erentielles du second ordre  {\it Ann. Sci. Ec. Norm. Super. } \textbf{30}
(1913), 255-375.

\bibitem{DubrovinMazzocco2000}
  B. Dubrovin and M. Mazzocco, Monodromy of certain Painlev{\'e}--VI transcendents and reflection groups,
\textit{Inventiones mathematicae},
  \textbf{141} (2000),
55--147.

\bibitem{DuistermaatBOOK}
J. J. Duistermaat, 
Discrete integrable systems: QRT maps and elliptic surfaces,
Springer Monographs in Mathematics,
Springer,
New York,
2010.

\bibitem{Nalini1} J. Duistermaat and N. Joshi, Okamoto's space for the first Painlev\'e equation in Boutroux coordinates  {\it Arch. Ration. Mech. Anal.} \textbf{202 (3)}
(2011), 707--785.

\bibitem{fuchs1884}
  L. Fuchs, Uber differentialgleichungen deren int\'egrale feste verzweigungspunkte besitzen, \textit{Sitz. Akad. Wiss. Berlin.}, \textbf{32} (1884), 669--720.
  

\bibitem{Fuchs05} R. Fuchs, Sur quelques {\'e}quations diff{\'e}rentielles lin{\'e}aires du second ordre,
\textit{CR Acad. Sci. Paris},
\textbf{141} (1905), 555--558.



\bibitem{GIL2012}
O. Gamayun, N. Iorgov, and O. Lisovyy,
Conformal field theory of Painlev\'e VI,
{\it Journal of High Energy Physics}, 
\textbf{10} (2012), 1-25.

\bibitem{Gambier06} B. Gambier, Sur les equations differentielles du deuxi\'eme ordre et du premier degr\'e dont l'int\'egrale g\'en\'erale est uniforme,
\textit{Comptes Rendus de l'Acad emie des Sciences, Paris},
\textbf{142} (1906), 1497--1500.


\bibitem{Gerard1975}	
R. G\'erard,
Geometric theory of differential equations in the complex domain,
Complex analysis and its applications,
Lectures, Internat. Sem., Trieste, 1975,
Internat. Atomic Energy Agency, Vienna,
(1976),	269--308.

\bibitem{Gerard1983}
R. G\'erard,
La g\'eom\'etrie des transcendantes de P. Painlev\'e,
Mathematics and physics,
Paris,
1979/1982,
Progr. Math. 37,
Birkh\"auser Boston, Boston, MA
(1983), 323--352.

\bibitem{GerardSec1972}
R. G\'erard,
A. Sec,
Feuilletages de Painlev\'e,
Bull. Soc. Math. France
100 (1972), 47--72.

\bibitem{GrifHarPRINC}
P. Griffiths, J. Harris, 
Principles of algebraic geometry,
Wiley-Interscience,
New York,
1978.




\bibitem{HartshorneAG}
R. Hartshorne, 
Algebraic geometry,
Graduate Texts in Mathematics, No. 52,
Springer-Verlag, New York-Heidelberg,
1977.


\bibitem{hirzebruch} F. Hirzebruch,
\"Uber eine Klasse von einfachzusammenh\"angenden komplexen Mannigfaltigkeiten. (German)
\textit{Math. Ann.} \textbf{124} (1951), 77--86.

\bibitem{hitchin95} N. Hitchin, Twistor spaces, Einstein metrics and isomonodromic deformations,
 \textit{Journal of Differential Geometry},
  \textbf{42} (1995), 30--112.




\bibitem{Nalini2} P. Howes and N. Joshi, Global Asymptotics of the second Painlev\'e equation in Okamoto's space, \textit{Constr. Approx.} \textbf{39} (2014), no. 1, 11--41.

\bibitem{NaliniPhD} N. Joshi, The connection problem for the first and second Painlev\'e transcendents, \textit{Thesis (Ph.D.) Princeton University.} (1987), 246pp.


\bibitem{NaliniMilena1}
  N. Joshi and M. Radnovi\'c,
  Asymptotic behavior of the fourth Painlev\'e transcendents in the space of initial values, \textit{Constr. Approx.} \textbf{44} (2016), no. 2, 195--231.

\bibitem{NaliniMilena2}
  N. Joshi and M. Radnovi\'c,
  Asymptotic behavior of the fifth Painlev\'e transcendents in the space of initial values, \textit{Proceedings of the London Mathematical Society}, \textbf{116} (2018), no. 6, 1329--1364.

\bibitem{NaliniMilena3}
  N. Joshi and M. Radnovi\'c,
  Asymptotic behavior of the third Painlev\'e transcendents in the space of initial values, \textit{Transactions of the American Mathematical Society}, \textbf{372} (2019), no. 9, 6507-6546.
\bibitem{kruskal} M. D. Kruskal, Asymptotology, \textit{Proceedings of Conference ''Mathematical Models in Physical Sciences'' at the University of Notre Dame} (1963), 17--48.


\bibitem{LisovyyTykhyy2014} O. Lisovyy and Y. Tykhyy, Algebraic solutions of the sixth Painlev\'e equation, \textit{Journal of Geometry and Physics}, \textbf{85} (2014) 124 -- 63.
  
\bibitem{Malgrange01} B. Malgrange, {\it Le group\"oide de Galois d’un feuilletage}, Monographie \textbf{38} vol \textbf{2} de L’enseignement math\'ematique (2001).


\bibitem{manin96} Yu. I. Manin, Sixth {P}ainlev\'e equation, universal elliptic curve, and mirror of ${\mathbb P}^2$, in \textit{Geometry of Differential Equations}, ed. A.G. Khovanski{\u\i},  A.N.  Varchenko, and V.A.   Vassil'ev, \textbf{186} (1998) 131--151.

  \bibitem{milnor}
  J. Milnor,
  Foliations and foliated vector bundles, in \textit{Collected papers of John Milnor}, ed. J. McCleary, vol. {IV} (2009), American Mathematical Society, Providence, Rhode Island, 279--320.
  
\bibitem{DLMF} NIST Digital Library of Mathematical Functions. \texttt{http://dlmf.nist.gov/}, Release 1.1.6 of 2022-06-30. F. W. J. Olver, A. B. Olde Daalhuis, D. W. Lozier, B. I. Schneider, R. F. Boisvert, C. W. Clark, B. R. Miller, B. V. Saunders, H. S. Cohl, and M. A. McClain, eds.
  


\bibitem{Okamoto} K. Okamoto, Sur les feuilletages associ\'es aux \'equations du second ordre \`a points critiques fixes de P. Painlev\'e. \textit{J Japan. J. Math. (N.S.)} \textbf{5} (1979), no. 1, 1--79.

\bibitem{OkamotoStudies}  K. Okamoto, Studies on the Painlev\'e equations. I. Sixth Painlev\'e equation PVI. Ann. Mat. Pura Appl. (4) 146 (1987), 337--381.


\bibitem{OkamotoHamil} K. Okamoto, Polynomial Hamiltonians associated with Painlev\'e equations I. \textit{Proc. Japan Acad.} \textbf{56} (1980), Ser. A, 264--268.

\bibitem{Painleve1897}
P. Painlev\'e, 
Le\c{c}ons sur la th\'eorie analytique des \'equations diff\'erentielles profess\'ees \`a Stockholm, A. Hermann, 1897, Paris.



\bibitem{Sakai} H. Sakai, Rational surfaces associated with affine root systems and geometry of the Painlev\'e equations. \textit{Comm. Math. Phys.} \textbf{220} (2001), 165-229.


 \bibitem{Takano} T. Shioda  and K. Takano, On some Hamiltonian structures of Painlev{\'e} systems, I \textit{Funkkcialaj Ekvacioj Serio Internacia}, Japana Matematika Societo  \textbf{40} (1997), 271--292.

\bibitem{Umemura90} H. Umemura, Birational automorphism groups and differential equations, Nagoya Math. J. \textbf{119} (1990), 1--80.

\bibitem{Umemura96} H. Umemura, Galois theory of algebraic and differential equation, Nagoya Math. J. \textbf{144} (1996) 1--58.

   \bibitem{Watanabe99} H. Watanabe, Birational canonical transformations and classical solutions of the sixth Painlev\'e
equation, \textit{Ann. Scuola Norm. Sup. Pisa Cl. Sci.}, \textbf{27} (1999), 379 -- 425.

\end{thebibliography}

\end{document}